\documentclass[12pt,a4paper]{article}
\usepackage{jheppub}
\usepackage{amsmath}
\usepackage{amsthm}
\usepackage{amsfonts}
\usepackage{amssymb}
\usepackage{tensor,epigraph}
\usepackage{hyperref}
\usepackage{cancel}
\usepackage[normalem]{ulem} 
\usepackage{musixtex}
\usepackage{booktabs}
\usepackage{empheq}
\usepackage{comment}
\usepackage{psfrag}
\usepackage[table]{xcolor}
\usepackage{floatrow}

\usepackage{caption}
\usepackage{subcaption}
\usepackage{graphicx}
\usepackage{color}
\usepackage{scrextend}
\usepackage{float}
\usepackage{multicol}
\usepackage{mathrsfs}

\usepackage{wrapfig,lipsum,booktabs}
 \usepackage{upgreek}
 \usepackage{braket}
 \usepackage{enumitem}
 \usepackage{verbatim}
 \usepackage{newtxtext,newtxmath}
 \usepackage{cancel}
 \usepackage{physics}
 \usepackage{tcolorbox}
 \usepackage{url}
 \usepackage{nicefrac}       %
 \usepackage{microtype}      %

\usepackage[utf8]{inputenc} %
\usepackage[T1]{fontenc}    %
\usepackage{url}            %

\newcommand{\n}{\nabla}

\newcommand{\D}{\text{d}}
\DeclareMathOperator\arctanh{arctanh}

\definecolor{klgreen}{rgb}{0.0, 0.5, 0.0}

\definecolor{color1}{rgb}{1,0,0}
\definecolor{color2}{rgb}{0,0,1}
\definecolor{color3}{rgb}{0,1,0}
\definecolor{color4}{rgb}{0.6,0.4,0.2}
\definecolor{color5}{rgb}{1,0,1}
\definecolor{color6}{rgb}{1,0.5,0}
\definecolor{color7}{rgb}{0.666667,0.666667,1}
\definecolor{color8}{rgb}{0.5,0,0.5}
\definecolor{color9}{rgb}{0.666667,0.666667,0}
\definecolor{green2}{rgb}{0.0, 0.5, 0.0}

\newcommand{\exclude}[1]{}

\newcommand{\beq}{\begin{equation}}
\newcommand{\eeq}{\end{equation}}
\newcommand{\bea}{\begin{eqnarray}}
\newcommand{\eea}{\end{eqnarray}}
\long\def\/*#1*/{}

\newtheorem{theorem}{Thm.}[section]
\newtheorem{corollary}{Cor.}[section]
\newtheorem{definition}{Def.}[section]

\usepackage{tikz}

\newcommand{\junk}[1]{}

\renewcommand{\Im}{\mathrm{Im}}
\renewcommand{\Re}{\mathrm{Re}}

\title{\Large Pseudospectra of Holographic Quasinormal Modes}
\author[1,2]{Daniel Are\'an,}
\author[1,2]{David García Fariña,}
\author[1]{Karl Landsteiner}
\affiliation[1]{Instituto de F\'isica Te\'orica UAM/CSIC, Calle Nicol\'as Cabrera 13-15, 28049 Madrid, Spain}
\affiliation[2]{Departamento de F\'isica Te\'orica, Universidad Aut{\'o}noma de Madrid, Campus de Cantoblanco, 28049 Madrid, Spain}

\preprint{IFT-UAM/CSIC-23-091}

\emailAdd{daniel.arean@uam.es}
\emailAdd{david.garciafarinna@estudiante.uam.es}
\emailAdd{karl.landsteiner@csic.es}

\abstract{

Quasinormal modes and frequencies are the eigenvectors and eigenvalues of a non-Hermitian differential operator. 
They hold crucial significance in the physics of black holes. The analysis of quasinormal modes of black holes in asymptotically Anti-de Sitter geometries plays also a key role in the study of strongly coupled quantum many-body systems via gauge/gravity duality.
In contrast to normal Sturm-Liouville operators, the spectrum of non-Hermitian (and non-normal) operators generally is unstable under small perturbations. This research focuses on the stability analysis of the spectrum of quasinormal frequencies pertaining to asymptotically planar AdS black holes, employing pseudospectrum analysis. Specifically, we concentrate on the pseudospectra of scalar and transverse gauge fields, shedding light on their relevance within the framework of gauge/gravity duality.

}

\begin{document}

\maketitle
\section{Introduction}

Quasinormal modes are the solutions of linear differential equations on an open domain. The openness is encoded in outgoing boundary conditions. Such boundary conditions arise naturally in many branches of physics from the study of black holes to optics. Quasinormal modes can be thought of as the eigenmodes of linear differential operators. Due to the open boundary conditions the differential operator is non-Hermitian and the eigenvalues are complex.

In gravitational physics, quasinormal modes are of utmost importance in the theory of black holes. Famously, a black hole horizon acts as a perfectly absorbing membrane. Therefore, solving a wave equation in a black hole background leads naturally to outgoing boundary conditions. In asymptotically flat space, radiation is also outgoing at infinity. 
The quasinormal modes that arise in asymptotically flat spacetimes are supposed to describe the late time ringdown of perturbed black holes. They appear for example in black hole mergers and are of highest interest for gravitational wave physics \cite{Nollert:1999ji,Kokkotas:1999bd,Berti:2009kk,Franchini:2023eda,Jaramillo:2022kuv}.

The importance of quasinormal modes in asymptotically Anti de-Sitter (AdS) spaces derives from gauge/gravity duality \cite{Maldacena:1997re,Aharony:1999ti,BigJ,zaanen2015holographic,Hartnoll:2018xxg}. Gauge/gravity duality is conjectured to describe certain strongly coupled quantum many-body systems. While up to date no concrete physical system has been found that indeed is modelled in all detail by gauge/gravity duality, it has lead to important insights into areas such as hydrodynamics and transport theory in the relativistic regime \cite{Policastro:2001yc, Baier:2007ix,Bhattacharyya:2007vjd}. Some key results are the extremely low specific shear viscosity of holographic models of the quark gluon plasma \cite{Kovtun:2004de} or phase transitions towards superconducting states \cite{Gubser:2008px,Hartnoll:2008vx} as well as strongly coupled quantum critical phases \cite{Herzog:2007ij}. Quasinormal modes are a key ingredient in this development. They can be understood as the poles of the retarded Green's functions of the strongly coupled quantum many-body system \cite{Horowitz:1999jd,Birmingham:2001pj,Kovtun:2005ev}. The shear viscosity can be read off from a quasinormal mode corresponding to momentum diffusion \cite{Policastro:2002se1} and second order phase transitions arise as a quasinormal frequency moves into the upper half plane thus indicating an instability of the ground state \cite{Herzog:2009ci,Amado:2009ts}.

In particular the modern way of understanding relativistic hydrodynamics is very much influenced by the properties of quasinormal modes in asymptotically AdS black holes. Relativistic hydrodynamics can be understood as an effective field theory organised in a derivative expansion \cite{Kovtun:2012rj}. As an effective field theory it has a limited range of validity.
This limit can be set as the wavelength at which the first non-hydrodynamic mode decays at a slower rate than the hydrodynamic mode \cite{Amado:2008ji}. More recently
it is conjectured that the limit is set by the radius in the complex plane at which the lowest hydrodynamic quasinormal mode collides with the first non-hydrodynamic mode \cite{Withers:2018srf,Grozdanov:2019,GrozdanovKovtun:2019}.

A question that so far has not been investigated in the context of gauge/gravity duality is the spectral stability of the quasinormal frequencies under small perturbations of the background geometry. We will refer to this spectral instability simply as instability. Since in this work we are not considering eigenvalues lying on the upper half of the complex plane (which would denote instability of the background), no confusion should arise from this wording.

In asymptotically flat space studies on the stability of quasinormal frequencies have a long history \cite{Nollert:1996rf,Nollert:1998ys,Aguirregabiria:1996} and recently the issue has been taken up in \cite{Jaramillo:2020tuu,Destounis:2021lum,Sarkar:2023rhp,Cheung:2021bol,Berti:2022xfj,Konoplya:2022pbc,Courty:2023rxk,Daghigh:2020jyk,Qian:2020cnz,alsheikh:tel-04116011,Destounis:2023ruj,Torres:2023nqg} (see also \cite{Boyanov:2022ark} for horizonless compact objects). The key new development of these recent studies is the application of the method of pseudospectra  \cite{Trefethen:2005,Sjostrand:2019,davies:2007}. 

The eigenvalues of Hermitian operators are stable in a technical sense. In fact, this stability is what makes perturbation theory successful in quantum physics.
On the other hand, the eigenvalues of non-Hermitian and non-normal operators generically are unstable. This means they can be displaced large distances by even small perturbations of the original operator. Here stability and smallness of a perturbation have precise meanings and we will give exact definitions later on.

In the case of quasinormal modes the spectral instability appears as a consequence of the non-conservative nature of the system. This non-conservative nature is associated with the damping of the fluctuations as they fall into the black hole and radiate away to infinity in asymptotically flat spacetime. AdS at large distances from the black hole horizon acts as a confining box but it is still true that fluctuations are damped as they fall into the black hole. 

Thus the damping of fluctuations in black hole geometries is independent of the asymptotics of the spacetime. It is then reasonable to expect that such quasinormal modes and frequencies are also inherently unstable as occurs in asymptotically flat space. The motivation for investigating  this question is twofold. Firstly, it would serve as increasing evidence of the independence of the instability on the asymptotic behavior of the spacetime. And secondly, it would shed a completely new light on the behavior of strongly coupled quantum many-body systems modelled by gauge/gravity duality. The instability of the quasinormal mode spectrum would indicate that also the excitation spectrum of the dual quantum field theories would be unstable under perturbations. In turn, this might raise questions about the robustness of conclusions for quantum many-body physics drawn from the behavior of quasinormal modes. In particular, transport properties in these systems could potentially be heavily model dependent, with a small deformation of the original theory leading to very different transport coefficients. Let us add here that despite the theoretically well established instability of the quasinormal mode spectrum in the asymptotically flat case, extraction of overtones from gravitational wave signals have been reported in \cite{Isi:2019aib,Giesler:2019uxc,Capano:2021etf,Capano:2022zqm}. 

To probe the stability of the quasinormal mode frequency (QNF) we follow the example of \cite{Jaramillo:2020tuu} and turn to pseudospectrum analysis \cite{Trefethen:2005}. Broadly speaking, pseudospectrum provides insight on how much perturbations of a given size can displace the spectrum of an operator. Accordingly, once the physically motivated notion of size is defined, pseudospectra serve as a powerful tool to discuss spectral stability or lack thereof. 
    
In this work we initiate a systematic study of the stability of quasinormal frequencies in AdS. We focus on two very simple cases: a real scalar field and a transverse gauge field in a Schwarzschild AdS$_{4+1}$ black brane (SAdS$_{4+1}$). 
By gauge/gravity duality these correspond to a scalar operator and a conserved $R$-current of the strongly coupled maximally symmetric $\mathcal{N}=4$ Yang-Mills theory \cite{Maldacena:1997re}. 
Our choice of fluctuations is motivated by two main reasons. Firstly, we consider the SAdS$_{4+1}$ background because it is the archetypical example for gauge/gravity duality. Secondly, we choose a real scalar and a transverse gauge field for their simplicity and their lack of a hydro mode. Hydro modes dominate the low energy spectrum of the dual QFT as their corresponding QNF goes to zero as the momentum goes to zero. Consequently, analyzing the stability of hydro modes is akin to studying the stability of the hydrodynamic description of the dual quantum field theory \cite{Kovtun:2005ev}. 
As the study of hydrodynamic modes involves additional technical complexities due to the presence of gauge symmetries and the corresponding constraints, we will leave this topic for a follow-up work.  

The organisation of this paper is as follows. Since the method of pseudospectra is relatively new and perhaps little known in the context of gauge/gravity duality we spend the first few sections reviewing some of the underlying mathematics. Here we draw heavily from \cite{Trefethen:2005} and also from a previous work on asymptotically flat space \cite{Jaramillo:2020tuu}. In particular we define the pseudospectrum and review its most important properties in section \ref{section:Pseudospectrum and Stability}.

In section \ref{section:Quasinormal Modes and Pseudospectra} we focus specifically on the application of pseudospectra to quasinormal modes. We point out the importance of choosing the proper coordinates. Following \cite{Warnick:2013hba} we chose a coordinate slice which we call {\em regular}. These coordinates have the important property that the equal-time slices are spacelike outside the horizon but coincide with infalling Eddington-Finkelstein coordinates at the horizon.\footnote{We note that infalling coordinates are very well suited to impose outgoing boundary conditions.} We discuss the importance of the choice of norm and give a physical reasoning why the chosen energy norm is the relevant one in the holographic context. Indeed, the energy in regular coordinates is not conserved due to outflow on the horizon. We conclude this section by discussing the holographic interpretation of the instability of the quasinormal modes.

In section \ref{section:Holographic Model} we construct the specific differential operators that describe the fluctuations of scalar and transverse gauge fields in an AdS black hole background. We show explicitly that the non-Hermiticity is concentrated on the horizon. As we discuss in detail this leads to a clear picture of the local origin of the non-Hermiticity and non-normality.  

In section \ref{section:Numerical method} we first outline our numerical methods which are based on pseudospectral methods \cite{Trefethen:2000,boyd}.
Next, we introduce the selective pseudospectrum which tests the stability under random local potential perturbations. And finally, to further explore the aforementioned stability, we construct specific potentials to test some relevant regimes.

Section \ref{section:Results} is the core of the paper and contains the results of our pseudospectrum calculations. We mostly concentrate on the results for the energy norm but, to illustrate the norm dependence, we also present the result for the $L^2$ norm in appendix \ref{app:l2norm}. Indeed it turns out that the pseudospectrum for both the scalar field and the transverse vector show the typical open contour lines indicating instability. 

In section \ref{section:Conclusions} we summarize our findings, present our conclusions and give an outlook on further studies in particular the highly interesting case of hydrodynamic modes.

Some of the more technical details on the calculation of the adjoint differential operators and the numerical methods are presented in the appendices 
\ref{app:Computation of Ldagger}, \ref{app:Extra Details on the Discretization.} and \ref{app:algorithm}. The numerical values of the quasinormal frequencies are presented in appendix \ref{app: Numerical Values of the QNFs}.

 \section{Pseudospectrum and Stability}\label{section:Pseudospectrum and Stability}

In this section we collect relevant theorems and definitions of the pseudospectrum of linear operators. We next specialize them to finite dimensional operators (matrices) and illustrate some key notions through the example of a $2\times 2$ matrix. For additional details, we refer the reader to \cite{Trefethen:2005}.

    Given a closed linear operator $L$ acting on a Hilbert space $\mathcal{H}$ with domain $\mathcal{D}(L)$, its spectrum  $\sigma\left(L\right)$ is defined as the set of points in the complex plane, $z$, where its resolvent:\footnote{Recall that a closed linear operator $L$ is a linear operator such that if a sequence $\{x_n\}$ in $\mathcal{H}$ converges to $x$, then the sequence $\{Lx_n\}$ converges to $y$, where $Lx=y$. From now on, when talking about operators, we implicitly assume that they are both closed and linear.}
    \begin{eqnarray}
        \mathcal{R}(z;L)=(L-z)^{-1}\,,
    \end{eqnarray}
    does not exist as a bounded operator on $\mathcal{H}$.
    
    Within $\sigma\left(L\right)$, we define the eigenvalues $\{\lambda_i\}$ through the eigenvalue equation:
    \begin{equation}
        L u_i = \lambda_i u_i\,,
    \end{equation}
    where $u_i\in \mathcal{D}(L)$ is the eigenvector corresponding to the eigenvalue $\lambda_i$.
Note that with our definitions, the eigenvalues correspond to isolated points in the spectrum (regions of dimension 0). Nonetheless, it is important to stress that, in general, the spectrum also contains regions of higher dimension, such as branch cuts (regions of dimension 1). That said, \cite{Warnick:2013hba} proved that in the case of Schwarzschild AdS black holes the spectrum of quasinormal frequencies indeed consists of a countable discrete subset of the complex plane.

    A particularly important property of eigenvalues is that as long as the operator is self-adjoint (or, in more physical terms, the system is conservative), the spectral theorem ensures that if we perturb the system with a bounded operator of size $\varepsilon$ the eigenvalues of the perturbed operator cannot suffer a displacement greater than $\varepsilon$ \cite{kato2013perturbation,Courant:1989}.\footnote{The notion of size here is quite nontrivial and we shall discuss it in greater detail later. In general, we will consider that the size is given by the norm of the operator, which is inherited from the function norm in $\mathcal{H}$.} More generally, this property holds for any normal operator $A$ satisfying $\left[A,A^\dagger\right]=0$, with $A^\dagger$ the adjoint of $A$. Physically, this ensures that any statement made relying on eigenvalue analysis is guaranteed to be right up to corrections of the same size as the error implicitly made when assuming a particular model. 
    
    However, many physical systems are inherently non-conservative. In these systems, normality of the relevant operators is not guaranteed and small perturbations could potentially alter the spectrum in a significant manner. As any description of a physical system is always a model and consequently has some inherent error associated with it, 
    one can conclude~\cite{Trefethen:2005} that eigenvalue analysis is insufficient when dealing with non-conservative systems defined through non-self-adjoint (and potentially non-normal) operators as eigenvalues may not be stable.
    
    In order to characterize the stability of eigenvalues we need to introduce the notion of $\varepsilon$-pseudospectrum, which can be defined in three mathematically equivalent ways \cite{Trefethen:2005}:
    \begin{definition}[Resolvent norm approach]
        \label{def:Pseudo Definition 1}
        Given a closed linear operator $L$ acting on a Hilbert space $\mathcal{H}$ with domain $\mathcal{D}(L)$, and $\varepsilon>0$, the $\varepsilon$-pseudospectrum $\sigma_\varepsilon(L)$ is
        \begin{equation}\label{eq:Pseudo Definition 1}
            \sigma_\varepsilon(L)=\{z\in \mathbb{C} : \lVert \mathcal{R}(z;L) \rVert>1/\varepsilon \}\,,
        \end{equation}
        with the convention $\lVert \mathcal{R}(z;L) \rVert=\infty$ for $z\in\sigma(L)$.
    \end{definition}

    \begin{definition}[Perturbative approach]
        \label{def:Pseudo Definition 2}
        Given a closed linear operator $L$ acting on a Hilbert space $\mathcal{H}$ with domain $\mathcal{D}(L)$, and $\varepsilon>0$, the $\varepsilon$-pseudospectrum $\sigma_\varepsilon(L)$ is
        \begin{equation}\label{eq:Pseudo Definition 2}
            \sigma_\varepsilon(L)=\{z\in \mathbb{C}, \exists V,\rVert V \lVert<\varepsilon : z\in\sigma(L+V) \}\,.
        \end{equation}
    \end{definition}
    
    \begin{definition}[Pseudoeigenvalue approach]
        \label{def:Pseudo Definition 3}
        Given a closed linear operator $L$ acting on a Hilbert space $\mathcal{H}$ with domain $\mathcal{D}(L)$, and $\varepsilon>0$, the $\varepsilon$-pseudospectrum $\sigma_\varepsilon(L)$ is
        \begin{equation}\label{eq:Pseudo Definition 3}
            \sigma_\varepsilon(L)=%
            \{z\in \mathbb{C},\exists u^\varepsilon \in \mathcal{D}(L) :\; \rVert (L-z ) u^\varepsilon\lVert<\varepsilon\lVert u^\varepsilon \rVert\}\,,
        \end{equation}
        where $u^\varepsilon$ is a $\varepsilon$-pseudoeigenvector with $\varepsilon$-pseudoeigenvalue $z$.
    \end{definition}
    Note that, contrary to the spectrum, the pseudospectrum depends on the operator norm as it needs a notion of what constitutes a large or small perturbation to quantify stability.
    
    \begin{definition}[Operator norm]
        \label{def:Operator norm Hilbert space}
        Given a bounded linear operator $V$ acting on a Hilbert space $\mathcal{H}$ equipped with a function norm $\lVert\cdot\rVert$, we define its norm $\lVert V \rVert$ as:
        \begin{equation}\label{eq:Operator norm Hilbert space}
            \norm{V}= \underset{u\in \mathcal{H}}{\max}\frac{\norm{Vu}}{\norm{u}}\,.
        \end{equation} 
    \end{definition}
    
    Definition \ref{def:Pseudo Definition 2} corresponds to the physical intuition we were seeking: the $\varepsilon$-pseudospectrum constitutes the maximal region containing all possible displacements of the eigenvalues under perturbations of size $\varepsilon$. It is then quite natural to represent the pseudospectrum as a contour map indicating the boundaries of these regions for multiple values of $\varepsilon$ (see \textit{e.g.} figure~\ref{fig:PseudospectrumExample} for one such a map corresponding to a $2\times2$ matrix).

    On the other hand, definition \ref{def:Pseudo Definition 1}, despite lacking such a clear physical interpretation, is significantly more powerful as it allows to establish a few relevant theorems whose proofs can be found in chapter 4 of \cite{Trefethen:2005}:
    \begin{theorem}
        Given a closed linear operator $L$ acting on a Hilbert space $\mathcal{H}$ with domain $\mathcal{D}(L)$, the $\varepsilon$-pseudospectrum $\sigma_\varepsilon(L)$ is a nonempty open set of $\,\mathbb{C}$, and any bounded connected component of $\sigma_\varepsilon(L)$ has a nonempty intersection with $\sigma(L)$.
    \end{theorem}

    \begin{theorem}
        \label{th:Boundedness of pseudospectrum}
        Given a closed linear operator $L$ acting on a Hilbert space $\mathcal{H}$ with domain $\mathcal{D}(L)$, the pseudospectra are strictly nested supersets of the spectrum: $\cap_{\varepsilon>0}\,\sigma_\varepsilon(L)=\sigma(L)$, and conversely, for any $\delta>0$, $\sigma_{\varepsilon+\delta}(L)\supseteq\sigma_{\varepsilon}(L)+\Delta_\delta$ with $\Delta_\delta$ the open disk of radius $\delta$.
    \end{theorem}

    \begin{theorem}
        \label{th:Pseudospectrum conjugation properties}
        Given a closed linear operator $L$ acting on a Hilbert space $\mathcal{H}$ with domain $\mathcal{D}(L)$, $z\in\mathbb{C}$ and $\varepsilon>0$, we have $\norm{\mathcal{R}(\overline{z};L^\dagger)}=\norm{\mathcal{R}(z;L)}$, $\sigma(L^\dagger)=\overline{\sigma(L)}$ and $\sigma_\varepsilon(L^\dagger)=\overline{\sigma_\varepsilon(L)}$.\footnote{Regarding notation, we denote the complex conjugate with a bar and complex conjugate transpose with an asterisk.}
    \end{theorem}
    
    Definition \ref{def:Pseudo Definition 3} hints at the numerical difficulties arising when obtaining eigenvalues and eigenvectors of non-normal operators. When operating at a precision of $\varepsilon$ we are unable to distinguish between an $\varepsilon$-pseudoeigenvalue and a genuine eigenvalue, which, by virtue of definition \ref{def:Pseudo Definition 2}, implies that we cannot obtain the spectrum of the desired operator but instead that of a perturbed one. Thus, in the absence of normality, one needs to be especially careful when employing numerical methods to prevent loss of predictivity. One practical consequence is that even numerical errors can lead to displacements and therefore it is often necessary in numerical approximations to compute quasinormal frequencies with (much) higher precision than the typical double precision floating-point machine numbers.

    Besides pseudospectra, the condition numbers $\{\kappa_i\}$ are another useful tool to study the stability of eigenvalues. They quantify the effect of perturbations of size $\varepsilon$ through the knowledge of the orthogonality between the eigenvectors of the unperturbed operator and its adjoint (often referred to as right- and left-eigenvectors, respectively). Heuristically, the condition numbers exploit the lack of orthogonality between left- and right-eigenvectors in non-normal operators to look for potential instabilities. This is formalized in the following definition and theorem \cite{Trefethen:2005}:
    \begin{definition}[Condition numbers]
        \label{def:Condition numbers}
        Given a closed linear operator $L$ acting on a Hilbert space $\mathcal{H}$ and its adjoint $L^\dagger$, we define the condition number $\kappa_i$ associated with the eigenvalue $\lambda_i$ of $L$ as:
        \begin{equation}
            \kappa_i=\frac{\norm{v_i}\norm{u_i}}{|\expval{v_i,u_i}|}\,,
        \end{equation}
        where $\expval{\cdot,\cdot}$ is the inner product associated with the norm $\norm{\cdot}$, $u_i$ is the right-eigenvector satisfying $L u_i=\lambda_i u_i$ and $v_i$ is the left-eigenvector satisfying $L^\dagger v_i=\overline{\lambda_i} v_i$
    \end{definition}
    
    \begin{theorem}[Stability and condition numbers]
        \label{th:Stability and condition numbers}
        Given a closed linear operator $L$ acting on a Hilbert space $\mathcal{H}$, whose spectrum contains a set of eigenvalues $\{\lambda_i\}$, and a bounded perturbation $V$ with $\norm{V}=\varepsilon$ we have:
        \begin{equation}
            |\lambda_i(\varepsilon)-\lambda_i|\leq \varepsilon  \kappa_i \,,
        \end{equation}
        where $\{\lambda_i(\varepsilon)\}$ are the eigenvalues of the perturbed operator $L(\varepsilon)=L+V$
    \end{theorem}
    \begin{proof}
        Let us consider a closed operator $L$, its adjoint $L^\dagger$ and $u_i$, $v_i$ such that:
        \begin{equation}
            Lu_i=\lambda_i u_i\,,\qquad L^\dagger v_i=\overline{\lambda_i} v_i \,.
        \end{equation}
    
        Defining the perturbed operator $L(\varepsilon)=L+V$ with $\norm{V}=\varepsilon$, the perturbed right-eigenvector $u_i(\varepsilon)=u_i+\varepsilon\delta u_i+\mathcal{O}\left(\varepsilon^2\right)$ satisfies:
        \begin{equation}
            L(\varepsilon)u_i(\varepsilon)=\lambda_i(\varepsilon) u_i(\varepsilon)\,,
        \end{equation}
        with $\lambda_i(\varepsilon)$ the corresponding perturbed eigenvalue.      
        Then, keeping only the leading order in $\varepsilon$ we have:
        \begin{align}
            \lambda_i(\varepsilon)&=\frac{\expval{v_i,L(\varepsilon)u_i(\varepsilon)}}{\expval{v_i,u_i(\varepsilon)}}=\frac{\expval{v_i,Lu_i}+\expval{v_i,V u_i}+\varepsilon\expval{v_i,L\delta u_i}}{\expval{v_i,u_i}+\varepsilon\expval{v_i,\delta u_i}}\nonumber\\
        &=\lambda_i+\varepsilon\frac{\expval{L^\dagger v_i,\delta u_i}}{\expval{v_i,u_i}}-\varepsilon\lambda_i\frac{\expval{v_i,\delta u_i}}{\expval{v_i,u_i}}+\frac{\expval{ v_i,V u_i}}{\expval{v_i,u_i}} \nonumber\\
        &=\lambda_i+\varepsilon\lambda_i\frac{\expval{ v_i,\delta u_i}}{\expval{v_i,u_i}}-\varepsilon\lambda_i\frac{\expval{v_i,\delta u_i}}{\expval{v_i,u_i}}+\frac{\expval{ v_i,V u_i}}{\expval{v_i,u_i}}=\lambda_i+\frac{\expval{ v_i,V u_i}}{\expval{v_i,u_i}}\,,
        \end{align}
        from where we recover the result we wanted to prove:
        \begin{equation}
            |\lambda_i(\varepsilon)-\lambda_i|=\left|\frac{\expval{ v_i,V u_i}}{\expval{v_i,u_i}}\right|\leq \left|\frac{\norm{v_i}\norm{Vu_i}}{\expval{v_i,u_i}}\right|\leq \left|\frac{\norm{v_i}\norm{V}\norm{u_i}}{\expval{v_i,u_i}}\right|=\varepsilon \kappa_i\,.
        \end{equation}
    \end{proof}
    
    For a normal operator all eigenvalues are stable and have condition number 1.\footnote{For a normal operator $L$ and its adjoint $L^\dagger$ can be simultaneously diagonalized and thus $v_i=u_i$ and $\expval{v_i,u_i}=\norm{v_i}\norm{u_i}$.} Interestingly enough, even for non-normal operators some eigenvalues may have condition number 1 and thus be stable; these are referred to as normal eigenvalues. 

    Lastly, we note that throughout this section we have assumed $\varepsilon$ to be small, but we have yet to give a formal description of smallness. As indicated by definition \ref{def:Pseudo Definition 2}, pseudospectrum analysis is, in spirit, equivalent to a problem of perturbation theory in Quantum Mechanics and thus we decide to inherit the latter's definition of smallness:
    
    \begin{definition}[Smallness]
        Given a closed linear operator of size $\varepsilon$, we say that $\varepsilon$ is small if it is negligible when compared to the minimum distance $d_{\min}$ between disconnected regions of the spectrum, 
        $\varepsilon/d_{\min}\ll 1$.
    \end{definition}

	\subsection{Matrix Pseudospectrum}\label{subsection:Matrix Pseudospectrum}
    Throughout this work, we will limit ourselves to studying pseudospectra of matrices arising from the discretization of differential operators. Consequently, in this subsection, we present some relevant results specialized to matrices. Note that all previous theorems and definitions still hold since $N\times N$ matrices are bounded linear operators acting on an $N$ dimensional Hilbert space comprised of vectors with $N$ components. Again further details on the following theorems and definitions can be found on \cite{Trefethen:2005}.
    
    We begin by constructing a more practical definition of matrix norm using definition \ref{def:Operator norm Hilbert space}:
    \begin{theorem}[Matrix norm]
        \label{th:Matrix norm}
        Given an $N$-dimensional Hilbert space $\mathcal{V}$ equipped with a norm $\norm{\cdot}$ and an $N\times N$ matrix $M$; the norm of $M$ is:
        \begin{equation}
            \norm{M}=\sqrt{\rho\left(M^\dagger M\right)}\,,
        \end{equation}
        with $\rho$ the spectral radius:
        \begin{equation}
            \rho(A)=\underset{\lambda\in \sigma\left(A\right)}{\max}\left\{|\lambda|\right\}\,.
        \end{equation}
    \end{theorem}
    \begin{proof}
        We begin by pointing out that, for matrices, all elements in the spectrum are eigenvalues. Then, by definition \ref{def:Operator norm Hilbert space}, we have:
        \begin{equation}
            \norm{M}= \underset{u\in \mathcal{V}}{\max}\frac{\norm{Mu}}{\norm{u}}=\underset{u\in \mathcal{V}}{\max}\frac{\sqrt{\expval{Mu,Mu}}}{\norm{u}}=\underset{u\in \mathcal{V}}{\max}\frac{\sqrt{\expval{u,M^\dagger Mu}}}{\norm{u}}=\underset{u\in \mathcal{V}}{\max}\frac{\sqrt{\expval{M^\dagger Mu,u}}}{\norm{u}}\,,
        \end{equation}  
        where $\expval{\cdot,\cdot}$ is the inner product associated with the norm and $M^\dagger$ the adjoint with respect to the aforementioned inner product.
        
        As $M^\dagger M$ is non-negative definite and self-adjoint, we can find an orthonormal basis $\{e_i\}$ satisfying $M^\dagger M e_i = \lambda_i e_i$ with $\lambda_i$ real positive eigenvalues. Choosing $\lambda_n=\max\{\lambda_i\}$, we then have:
        \begin{equation}
            \expval{u,M^\dagger Mu}=\sum_{i,j}\overline{u^i}u^j \expval{e_i,M^\dagger M e_j}=\sum_{i}\overline{u^i}u^i \lambda_i\leq \sum_{i}\lambda_n\,\overline{u^i}u^i = \lambda_n \norm{u}^2=|\lambda_n| \norm{u}^2\,,
        \end{equation}
        which implies:
         \begin{equation}
            \norm{M}=\underset{u\in \mathcal{V}}{\max}\frac{\sqrt{\expval{u,M^\dagger Mu}}}{\norm{u}}=\sqrt{|\lambda_n|}=\sqrt{\underset{\lambda\in \sigma\left(M^\dagger M\right)}{\max}\left\{|\lambda|\right\}}=\sqrt{\rho\left(M^\dagger M\right)}\,.
         \end{equation}
    \end{proof}
    Using this expression of the matrix norm, we can rewrite definition \ref{def:Pseudo Definition 1} in terms of the minimum generalized singular value. This statement is formalized in the following theorem:
    \begin{theorem}
        \label{th:Matrix pseudospectrum}
        Given an $N\times N$ matrix $M$ acting on an $N$-dimensional Hilbert space $\mathcal{V}$ equipped with a norm $\norm{\cdot}$, the $\varepsilon$-pseudospectrum $\sigma_{\varepsilon}(M)$ is
        \begin{equation}
            \sigma_{\varepsilon}(M)=\{z\in \mathbb{C}: s_{\min}(M-z)<\varepsilon\}\,,
        \end{equation}
        where $s_{\min}$ is the smallest generalized singular value:
        \begin{equation}
            s_{\min}(A)=\sqrt{\underset{\lambda\in\sigma(A^\dagger A)}{\min} {|\lambda|}}\,.
        \end{equation}
    \end{theorem}
    \begin{proof}
        First, we note that given an invertible matrix $A$, we have:
        \begin{equation}
            \underset{\lambda\in\sigma(A^{-1})}{\max} \{|\lambda|\}=\left(\underset{\lambda\in\sigma(A)}{\min} \{|\lambda|\}\right)^{-1}\,,
        \end{equation} 
        which follows trivially from the diagonalization of $A$. Then, the norm of the inverse can be written as follows:
        \begin{equation}\label{eq:Proof Matrix resolvent, aux eq 1}
            \norm{A^{-1}}=\left(\underset{\lambda\in\sigma(A)}{\min} \{|\lambda|\}\right)^{-1/2}=\left[s_{\min}(A)\right]^{-1}\,.
        \end{equation}
        Note that we can consistently define the norm of the inverse of a non-invertible matrix to be $\infty$, thereby extending this proof to arbitrary matrices. Now, specializing to the resolvent, we have:
        \begin{equation}
            \norm{\mathcal{R}(z;M)}=\left[s_{\min}(M-z)\right]^{-1} \Longrightarrow \norm{\mathcal{R}(z;M)}<1/\varepsilon \Leftrightarrow s_{\min}(M-z)>\varepsilon\,,
        \end{equation}
        which, applying definition \ref{def:Pseudo Definition 1}, implies that the theorem holds.
    \end{proof}
    
    Having expressed the pseudospectrum in terms of the generalized singular values, one can now establish the following corollary of theorem \ref{th:Boundedness of pseudospectrum}:
    \begin{corollary}
        \label{cor:Boundedness of pseudospectrum}
        Given an $N\times N$ normal matrix $M$ acting on an $N$-dimensional Hilbert space $\mathcal{V}$ equipped with a norm $\norm{\cdot}$, its $\varepsilon$-pseudospectrum $\sigma_{\varepsilon}(M)$ is
        \begin{equation}
            \sigma_{\varepsilon}(M)=\sigma(M)+\Delta_\varepsilon\,,
        \end{equation}
        with $\Delta_\varepsilon$ the open disk of radius $\varepsilon$.
    \end{corollary}
    \begin{proof}
        Let us consider a normal matrix $M$ and its adjoint $M^\dagger$. As $M$ and $M^\dagger$ commute, they can be simultaneously diagonalized. Furthermore, theorem \ref{th:Pseudospectrum conjugation properties} implies that the eigenvalues of $M$ and $M^\dagger$ are complex conjugate to each other, \textit{i.e.}, $\lambda_i^\dagger=\overline{\lambda_i}$ for $\lambda_i\in\sigma(M)$ and $\lambda_i^\dagger\in\sigma(M^\dagger)$.  
    Then, the eigenvalues of $(M-z)^\dagger(M-z)$, $\left\{\tilde{\lambda}\right\}$ are given by:
        \begin{equation}
            \left\{\tilde{\lambda}\right\}=\sigma\left(M^\dagger M-zM^\dagger+\overline{z}M+|z|^2\right)=\left\{|\lambda|^2-z\overline{\lambda}-\overline{z}\lambda+|z|^2\right\}=\left\{|z-\lambda|^2\right\}\,,
        \end{equation}
        where $\left\{\lambda\right\}=\sigma(M)$; and, consequently, $s_{\min}(M-z)$ is:
        \begin{equation}
            s_{\min}(M-z)=|z-\lambda_i|\,,
        \end{equation}
        for some $\lambda_i\in\left\{\lambda\right\}$.
        
        Lastly, recalling theorem \ref{th:Matrix pseudospectrum}, we conclude that the $\varepsilon$-pseudospectrum is given by:
        \begin{equation}
            \sigma_\varepsilon(M-z)=\{z\in\mathbb{C},\lambda_i\in\sigma(M):|z-\lambda_i|<\varepsilon\}=\sigma(M)+\Delta_\varepsilon\,.
        \end{equation}
    \end{proof}

    Up to this point, we have been totally generic and have not assumed any particular norm. However, when working with $N\times N$ matrices it is particularly convenient to choose the $N\times N$ matrix $\ell^2$-norm ($\ell^2_N$-norm) defined as:
    \begin{equation}
        \norm{u}_2=\sqrt{\sum_{a=1}^{N} |u^a|^2}\,.
    \end{equation}

    As we will discuss later, the $\ell^2_N$ norm allows for significant optimization of the pseudospectrum algorithm, thus it is convenient to always employ it. However, we have a preferred norm, the one induced by the discretization of the original Hilbert space where the differential operator acts. Then, to exploit the advantages of the $\ell^2_N$-norm, we need to establish a connection between both norms.\footnote{Relating definitions in a generic norm to the $\ell^2_N$-norm is identical to performing a change of basis in the Hilbert space from a "coordinate" basis with nontrivial metric to an "orthogonal" basis with euclidean metric.}
    
    \begin{theorem}
        \label{th:psuedospectrum in different norms}
        Given the $\ell^2_N$-norm $\norm{\cdot}_2$ and generic $G$-norm $\norm{\cdot}_G$ such that:
        \begin{equation}
            \label{eq:G-norm in th:psuedospectrum in different norms }  
            \expval{v,u}_G= G_{ij} \overline{v^i} u^j\,,
        \end{equation}
        with $G=F^*F$ a symmetric positive definite $N\times N$ matrix.
        \begin{itemize}
            \item The $\varepsilon$-pseudospectrum of a matrix $M$ in the $G$-norm $\sigma_\varepsilon^G(M)$ satisfies
            \begin{equation}
            \label{eq:pseudo relation in th:psuedospectrum in different norms }      \sigma_\varepsilon^G(M)=\sigma_\varepsilon^{\ell^2_N}\left(FMF^{-1}\right)\,,
            \end{equation}
            with $\sigma_\varepsilon^{\ell^2_N}(M)$ the pseudospectrum in the $\ell^2_N$-norm.
            
            \item The condition number of the eigenvalue $\lambda_i$ of a matrix $M$ in the G-norm $\kappa_i^G$ satisfies:
            \begin{equation}
              \kappa_i^G=\frac{\norm{\tilde{v}_i}_2\norm{\tilde{u}_i}_2}{|\expval{\tilde{v}_i,\tilde{u}_i}_2|}\,,
            \end{equation}
            where $\tilde{u}_i$ and $\tilde{v}_i$ fulfill:
            \begin{equation}
                FMF^{-1}\tilde{u}_i=\lambda_i \tilde{u}_i\,,\qquad \left(FMF^{-1}\right)^*\tilde{v}_i=\overline{\lambda_i} \tilde{v}_i\,.
            \end{equation}
        \end{itemize}
        
    \end{theorem} 
    \begin{proof}
        First, let us assume an $N$-dimensional Hilbert space $\mathcal{V}$ equipped with the $G$-norm $\norm{\cdot}_G$ that induces the following inner product:
        \begin{equation}
            \expval{v,u}_G=v^* G u\,.
        \end{equation}
        For this to represent a well-defined inner product, the matrix $G$ has to be symmetric and positive definite; thus it can always be expressed as $G=F^*F$ using a Cholesky decomposition (for a proof see, for instance, ch. 4 of \cite{Horn:2017}).\footnote{In the language of differential geometry, performing a Cholesky decomposition is equivalent to choosing triangular vielbeins. For any metric, this can always be achieved by exploiting the internal O($N$) symmetry of the "orthogonal" coordinates.}
        \begin{itemize}
            \item Considering now an $N\times N$ matrix $A$, we have: 
            \begin{align}
                \norm{A}_G^2&=\underset{u\in \mathcal{V}}{\max} \frac{\expval{Au,Au}_G}{\expval{u,u}_G}=\underset{u\in \mathcal{V}}{\max} \frac{\left(Au\right)^*GAu }{u^*Gu}=\underset{u\in \mathcal{V}}{\max} \frac{u^* A^* F^* F A u }{u^*F^* Fu}\nonumber\\
                &=\underset{u\in \mathcal{V}}{\max} \frac{\expval{FAu,FAu}_2}{\expval{Fu, Fu}_2}=\underset{\tilde{u}\in \mathcal{V}}{\max} \frac{\expval{FAF^{-1}\tilde{u},FAF^{-1}\tilde{u}}_2}{\expval{\tilde{u}, \tilde{u}}_2}=\norm{FAF^{-1}}_2^2\,,
            \end{align}
            where we have introduced $\tilde{u}=Fu$. This relates the $G$-norm to the $\ell^2_N$ norm for any generic matrix. In particular, for the resolvent of a matrix $M$:
            \begin{align}
                \norm{\mathcal{R}(z;M)}_G&=\norm{F\mathcal{R}(z;M)F^{-1}}_2=\norm{F(M-z)^{-1}F^{-1}}_2\nonumber\\
                &=\norm{(FMF^{-1}-z)^{-1}}_2=\norm{\mathcal{R}(z;FMF^{-1})}_2\,,
            \end{align}
            and using definition \ref{def:Pseudo Definition 1} we trivially recover \eqref{eq:pseudo relation in th:psuedospectrum in different norms }.
            \item Now we assume an $N\times N$ matrix $M$ with left- and right-eigenvalues $v_i$ and $u_i$:
            \begin{equation}\label{eq:proof th psuedospectrum in different norms, def left right eigenvectors}
                M u_i=\lambda_i u_i\,,\qquad M^\dagger v_i=\overline{\lambda_i} v_i\,.
            \end{equation}
            The condition number $\kappa_i^G$ corresponding to the eigenvalue $\lambda_i$ is then given by:
            \begin{equation}
                \kappa_i^G=\frac{\sqrt{v_i^* G v_i}\sqrt{u_i^* G u_i}}{|v_i^* G u_i|}=\frac{\sqrt{(Fv_i)^* F v_i}\sqrt{(Fu_i)^*  F u_i}}{|(Fv_i)^* F u_i|}=\frac{\norm{\tilde{v}_i}_2\norm{\tilde{u}_i}_2}{|\expval{\tilde{v}_i,\tilde{u}_i}_2|}\,,
            \end{equation}
            where we have introduced $\tilde{u}_i=F u_i$ and $\tilde{v}_i=Fv_i$.

            Now, we note that $M^\dagger$ in the G norm is defined as:
            \begin{equation}
                M^\dagger=\left(GMG^{-1}\right)^*=F^{-1}\left(FMF^{-1}\right)^* F\,,
            \end{equation}
            which trivially follows from the definition of adjoint:
            \begin{equation}
                \expval{M^\dagger g,f}_G=\expval{ g,Mf}_G\Rightarrow g^*\left(M^\dagger\right)^* Gf=g^*GMf\Rightarrow \left(M^\dagger\right)^* G=GM\,,
            \end{equation}
            with $f$ and $g$ two arbitrary vectors.

            Then, we can rewrite \eqref{eq:proof th psuedospectrum in different norms, def left right eigenvectors} in terms of $\tilde{u}_i$ and $\tilde{v}_i$ as:
            \begin{align}
                &MF^{-1} \tilde{u}_i=F^{-1} \lambda_i\tilde{u}_i\Rightarrow FMF^{-1} \tilde{u}_i=\lambda_i\tilde{u}_i\,,\\
                &F^{-1}\left(FMF^{-1}\right)^* F F^{-1} \tilde{v}_i=F^{-1} \overline{\lambda_i} \tilde{v}_i\Rightarrow  \left(FM F^{-1}\right)^* \tilde{v}_i=F^{-1} \overline{\lambda_i}\tilde{v}_i\,,
            \end{align}
            thus concluding the proof.
        \end{itemize}
    \end{proof}

    Lastly, let us address a nontrivial subtlety of the discretization process. Definition \ref{def:Pseudo Definition 2} of the $\varepsilon$-pseudospectrum assumes bounded perturbations, or, in more physical terms, potential perturbations. However, in a discretized setting all operators become bounded and one cannot distinguish the discretized version of an originally bounded operator from that of an originally unbounded one. Then, in many cases it will be more interesting to probe a particular type of perturbations using definition \ref{def:Pseudo Definition 2}, either choosing a specific form of the perturbation matrix or generating random matrices satisfying some physically motivated constraints.

\subsubsection{Norm Dependence in a Simple Example}
    In order to make the formalism of this section a bit more transparent, we present the pseudospectrum of a $2\times2$ matrix $A$
    \begin{equation}\label{eq:A matrix Example}
        A=\begin{pmatrix}
            -1&0\\
            -50&-2
        \end{pmatrix}\,,
    \end{equation}
    with eigenvalues $\omega_1=-2$ and $\omega_2=-1$.
    
    Furthermore, to make the norm dependence explicit, we present the analysis in two norms, the $\ell^2_2$-norm and the $G$-norm:
    \begin{equation}
        \norm{u}_G= G_{ij} \overline{u^i} u^j\,,
    \end{equation}
    with $G$
    \begin{equation}
        G=\left(
        \begin{array}{cc}
            2\cdot10^4 & 50 \\
            50 & 1 \\
        \end{array}
        \right)=\left(
        \begin{array}{cc}
            100\sqrt{2} & 0\\[\medskipamount]
            \frac{1}{2\sqrt{2}} & \sqrt{\frac{7}{8}}
        \end{array}
        \right)
        \left(
        \begin{array}{cc}
            100\sqrt{2} & \frac{1}{2\sqrt{2}} \\[\medskipamount]
            0 & \sqrt{\frac{7}{8}}
        \end{array}
        \right)=F^* F\,.
    \end{equation}
    \begin{itemize}
        \item $\ell^2_2$-norm:
        
        In this case, the adjoint of $A$ is its complex conjugate transpose
        \begin{equation}
            A^\dagger=\begin{pmatrix}
                -1&-50\\
                0&-2
            \end{pmatrix}\,,
        \end{equation}
        and $A$ is not normal:
        \begin{equation}
            A^\dagger A-AA^\dagger=\left(
            \begin{array}{cc}
                -2500 & -50 \\
                -50 & 2500 \\
            \end{array}
            \right)\neq0\,.
        \end{equation}
        Then, one could potentially find instabilities associated with non-normality. A suspicion confirmed by the condition numbers
        \begin{equation}
            \kappa_1=\sqrt{2501}\approx 50\,, \qquad \kappa_2=\sqrt{2501}\approx 50\,,
        \end{equation}
        which differ greatly from 1, and the pseudospectrum of figure \ref{fig:PseudospectrumExampleNonNormal}, which is characterized by extended $\varepsilon$-pseudospectrum boundaries that eventually engulf both eigenvalues for small values of $\varepsilon$.
        
        \item $G$-norm:
        
        In this case, $A$ is self-adjoint:
        \begin{equation}
            A^\dagger=\overline{G^{-1}AG}=\begin{pmatrix}
                -1&0\\
                -50&-2
            \end{pmatrix}=A\,,
        \end{equation}
        and consequently its spectrum is stable.
        
        Here, the condition numbers are
        \begin{equation}
            \kappa_1=1\,, \qquad \kappa_2=1\,,
        \end{equation}
        and the pseudospectrum is characterized by concentric circles of radius $\varepsilon$ around the eigenvalues (figure \ref{fig:PseudospectrumExampleNormal}), both denoting stability.
    \end{itemize}

    As expected, the choice of norm is nontrivial; the notion of stability depends on the definition of small perturbations. It is also interesting to remark that figures \ref{fig:PseudospectrumExampleNonNormal} and \ref{fig:PseudospectrumExampleNormal} show very common features of the pseudospectrum of non-normal and normal operators, respectively. In general, the $\varepsilon$-pseudospectrum of a normal operator is a set disks of radius $\varepsilon$ centered on the eigenvalues (as indicated by corollary \ref{cor:Boundedness of pseudospectrum}); while for a non-normal operator it presents extended regions that may contain multiple eigenvalues.\footnote{Note that for normal operators the $\varepsilon$-pseudospectrum does engulf multiple eigenvalues if $\varepsilon$ is large. As expected, this behavior appears for $\varepsilon$ larger than half the distance between the two eigenvalues.} 

    \begin{figure}[h!]
        \centering
        \begin{subfigure}[b]{0.49\linewidth}
            \centering
            \includegraphics[width=\linewidth]{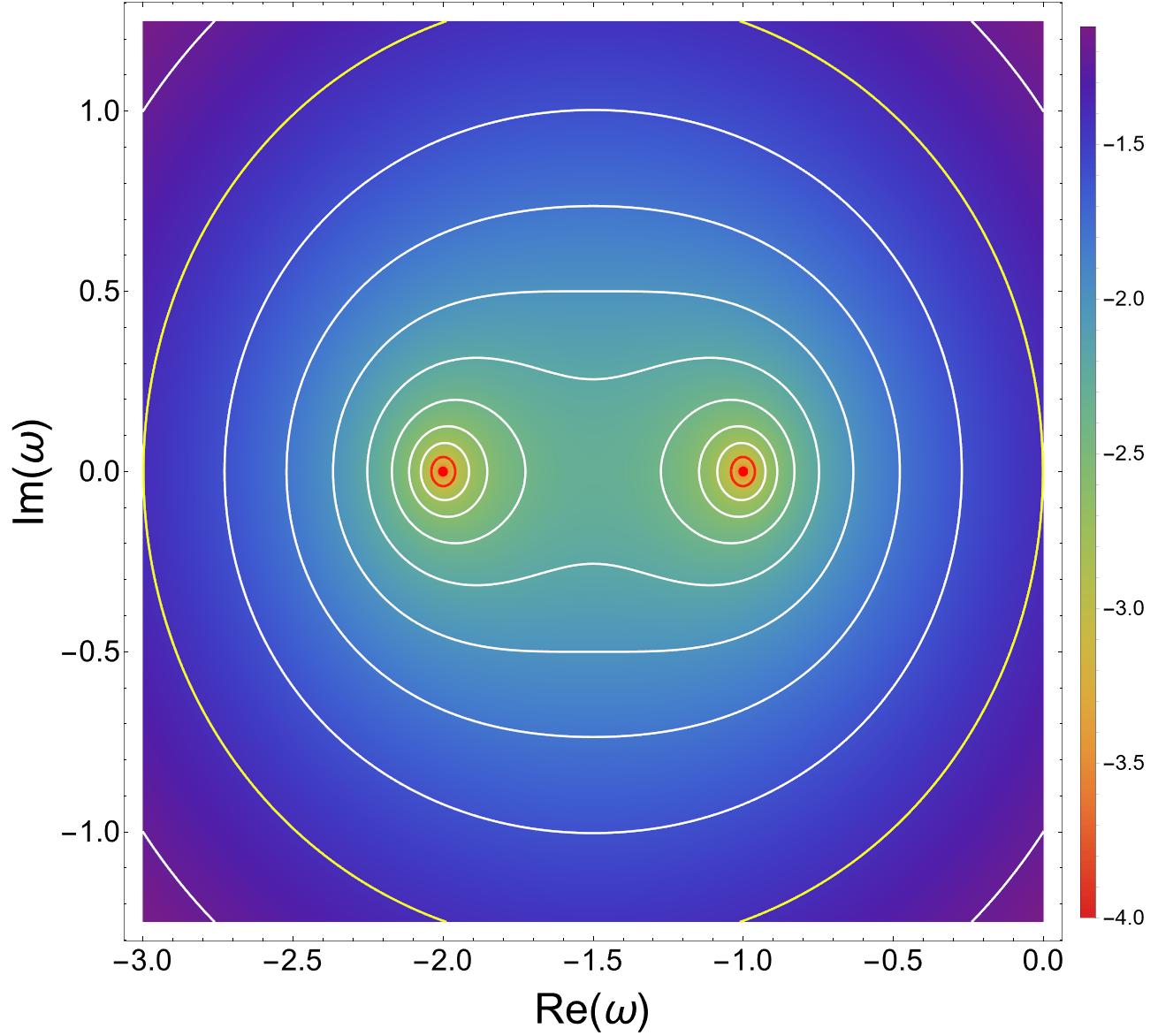}
            \captionsetup{justification=centering}
            \caption{$\ell^2_2$-norm}
            \label{fig:PseudospectrumExampleNonNormal}
        \end{subfigure}
        \begin{subfigure}[b]{0.49\linewidth}
            \centering
            \includegraphics[width=\linewidth]{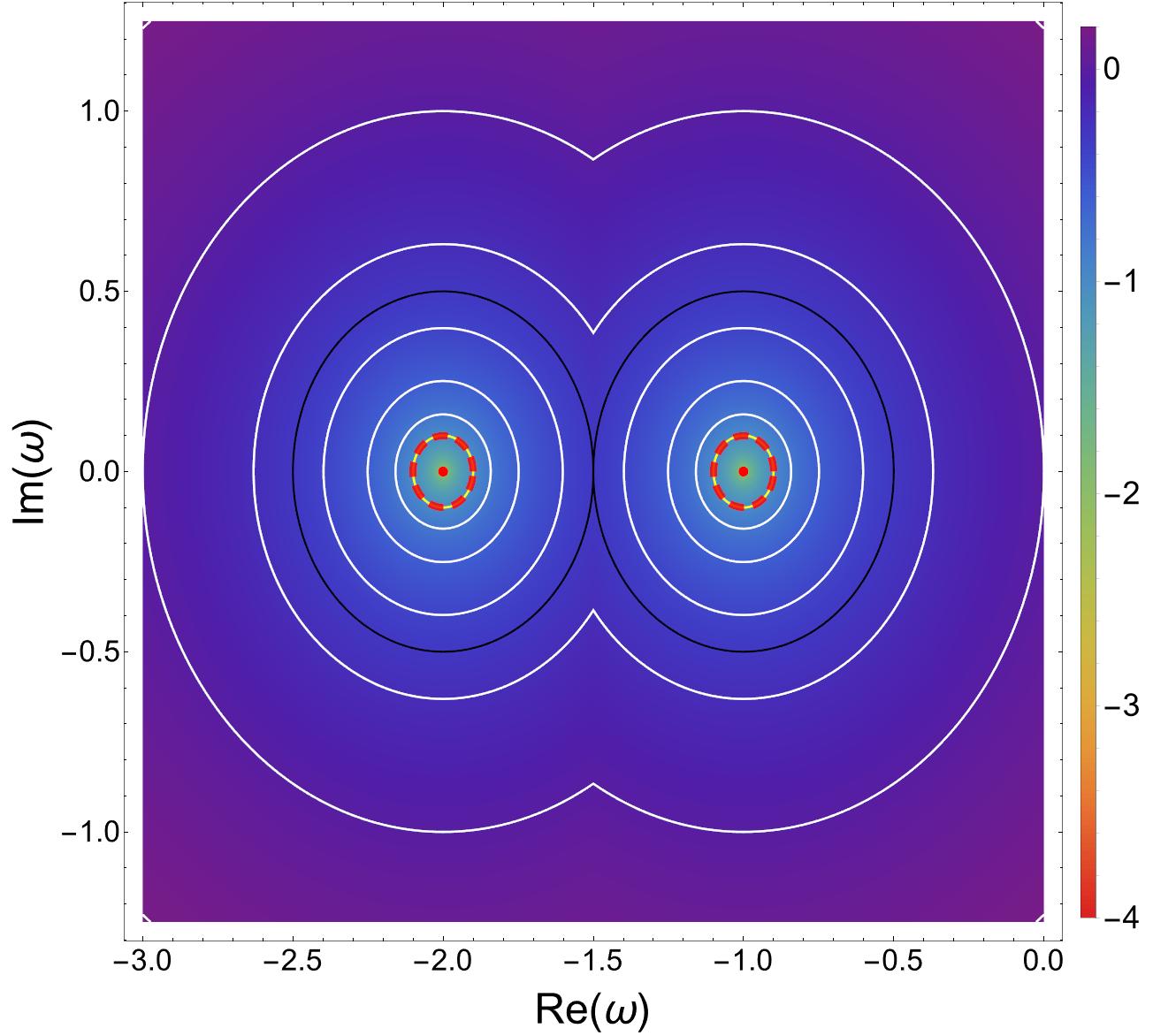}
            \captionsetup{justification=centering}
            \caption{$G$-norm}
            \label{fig:PseudospectrumExampleNormal}
        \end{subfigure}
        \caption{Pseudospectrum of the $A$ matrix \eqref{eq:A matrix Example} in two different norms. In both figures the red dots correspond to the eigenvalues, the white lines to boundaries of different $\varepsilon$-pseudospectra and the heat map to the logarithm of the inverse of the resolvent. In (a) the solid red and yellow lines represent the boundaries of $\sigma(A)+\Delta_{0.04}$ and $\sigma_{0.04}(A)$, respectively; note the great discrepancy associated with the non-normality. On the other hand, in (b) the dashed red line and the solid yellow line are superimposed and correspond to the boundaries of $\sigma(A)+\Delta_{0.1}$ and $\sigma_{0.1}(A)$, respectively. The solid black line in (b) represents the boundary of $\sigma_{0.5}(A)$.}
        \label{fig:PseudospectrumExample}
    \end{figure}

    \section{Quasinormal Modes and Pseudospectra}\label{section:Quasinormal Modes and Pseudospectra}
    Quasinormal modes (QNMs) in AdS are solutions to the linearized equations of motion in a black hole background with well-defined quasinormal frequencies (QNFs), subject to outgoing  boundary conditions on the horizon and normalizable boundary conditions on the AdS boundary \cite{Berti:2009kk}.\footnote{Imposing normalizable boundary conditions corresponds to choosing source-less excitations in the dual QFT or, equivalently, to selecting solutions where the leading term on the AdS boundary vanishes. Note that in some cases the subleading mode can also be identified with the source; this corresponds to a different choice of quantization scheme in the dual QFT \cite{Klebanov:1999tb}. For simplicity, in the present work we do not consider that possibility.} QNMs are small excitations of the black hole spacetime that eventually fall into the horizon and die out, provided the background is stable.\footnote{Here, stability of the background entails that perturbations decay as opposed to growing exponentially.} From the dual QFT perspective, they correspond to excitations over the thermal equilibrium that eventually thermalize. The decaying behavior is reflected in the associated QNFs, which are complex as the system is inherently non-conservative.

    As argued in the previous section, the spectra of non-conservative systems are potentially unstable.
    One thus needs to complement the study of QNFs with pseudospectrum analysis. 
    In order to achieve this, we 
    need to cast the problem of finding QNFs as an eigenvalue equation and define a physically motivated norm.
In this section we implement these two steps to determine the pseudospectrum of fluctuations in asymptotic AdS black hole geometries.
    
    Specifically, as a background we consider a static $(\D+1)$-dimensional AdS planar black hole (black brane) spacetime, whose dual is a d-dimensional strongly coupled QFT in thermal equilibrium at temperature given by the Hawking temperature of the black brane. The corresponding metric in Poincaré coordinates is given by:
    \begin{equation}\label{eq:Poincare Patch SAdS}
        ds^2 = \frac{r^2}{l^2}\left(-f(r) dt^2 + \delta_{ij} dx^i dx^j\right) + \frac{l^2}{r^2}\frac{dr^2}{f(r)}\,,
    \end{equation}
    where $l$ is the AdS radius and $f(r)$ is the blackening factor whose largest zero we denote by $r_h$. In these coordinates, the Hawking temperature $T$ is given by:
    \begin{equation}\label{eq:Temperature black brane}
        T=\frac{r_h^2}{4\pi l^2}\left.\frac{df(r)}{dr}\right|_{r=r_h}\,,
    \end{equation}
    and the AdS boundary and the horizon are located at $r=\infty$ and $r=r_h$, respectively.
    
    \subsection{Construction of the Eigenvalue Problem}\label{subsection:Construction of the Eigenvalue Problem}
    Our approach relies on utilizing regular coordinates, as introduced in \cite{Warnick:2013hba,Ficek:2023imn}, to express the linearized equations of motion as an eigenvalue problem while also implementing the outgoing boundary conditions as regularity conditions on the horizon.\footnote{Regular coordinates are the AdS equivalent to the hyperboloidal ones used in asymptotically flat spacetimes (see \cite{Schmidt:1993rcx,Zenginoglu:2011jz,Ansorg:2016ztf,PanossoMacedo:2018hab} for applications to QNMs and  \cite{Bizon:2020qnd,PanossoMacedo:2023qzp} for a more recent overview).} 
    
    For the sake of completeness, in this subsection we present the approach in a general setting, only assuming a black brane background.\footnote{Throughout this subsection we will not address the normalizable boundary conditions. As we will see later, these can be easily implemented as Dirichlet or Neumann boundary conditions with an adequate redefinition of the QNMs.}
    In \ref{subsubsection:Example: Scalar in SAdS5}, we specialize this discussion to a scalar in an AdS$_{4+1}$ Schwarzschild black brane (SAdS$_{4+1}$) background, hoping to shed more light on the nature of the approach.

    In general, we consider linearized equations of motion of the form:
    \begin{equation}\label{eq:General linearized eoms}
        \left[-\nabla_M \nabla^M +K^M\nabla_M +U\right]\Phi=0\,,
    \end{equation}
    with $\Phi$ a multi-component field, $K^M$ and $U$ functions and $\nabla_M$ the covariant derivative.
    
    Exploiting the symmetries of the black brane, we choose coordinates $\{\mathscr{t},\boldsymbol{\mathscr{x}},\mathscr{r}\}$ such that the linear operator of equation \eqref{eq:General linearized eoms} only depends on $\mathscr{r}$. Furthermore, in order to eliminate $\boldsymbol{\mathscr{x}}$ derivatives we express $\Phi$ as a superposition of Fourier modes: 
    \begin{equation}\label{eq:Fourier decomposition QNM general no frequency}
        \Phi(\mathscr{t},\boldsymbol{\mathscr{x}},\mathscr{r})=\int \frac{d^{\D-1}k}{(2\pi)^{\D-1}}\Phi(\mathscr{t},\mathbf{k},\mathscr{r}) e^{i\mathbf{k}\boldsymbol{\mathscr{x}}}\,,
    \end{equation}
    and solve for $\Phi(\mathscr{t},\mathbf{k},\mathscr{r})$ instead.
    We can then rewrite equation \eqref{eq:General linearized eoms} as
    \begin{equation}
        \label{eq:General linearized eoms as eigenvalue problem step 1}
        g^{\mathscr{t}\mathscr{t}}\partial_\mathscr{t}^2 \Phi(\mathscr{t},\mathbf{k},\mathscr{r})= F_1\left[\partial_\mathscr{r}^2,\partial_\mathscr{r};\mathbf{k},\mathscr{r}\right]\Phi(\mathscr{t},\mathbf{k},\mathscr{r})+F_2\left[\partial_\mathscr{r};\mathbf{k},\mathscr{r}\right]\partial_\mathscr{t}\Phi(\mathscr{t},\mathbf{k},\mathscr{r})\,,
    \end{equation}
    where $F_1\left[\partial_\mathscr{r}^2,\partial_\mathscr{r};\mathbf{k},\mathscr{r}\right]$ and $F_2\left[\partial_\mathscr{r};\mathbf{k},\mathscr{r}\right]$ are differential operators on $\mathscr{r}$  of second and first order, respectively.
As long as the $\mathscr{t}$-hypersurfaces are not null ($g^{\mathscr{t}\mathscr{t}}\neq0$), we can introduce the auxiliary variable $\Psi=\partial_\mathscr{t}\Phi$ and rewrite the equations of motion \eqref{eq:General linearized eoms as eigenvalue problem step 1} as:\footnote{For notational convenience, we refer to the $x^M=\text{const}$ hypersurface as $x^M$-hypersurface.}
    
    \begin{equation}
        \label{eq:General linearized eoms as eigenvalue problem step 2}
        \partial_\mathscr{t}\begin{pmatrix}
            \Phi(\mathscr{t},\mathbf{k},\mathscr{r})\\
            \Psi(\mathscr{t},\mathbf{k},\mathscr{r})
        \end{pmatrix} = 
        \begin{pmatrix}
            0&1\\
        L_1\left[\partial_\mathscr{r}^2,\partial_\mathscr{r};\mathbf{k},\mathscr{r}\right]&L_2\left[\partial_\mathscr{r};\mathbf{k},\mathscr{r}\right]   
        \end{pmatrix}\begin{pmatrix}
            \Phi(\mathscr{t},\mathbf{k},\mathscr{r})\\
            \Psi(\mathscr{t},\mathbf{k},\mathscr{r})
        \end{pmatrix}\,,
    \end{equation}
    where we have defined $L_i=\left(g^{\mathscr{t}\mathscr{t}}\right)^{-1}F_i$.

    Recalling that a QNM $\Phi(\mathscr{t},\mathbf{k},\mathscr{r})$ has a well-defined QNF $\omega$ or, equivalently, that it is an eigenfunction of the Killing vector associated with stationarity, $\mathfrak{t}=\partial_\mathscr{t}$:
    \begin{equation}
        \mathfrak{t}\begin{pmatrix}
            \Phi(\mathscr{t},\mathbf{k},\mathscr{r})\\
            \Psi(\mathscr{t},\mathbf{k},\mathscr{r})
        \end{pmatrix}=\partial_\mathscr{t}\begin{pmatrix}
            \Phi(\mathscr{t},\mathbf{k},\mathscr{r})\\
            \Psi(\mathscr{t},\mathbf{k},\mathscr{r})
        \end{pmatrix}=-i \omega \begin{pmatrix}
            \Phi(\mathscr{t},\mathbf{k},\mathscr{r})\\
            \Psi(\mathscr{t},\mathbf{k},\mathscr{r})
        \end{pmatrix}\,,
    \end{equation}
    we can cast expression \eqref{eq:General linearized eoms as eigenvalue problem step 2} as a standard eigenvalue problem:
    \begin{equation}
        \label{eq:General linearized eoms as eigenvalue problem step 3}
        \omega\begin{pmatrix}
            \Phi(\mathscr{t},\mathbf{k},\mathscr{r})\\
            \Psi(\mathscr{t},\mathbf{k},\mathscr{r})
        \end{pmatrix} = i
        \begin{pmatrix}
            0&1\\
        L_1\left[\partial_\mathscr{r}^2,\partial_\mathscr{r};\mathbf{k},\mathscr{r}\right]&L_2\left[\partial_\mathscr{r};\mathbf{k},\mathscr{r}\right]   
        \end{pmatrix}\begin{pmatrix}
            \Phi(\mathscr{t},\mathbf{k},\mathscr{r})\\
            \Psi(\mathscr{t},\mathbf{k},\mathscr{r})
        \end{pmatrix}\,.
    \end{equation}
    
    Note that to write the equations of motion as an eigenvalue problem it is fundamental that the $\mathscr{t}$-hypersurfaces are spacelike as opposed to null. Otherwise, we would obtain a generalized eigenvalue equation of the form:
   \begin{equation}
        \label{eq:General linearized eoms as eigenvalue problem null case}
        \omega F_2\left[\partial_\mathscr{r};\mathbf{k},\mathscr{r}\right]\Phi(\mathscr{t},\mathbf{k},\mathscr{r}) = -i F_1\left[\partial_\mathscr{r}^2,\partial_\mathscr{r};\mathbf{k},\mathscr{r}\right]\Phi(\mathscr{t},\mathbf{k},\mathscr{r})\,,
    \end{equation}
    instead of the standard eigenvalue problem discussed in section \ref{section:Pseudospectrum and Stability}.\footnote{Technically, one can extend the formalism presented in section \ref{section:Pseudospectrum and Stability} to also hold for generalized eigenvalue problems of the form $\left(L-\lambda A\right)u=0$. However, the adequate notion of $\varepsilon$-pseudospectrum becomes unclear as one has to choose whether to perturb $L$, $A$ or both (see Chapter 45 of \cite{Trefethen:2005} and references therein for an extensive discussion).}

    Regarding the boundary conditions, we seek to choose $\{\mathscr{t},\boldsymbol{\mathscr{x}},\mathscr{r} \}$ such that demanding regularity on the horizon is equivalent to imposing outgoing boundary conditions. This is typically achieved working with infalling Eddington-Finkelstein (IEF) coordinates $\{u,\tilde{\mathbf{x}},\tilde{r}\}$, which, in terms of the Poincaré coordinates $\{t,\mathbf{x},r\}$, are given by:
    \begin{equation}\label{eq:IEF coordinates}
        u=t+\int \frac{dr}{f(r)} \left(\frac{l}{r}\right)^2\,, \qquad \tilde{\mathbf{x}}=\mathbf{x}\,, \qquad \tilde{r}=r\,.
    \end{equation}
    However, writing the metric \eqref{eq:Poincare Patch SAdS} in IEF coordinates:
    \begin{equation}\label{eq:IEF metric}
        ds^2 = \frac{\tilde{r}^2}{l^2}\left(-f(\tilde{r}) du^2 + \delta_{ij} dx^i dx^j\right) + 2du  d\tilde{r}\,,
    \end{equation}
    we conclude that the $u$-hypersurfaces are null. Thus IEF coordinates do not result in an standard eigenvalue problem for the QNMs.

    Nonetheless, we can define regular coordinates that resemble the IEF ones near the horizon while having spacelike $\mathscr{t}$-hypersurfaces outside the horizon. More concretely, we construct regular coordinates such that the spacelike $\mathscr{t}$-hypersurfaces match the null IEF $u$-hypersurfaces exactly on the horizon \cite{Warnick:2013hba}. A particularly simple choice is given by:
    \begin{equation}\label{eq:Regular Coordinates}
        \mathscr{t}=u-\frac{l^2}{r_h}\left(1-\frac{r_h}{\tilde{r}}\right)=t-\frac{l^2}{r_h}\left(1-\frac{r_h}{r}\right)+\int \frac{dr}{f(r)}\left(\frac{l}{r}\right)^2\,, \qquad \boldsymbol{\mathscr{x}}=\tilde{\mathbf{x}}=\mathbf{x} \,, \qquad\mathscr{r}=\tilde{r}=r\,,
    \end{equation}
    in terms of which the metric \eqref{eq:Poincare Patch SAdS} is:
    \begin{equation}\label{eq:General Regular metric}
        ds^2 = \frac{\mathscr{r}^2}{l^2}\left(-f(\mathscr{r}) d\mathscr{t}^2 + \delta_{ij} d\mathscr{x}^i d\mathscr{x}^j\right) + 2\left(1-f(\mathscr{r})\right)d\mathscr{t}  d\mathscr{r}+\frac{l^2}{\mathscr{r}^2}\left(2-f(\mathscr{r})\right)d\mathscr{r}^2\,.
    \end{equation}

    In practice, for most numerical applications it is  more convenient to introduce a dimensionless compactified radial coordinate $\rho$:
    \begin{equation}
        \rho=1-\frac{r_h}{\mathscr{r}}\,,
    \end{equation}
    such that the horizon is mapped to $\rho=0$ and the AdS boundary to $\rho=1$. In terms of these compactified regular coordinates $\{\mathscr{t},\boldsymbol{\mathscr{x}},\rho \}$, the black brane metric \eqref{eq:Poincare Patch SAdS} takes the form:
    \begin{equation}
        \label{eq:General Regular metric compactified}
        ds^2 = \frac{l^2}{(1-\rho)^2}\left(-\mathscr{f}(\rho) d\mathscr{t}^2 + \delta_{ij} d\mathscr{x}^id\mathscr{x}^j +2\left(1-\mathscr{f}(\rho)\right)d\mathscr{t} d\rho +\left(2-\mathscr{f}(\rho)\right)z_h^2 d\rho^2\right) \,.
    \end{equation}
    where $\mathscr{f}(\rho)=f\left(r\left(\rho\right)\right)$ and $z_h=l^2/r_h$.

    It is worth stressing that the components of the Killing vector $\mathfrak{t}=\partial_{\mathscr{t}}=\partial_t$ are invariant under the change of coordinates from regular to Poincaré coordinates. Consequently, the QNFs defined as the eigenvalues of $\partial_t$ and $\partial_\mathscr{t}$ match. In gauge/gravity duality this is important as the holographic dictionary identifies the poles of the retarded propagators with the eigenvalues of $\partial_t$.\footnote{Here, by the holographic dictionary we refer to the standard one formulated in Poincaré coordinates. Note that one can always formulate a new holographic dictionary in a different coordinate set and identify the QNFs with the poles of the retarded propagators.}

    Before moving on to a concrete example, we summarize the core principles behind the regular coordinates: Given an AdS$_{\D+1}$ black brane spacetime whose metric in Poincaré coordinates is given by equation \eqref{eq:Poincare Patch SAdS} a regular coordinate set $\{\mathscr{t},\boldsymbol{\mathscr{x}},\mathscr{r}\}$ is one that fulfils the following three requirements: the metric components are functions  only of $\mathscr{r}$, $\mathscr{t}$-hypersurfaces are spacelike outside the horizon, and those
        $\mathscr{t}$-hypersurfaces
        match the IEF $u$-hypersurfaces on the horizon. Imposing outgoing boundary conditions in these coordinates  corresponds to demanding regularity on the horizon. Moreover, equations of motion
        for the QNMs of the form \eqref{eq:General linearized eoms} can be written as a standard eigenvalue problem.

    \subsubsection{Example: Real Scalar in \texorpdfstring{$\text{SAdS}_{4+1}$}{SAdS4+1}}\label{subsubsection:Example: Scalar in SAdS5}

    To better clarify the role played by the regular coordinates in the construction of the eigenvalue problem, here we specialize the discussion to a real scalar field $\phi$ with action
    \begin{equation}\label{eq:Action real scalar}
        \mathcal{S}[\phi]=-\frac{1}{2}\int d^{4+1}x \sqrt{-g}\,\, \left(\partial_M\phi\partial^M \phi+m^2 \phi^2 \right)\,,
    \end{equation}
    in a SAdS$_{4+1}$ background. The corresponding linearized equations of motion are given by:
    \begin{equation}\label{eq:EOM real scalar}
        \left[-\nabla_M \nabla^M + m^2\right]\phi=0\,.
    \end{equation}
    
    For better emphasis on the key aspects of our approach, we present this discussion in Poincaré, IEF and compactified regular coordinates.
    
   \hfill\break
        \textbf{Poincaré Coordinates}: the SAdS$_{4+1}$ metric in Poincaré coordinates is given by \eqref{eq:Poincare Patch SAdS} with $f(r)=1-\left(\frac{r_h}{r}\right)^4$ and $\{i,j\}=1,2,3$.
        In terms of the Fourier modes $\phi(t,\mathbf{k},r)$, the linearized equation of motion \eqref{eq:EOM real scalar} for a QNM trivially produces the following eigenvalue problem:
        \begin{equation}\label{eq:EOM real scalar Poincaré} 
            \omega^2\,\phi(t,\mathbf{k},r)=f(r) \left[ \mathbf{k}^2 +\frac{m^2 r^2 }{l^2} -\frac{\left(r^5  f(r)\right)' }{r l^4 } \partial_r -\frac{r^4 f(r)}{l^4}\partial_r^2\right]\phi(t,\mathbf{k},r)\,,
        \end{equation}
        where we have used that the QNM is an eigenfunction of $\partial_t$ with QNF $\omega$.
        
        In order to impose outgoing boundary conditions, we first analyze the behavior near the horizon $r=r_h$. Solving the equations of motion with a Frobenius series, we find:
        \begin{equation}
            \phi(t,\mathbf{k},r)= e^{-i\omega t}\left\{ c_1 \exp\left[-i\frac{\omega l^2}{4 r_h} \log(r-r_h)\right] \left(1+...\right) + c_2 \exp\left[i\frac{\omega l^2}{4 r_h} \log(r-r_h)\right] \left(1+...\right) \right\}\,.
        \end{equation}
        
        As the outgoing solution corresponds to fixing $c_2=0$, imposing outgoing boundary conditions correspond to fixing a near-horizon behavior:
        \begin{equation}\label{eq:outgoing behavior auxiliar}
            \phi(t,\mathbf{k},r\rightarrow r_h)\approx c_1\exp\left[-i\omega\left(t + \frac{l^2}{4 r_h} \log(r-r_h) \right)\right]\,.
        \end{equation}
        
        This choice for the outgoing solution matches our physical interpretation of a mode that falls into the horizon. Defining the tortoise coordinate $r_*$
        \begin{equation}
            r_*= \frac{l^2 }{2 r_h}\left[\arctan\left(\frac{r}{r_h}\right)-\arctanh\left(\frac{r_h}{r}\right)\right]=\frac{l^2}{4 r_h} \log(r-r_h)+ \mathcal{O}\left((r-r_h)^0\right)\,,
        \end{equation}
        the near-horizon behavior can be written as two modes, one falling towards the brane and one exiting it:
        \begin{equation}
            \phi(t,\mathbf{k},r\rightarrow r_h)\approx c_1\exp\left[-i\omega\left(t + r_* \right)\right] + c_2\exp\left[-i\omega\left(t - r_* \right)\right]\,,
        \end{equation}
        and choosing a QNM that falls into the horizon corresponds to discarding the ingoing mode, \textit{i.e.}, taking $c_2=0$.
        
       \hfill\break
       \textbf{IEF Coordinates}: the SAdS$_{4+1}$ metric in IEF coordinates is given by \eqref{eq:IEF metric} with $f(\tilde{r})=1-\left(\frac{r_h}{\tilde{r}}\right)^4$ and $\{i,j\}=1,2,3$.
        The linearized  equation of motion for a QNM \eqref{eq:EOM real scalar} reads:
        \begin{equation}\label{eq:EOM real scalar IEF} 
            \omega \left[\frac{3}{\tilde{r}}+2 \partial_{\tilde{r}}\right] \phi(u,\mathbf{k},\tilde{r})=i\left[ \frac{\mathbf{k}^2 l^2}{\tilde{r}^2} +m^2-\frac{\left(\tilde{r}^5f(\tilde{r})\right)' }{\tilde{r}^3 l^2}\partial_{\tilde{r}} - \frac{\tilde{r}^2 f(\tilde{r})}{l^2}\partial_{\tilde{r}}^2\right] \phi(u,\mathbf{k},\tilde{r})\,,
        \end{equation}
        and it cannot be cast as a standard eigenvalue problem due to the derivative term.  
        The near-horizon solution is given by:
        \begin{equation}
            \phi(u,\mathbf{k},\tilde{r})= e^{-i\omega u}\left\{ c_1 \left(1+...\right) + c_2 \exp\left[i\frac{\omega l^2}{2 r_h} \log(\tilde{r}-r_h)\right] \left(1+...\right) \right\}\,,
        \end{equation}
        with the outgoing solution corresponding to fixing $c_2=0$.\footnote{This follows from the Poincaré result noting that under the coordinate transformation \eqref{eq:IEF coordinates} we have:
        \begin{equation*}
            e^{-i\omega u}=\exp\left[-i\omega\left(t+\frac{l^2}{4r_h}\log\left(r-r_h\right)\right)\right]+ \mathcal{O}\left((r-r_h)^0\right)\,.
        \end{equation*}} Consequently, imposing outgoing boundary conditions corresponds to demanding regularity on the horizon:
        \begin{equation}
            \phi(t,\mathbf{k},\tilde{r}\rightarrow r_h)\approx c_1 e^{-i\omega u}\,.
        \end{equation}
        
        \hfill\break
        \textbf{Compactified Regular Coordinates}:
        the SAdS$_{4+1}$ metric in compactified regular coordinates $\{\mathscr{t},\boldsymbol{\mathscr{x}},\rho\}$ is given by \eqref{eq:General Regular metric compactified}  with $\mathscr{f}(\rho)=1-\left(1-\rho\right)^4$ and $\{i,j\}=1,2,3$.
        Then, the linearized equation of motion for a QNM \eqref{eq:EOM real scalar} is:
        \begin{align}\label{eq:EOM real scalar Regular} 
            z_h^2 \omega^2 \left[\mathscr{f}(\rho)-2\right]\phi(\mathscr{t},\mathbf{k},\mathscr{r})=iz_h\omega\left[(1-\rho)^3\left(\frac{\mathscr{f}(\rho)-1}{(1-\rho)^3}\right)' +2 \left(\mathscr{f}(\rho)-1\right)\partial_{\rho}  \right]\phi(\mathscr{t},\mathbf{k},\mathscr{r})\nonumber\\ 
            -\left[\frac{m^2 l^2}{(1-\rho)^2}+z_h^2\mathbf{k}^2-(1-\rho)^3\left(\frac{\mathscr{f}(\rho) }{(1-\rho)^3}\right)'\partial_{\rho} -\mathscr{f}(\rho)\partial_{\rho}^2 \right]\phi(\mathscr{t},\mathbf{k},\mathscr{r})\,,
        \end{align}
        and, introducing the dimensionless variables
        \begin{equation}
            \mathfrak{w}=z_h\omega=\frac{\omega}{\pi T}\,,\quad \boldsymbol{\mathfrak{q}}=z_h\mathbf{k}=\frac{\mathbf{k}}{\pi T}\,,
        \end{equation}
        and the auxiliary field $\psi(\mathscr{t},\mathbf{k},\mathscr{r})=z_h\partial_\mathscr{t}\phi(\mathscr{t},\mathbf{k},\mathscr{r})$, equation \eqref{eq:EOM real scalar Regular}  can be written as an eigenvalue problem
        \footnote{One could cast the eigenvalue problem without introducing the dimensionless variables $\mathfrak{w}$ and $\boldsymbol{\mathfrak{q}}$. However, it is convenient to use them as they allows us to remove all $z_h$ dependence and thus reduce the parameter space.}
        \begin{equation}\label{eq:Eigenvalue problem real scalar 1}
            \mathfrak{w}\begin{pmatrix}
            \phi(\mathscr{t},\mathbf{k},\rho)\\
            \psi(\mathscr{t},\mathbf{k},\rho)
            \end{pmatrix} = i
            \begin{pmatrix}
            0&1 \\
            L_1\left[\partial_\mathscr{\rho}^2,\partial_\mathscr{\rho};\boldsymbol{\mathfrak{q}},\rho\right]&L_2\left[\partial_\rho;\boldsymbol{\mathfrak{q}},\rho\right]   
            \end{pmatrix}\begin{pmatrix}
            \phi(\mathscr{t},\mathbf{k},\rho)\\
            \psi(\mathscr{t},\mathbf{k},\rho)
            \end{pmatrix}\,,
        \end{equation}
        where the differential operators
        $L_1$ and $L_2$ take the form:
        \begin{align}\label{eq:Eigenvalue problem real scalar 2}
            L_1\left[\partial_\rho^2,\partial_\rho;\boldsymbol{\mathfrak{q}},\rho\right]&=\left[\mathscr{f}(\rho)-2\right]^{-1}\left[\frac{m^2 l^2}{(1-\rho)^2}+\boldsymbol{\mathfrak{q}}^2-(1-\rho)^3\left(\frac{\mathscr{f}(\rho) }{(1-\rho)^3}\right)'\partial_{\rho} -\mathscr{f}(\rho)\partial_{\rho}^2 \right]\,,\\[\medskipamount]\label{eq:Eigenvalue problem real scalar 3}
            L_2\left[\partial_\rho;\boldsymbol{\mathfrak{q}},\rho\right]&=\left[\mathscr{f}(\rho)-2\right]^{-1}\left[(1-\rho)^3\left(\frac{\mathscr{f}(\rho)-1}{(1-\rho)^3}\right)' +2 \left(\mathscr{f}(\rho)-1\right)\partial_{\rho}  \right]\,.
        \end{align}
            
        With regard to the near-horizon behavior, we have:
        \begin{equation}
            \phi(\mathscr{t},\mathbf{k},\rho)= e^{-i\omega\mathscr{t}}\left\{ c_1 \left(1+...\right) + c_2 \exp\left[i\frac{z_h\omega}{2} \log(\rho)\right] \left(1+...\right) \right\}\,,
        \end{equation}
        with the outgoing solution corresponding to fixing $c_2=0$.\footnote{Once again, this follows from the Poincaré result as under the coordinate transformation \eqref{eq:Regular Coordinates} we have:
        \begin{equation*}
            e^{-i\omega\mathscr{t}}=\exp\left[-i\omega\left(t+\frac{l^2}{4r_h}\log\left(r-r_h\right)\right)\right]+ \mathcal{O}\left((r-r_h)^0\right)\,.
        \end{equation*}} Then, imposing outgoing boundary conditions corresponds to demanding regularity on the horizon:
        \begin{equation}
            \phi(\mathscr{t},\mathbf{k},\rho\rightarrow 0)\approx c_1e^{-i\omega\mathscr{t}}\,.
        \end{equation}

    Therefore, it is clear that regular coordinates constitute a sort of middle ground between IEF and Poincaré coordinates.
    They exploit the regular behavior on the horizon of the IEF coordinates while maintaining the spacelike character of the constant-time hypersurfaces found in Poincaré coordinates. In fact, with our particular choice of regular coordinates, the $\mathscr{t}$-hypersurfaces can be thought of as interpolating between $u$- and $t$-hypersurfaces (see figure \ref{fig:Penrose Diagram SAdS5}).

    \begin{figure}[h!]
        \centering
        \includegraphics[height=0.6\linewidth]{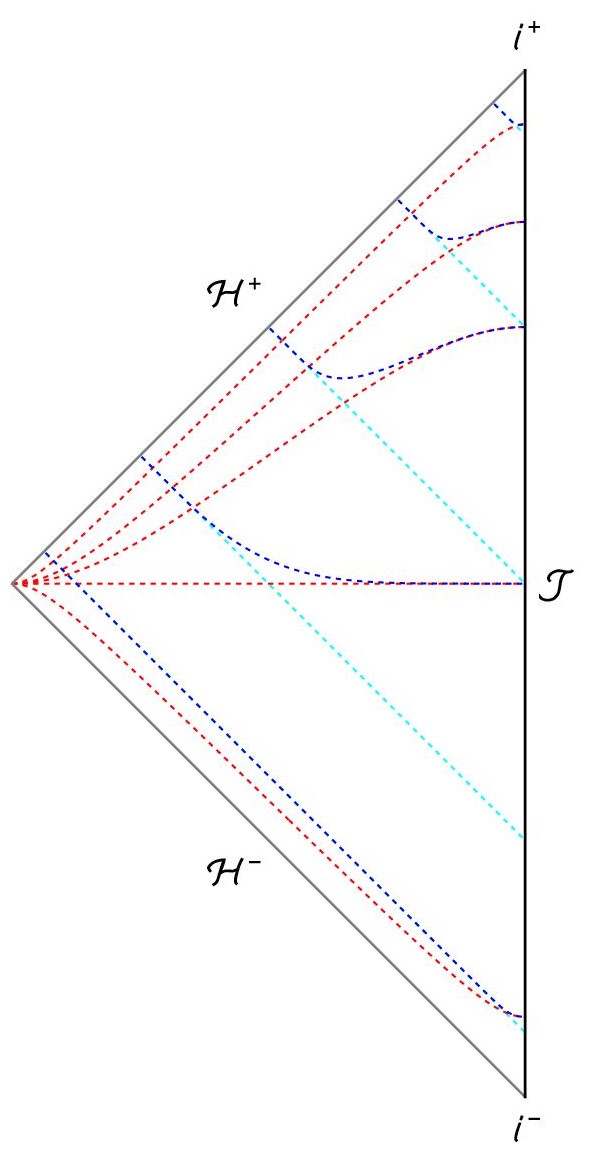}
        \caption{Penrose Diagram of SAdS$_{4+1}$. The AdS boundary is denoted by $\mathcal{J}$, $\mathcal{H}^{+}$ ($\mathcal{H}^{-}$) represents the future (past) horizon and $i^+$ ($i^-$) denotes the future (past) time-like infinity. The dashed blue lines correspond to regular $\mathscr{t}$-hypersurfaces and the red and cyan dashed lines represent Poincaré $t$ and IEF $u$-hypersurfaces, respectively.}
        \label{fig:Penrose Diagram SAdS5}
    \end{figure}
        
    \subsection{Choice of Norm}

    As discussed in section \ref{section:Pseudospectrum and Stability}, pseudospectra are norm dependent. Moreover, we concluded that the spectrum of an operator may be unstable in one norm but not in a different one (as illustrated in figure \ref{fig:PseudospectrumExample}), which shows the nontrivial role of the norm. Ideally, we want to find a physically motivated norm such that our physical notion of smallness matches the mathematical one. 

    Implicitly, we have been assuming that the QNMs are fluctuations small enough not to affect the background, or equivalently, that their backreaction is negligible.\footnote{This is the underlying assumption behind studying their linearized equations of motion on a fixed background.} Consequently, the physically motivated notion of size for the QNMs is their contribution to the energy-momentum tensor, which is directly related to their backreaction. Thus, the adequate norm is the energy norm $\norm{\cdot}_E$, introduced in \cite{Jaramillo:2020tuu,Gasperin:2021kfv} in the case of quasinormal modes and generally advocated for in \cite{Trefethen:2005}\footnote{For an application of pseudosepctra with energy norm in hydrodynamics see \cite{reddy1993pseudospectra}}. It defines the norm of a QNM $\Phi$ as its energy in a constant-time hypersurface:
    \begin{equation}\label{eq:Energy norm def}
        \norm{\Phi}_E^2=-\int_{\mathscr{t}=\text{const}} \star J\left[\Phi\right]\,,
    \end{equation}
    where $\star J\left[\Phi\right]$ is the Hodge dual of the current 1-form
    \begin{equation}\label{eq:current 1-form}
        J\left[\Phi\right]=\mathfrak{t}^M T_{MN}\left[\Phi\right] dx^N\,,
    \end{equation}
    with $T_{MN}\left[\Phi\right]$ the QNM's leading order contribution to the energy-momentum tensor. 
    
    Therefore, 
    the energy norm is generically given by:\footnote{Note that the leading order contribution of the fluctuations to the energy momentum tensor is quadratic. This occurs because the linear contribution vanishes as it is proportional to the equations of motion of the background.}
    \begin{align}\label{eq:Energy norm general} 
        \norm{\Phi}_E^2=\int_{\mathscr {t}=\text{const}} d^\D x\,\, \Phi^*(\mathscr{t},\boldsymbol{\mathscr{x}},\mathscr{r}) &\left[\overleftarrow{\partial}_M  H_1(\mathscr{r}) \overrightarrow{\partial}^M+H_M(\mathscr{r}) \overrightarrow{\partial}^M +\right.\nonumber\\ 
        &\left. +\overleftarrow{\partial}^M H_M^*(\mathscr{r})+ H_2(\mathscr{r})\right] \Phi(\mathscr{t},\boldsymbol{\mathscr{x}},\mathscr{r}) \,,
    \end{align}
    where $H_1(\mathscr{r}),H_2(\mathscr{r}),H_M(\mathscr{r})$ are functions respecting the black brane symmetries. Expression \eqref{eq:Energy norm general} can be further simplified decomposing into Fourier modes and introducing the auxiliary field  $\Psi(\mathscr{t},\mathbf{k},\mathscr{r})=\partial_\mathscr{t}\Phi(\mathscr{t},\mathbf{k},\mathscr{r})$, obtaining:
    \begin{equation}\label{eq:Energy norm general 2}
        \norm{\Phi}_E^2=\int_{\mathscr {t}=\text{const}} \frac{d^{\D-1}{k}}{(2\pi)^{\D-1}}\,d\mathscr{r}\,\, 
        \begin{pmatrix}
            \Phi^*(\mathscr{t},\mathbf{k},\mathscr{r}) &\Psi^*(\mathscr{t},\mathbf{k},\mathscr{r})
        \end{pmatrix}
        \mathcal{G}\left[\overleftarrow{\partial}_\mathscr{r},\overrightarrow{\partial}_\mathscr{r};\mathbf{k},\mathscr{r}\right]
        \begin{pmatrix}
            \Phi(\mathscr{t},\mathbf{k},\mathscr{r})\\[\medskipamount]
            \Psi(\mathscr{t},\mathbf{k},\mathscr{r})
        \end{pmatrix}\,,
    \end{equation}
    where we have defined 
    \begin{equation}\label{eq:G operator}
        \mathcal{G}\left[\overleftarrow{\partial}_\mathscr{r},\overrightarrow{\partial}_\mathscr{r};\mathbf{k},\mathscr{r}\right]=\begin{pmatrix}
        \mathcal{G}_{11}\left[\overleftarrow{\partial}_\mathscr{r},\overrightarrow{\partial}_\mathscr{r};\mathbf{k},\mathscr{r}\right]&\mathcal{G}_{12}\left[\overleftarrow{\partial}_\mathscr{r};\mathbf{k},\mathscr{r}\right]\\[\medskipamount]
            \mathcal{G}_{21}\left[\overrightarrow{\partial}_\mathscr{r};\mathbf{k},\mathscr{r}\right] & \mathcal{G}_{11}[\mathscr{r}]
        \end{pmatrix}\,,
    \end{equation}
    with $\mathcal{G}_{ab}$ differential operators on $\mathscr{r}$.
    
    As discussed in the previous section, we present the eigenvalue problem for the Fourier modes, \textit{i.e.}, we work out the spectral stability of the operator acting on the subspace of QNMs with well-defined momentum $\mathbf{k}$. Consequently, as we do not consider mixing between subspaces, we can drop the integral over $\mathbf{k}$ and introduce a constant prefactor $C$ instead. Furthermore, as we are mainly interested in the operator norm (which is independent of constant prefactors by virtue of definition \ref{def:Operator norm Hilbert space}), we can set $C=1$ without losing any generality. With all these considerations, the relevant inner product can be written as:
    \begin{equation}
        \expval{\Phi_1,\Phi_2}_E=\int_{\mathscr {t}=\text{const}}d\mathscr{r}\,\, 
        \begin{pmatrix}
            \Phi_1^*(\mathscr{t},\mathbf{k},\mathscr{r}) &\Psi_1^*(\mathscr{t},\mathbf{k},\mathscr{r})
        \end{pmatrix}
        \mathcal{G}\left[\overleftarrow{\partial}_\mathscr{r},\overrightarrow{\partial}_\mathscr{r};\mathbf{k},\mathscr{r}\right]
        \begin{pmatrix}
            \Phi_2(\mathscr{t},\mathbf{k},\mathscr{r})\\[\medskipamount]
            \Psi_2(\mathscr{t},\mathbf{k},\mathscr{r})
        \end{pmatrix}\,.
    \end{equation}

    Lastly, note that one needs to explicitly check that the energy norm is positive definite; otherwise, it would not be a well-defined norm. Whether this condition is satisfied depends on the particular choice of Hilbert space. Consequently, one must select a function space that enforces the desired asymptotic behavior for the QNMs while also ensuring that all functions residing in the aforementioned space have positive energy. In section \ref{section:Holographic Model} we discuss this topic in detail for our particular model.

    \subsubsection{The Nature of the Energy Norm}
 
    In this subsection we address the coordinate dependence of the energy norm. Definition \eqref{eq:Energy norm def} depends on the selection of constant-time hypersurfaces, and consequently, one could wonder if regular $\mathscr{t}$-hypersurfaces actually represent a well-motivated choice. 

    A physically well-defined energy should decay over time, as the fluctuations fall into the black brane. In order to showcase that the energy norm \eqref{eq:Energy norm def} satisfies this property, we consider the simple model introduced in subsection \ref{subsubsection:Example: Scalar in SAdS5}. In particular, to stress the importance of choosing regular $\mathscr{t}$-hypersurfaces, we compute the time derivative of energy for the real scalar field in Poincaré, IEF and compactified regular coordinates. Completely analogous considerations hold in the case of the transverse gauge field.
    
    \hfill\break
    \noindent\textbf{Poincaré Coordinates}: the energy of a QNM in Poincaré coordinates is given by
    \begin{equation}
        E[\phi]=\frac{1}{2l^5}\int_{t=\text{const}} d^3\mathbf{x}\, dr \,\,r^3\,\left[m^2l^2 \phi^2 +   \frac{l^4 }{r^2} \left(\partial_\mathbf{x}\phi\right)^2 + r^2f(r)\left(\partial_r\phi\right)^2+
         \frac{l^4}{r^2f(r)} \left(\partial_t \phi \right)^2\right]\,,
    \end{equation}
    where, for notational convenience, we omit the dependence of the field $\phi$ on the coordinates. Using the conservation equation of the current \eqref{eq:current 1-form}:
    \begin{equation}
        d \star J=0\,,
    \end{equation}
    we can conclude that the time derivative of the energy vanishes
    \begin{equation}
        \frac{d}{dt}E[\phi] = \frac{r_h^5 }{l^5} f(r_h)\int_{r=r_h} d^3\mathbf{x}\,\, \partial_t\phi \partial_r \phi = 0\,
    \end{equation}
    as it is proportional to the blackening factor evaluated at the horizon.\footnote{In general, we would also have a term arising from the AdS boundary. However, in all the coordinates studied in this section, this term vanishes due to the normalizable boundary conditions. Physically this behavior was to be expected, as the AdS boundary behaves as a wall for the QNMs.} A quasinormal mode however is an exponentially decaying solution and one would expect the energy to decay as well. The resolution to this puzzle is that in Poincar{\'e} coordinates a QNM is actually singular on the horizon and thus the energy is not well defined. We also note that the term proportional to the time derivatives in the energy is singular on the horizon. For these reasons the Poincar{\'e} coordinates are not suited to define a norm on which the pseudospectrum analysis could be based.
    \hfill\break
    
    \noindent\textbf{IEF Coordinates}: in IEF coordinates, the energy of a QNM is
    \begin{equation}
        E[\phi]=\frac{1}{2l^5}\int_{u=\text{const}} d^3\tilde{\bold{x}}\, d\tilde{r}\,\,\tilde{r}^3\,\left[m^2l^2 \phi^2 + \frac{l^4}{\tilde{r}^2} \left(\partial_{\tilde{\bold{x}}}\phi\right)^2 + \tilde{r}^2f(\tilde{r})\left(\partial_{\tilde{r}} \phi\right)^2\right]\,,
    \end{equation}
    and its time derivative is:
    \begin{equation}
        \frac{d}{du}E[\phi]  =- \frac{r_h^5}{l^5}\int_{\tilde{r}=r_h} d^3\tilde{\mathbf{x}}\,\, (\partial_u\phi)^2 < 0\,.
    \end{equation}

    Thus, in IEF the decay rate of the energy is a function of the time derivative of the QNMs. Indeed, this was the picture we expected, proving that the IEF coordinates are well suited to define a physically motivated energy. However, IEF coordinates do not generate a standard eigenvalue problem; and thus we do not use them. 
    \hfill\break
    
    \noindent\textbf{Compactified Regular Coordinates}: in compactified regular coordinates, the energy of a QNM associated with a real scalar with action \eqref{eq:Action real scalar} is
        \begin{align} E[\phi]=\frac{l^3}{2z_h^3}\int_{\mathscr{t}=\text{const}}d^3\boldsymbol{\mathscr{x}}\,\frac{d\rho}{(1-\rho)^3}\,\, \Biggl[&\left(\frac{m^2l^2}{z_h(1-\rho)^2}\phi^2+ z_h\left(\partial_{\boldsymbol{\mathscr{x}}}\phi \right)^2 \right) + \frac{\mathscr{f}(\rho)}{z_h}\left(\partial_\rho \phi\right)^2\nonumber\\
    &+z_h\left(2-\mathscr{f}(\rho)\right)\left(\partial_\mathscr{t}
        \phi\right)^2 \Biggr]\,,
    \end{align}
     and its time derivative is
     \begin{equation}
        \frac{d}{d\mathscr{t}}E[\phi] =
        -\frac{l^3}{z_h^3}\int_{\rho=0} d^3\boldsymbol{\mathscr{x}}\,\, 
        (\partial_\mathscr{t}\phi)^2 < 0\,,
    \end{equation}
    which, as IEF coordinates, does not vanish.
    Since the compactified regular coordinates also allow to compute the QNMs from a standard eigenvalue problem it is the best suited coordinate system for the pseudospectrum analysis.

    \subsubsection{Time Independence of the Operator Norm}

    Since we have found that the energy norm in regular coordinates is time dependent, one could ask what $\mathscr{t}$-hypersurface we should choose to define it. This is however not a relevant question for the 
    pseudospectrum analysis. 
    Pseudospectra depend on the norm \eqref{eq:Energy norm general} only trough the operator norm \eqref{eq:Operator norm Hilbert space}, which, by virtue of its definition, is time independent. 

    For functions belonging to the relevant function space defined by the boundary conditions, the energy norm indeed defines a positive definite inner product and thus promotes the function space to a Hilbert space. Functions change over time, moving around in the Hilbert space. However, the Hilbert space itself does not change, as it contains all possible function configurations. Consequently, noting that the operator norm is defined as the maximum value of the norm of the operator acting on a normalized member of the aforementioned Hilbert space, we can trivially conclude that the operator norm is time independent. Thus, we need not worry about choosing one particular $\mathscr{t}$-hypersurface or another.

    \subsection{Spectral and Background Instability}\label{ssec:Spectral and Spacetime Instability}
    Before moving on, this is a good point to discuss in detail the holographic perspective on the (in)stability probed through pseudospectra. From the gravitational side, pseudospectra probe how much the QNFs change if the background is slightly modified in some way (\textit{e.g.}, by a change of the geometry and/or the background value of the fields).
    Consequently, in the dual quantum field theory, pseudospectra help us estimate how much the poles of the retarded Green's functions might change if the theory is slightly perturbed. In both cases, these perturbations should be understood as perturbations to the Lagrangian, leading to a change in the spectrum of excitations. To stress this point, in this section we will explicitly refer to the (in)stability associated with these perturbations as spectral (in)stability.

    Spectral instability is in nature very different from the usual background instability which occurs when a linear fluctuation grows indefinitely. The latter signifies an instability of the background or equivalently the ground state of the dual field theory. In contrast, spectral instability denotes instability of the spectrum of the model under slight deformations. Heuristically, spectral instability implies that a small perturbation to the Lagrangian can have a large effect on the energy spectrum.

    Hence, spectral instability implies that the frequencies and damping rates computed through an idealized model will differ from those observed in nature. In the context of holography this suggests that the transport properties of a quantum many-body system are likely to differ significantly from those predicted by the holographic model.
    
    It is worth mentioning that in some cases spectral instability may be connected to background instability. A model with no QNFs in the upper half of the complex plane, \textit{i.e.}, with a stable background, could exhibit enough spectral instability so that a small perturbation could drive a QNF into the upper half of the complex plane. This would imply that when modelling a quantum many-body system through the aforementioned model, one may not even be able to accurately predict the equilibrium state. The small error made in the idealization of the system could lead to very different conclusions.\footnote{An example of this is the pseudospectrum of the Orr-Sommerfeld operator for Couette flow in hydrodynamics \cite{reddy1993pseudospectra,trefethen1993hydrodynamic}.}
    In holography this might become relevant in models with second order phase transitions such as the holographic superconductor \cite{Amado:2009ts,Gubser:2008px,Hartnoll:2008vx}. There, spectral instability might be of particular importance in the presence of superflow and for the validity of the Landau criterion \cite{Amado:2013aea,Gouteraux:2022kpo}. It is important to note, however, that we have not found evidence of this kind of instability for the scalar and transverse gauge fields studied in this paper.

    \section{Holographic Model}\label{section:Holographic Model}
    
    After presenting the general formalism, we now proceed to construct the specific fluctuations whose stability we are going to study. We seek to probe the spectral stability of a real scalar $\phi$ and of the transverse components of a $U(1)$ gauge field $A_M$ in a SAdS$_{4+1}$ background with compactified regular coordinates $\{\mathscr{t},\boldsymbol{\mathscr{x}},\rho\}$ and metric \eqref{eq:General Regular metric compactified}. Their respective actions are given by:
        \begin{equation}\label{eq:Action real scalar (bis)}
            \mathcal{S}[\phi]=-\frac{1}{2}\int d^{4+1}x\,\, \sqrt{-g} \left(\partial_M\phi\partial^M \phi+m^2 \phi^2 
            \right)
            \,,
        \end{equation}
        \begin{equation}\label{eq:Action gauge field}
            \mathcal{S}[A]=-\frac{1}{4}\int d^{4+1}x\,\,\sqrt{-g}\left(  F_{MN}F^{MN} %
            \right)\,,
        \end{equation}
    where $F_{MN}=\partial_M A_N-\partial_N A_M$ is the field strength tensor and $m$ is the scalar  mass, which we assume to be above the Breitenlohner-Freedman (BF) bound ($m^2l^2> -4$) \cite{Breitenlohner:1982bm}. This latter requirement ensures that the QNM is normalizable and that the dual QFT is unitary \cite{Aharony:1999ti}.

\subsection{Real Scalar}
\label{ssec:scalarmodel}
    The eigenvalue problem for a real scalar field with action \eqref{eq:Action real scalar (bis)} has already been presented in equations \eqref{eq:Eigenvalue problem real scalar 1}-\eqref{eq:Eigenvalue problem real scalar 3} of example \ref{subsubsection:Example: Scalar in SAdS5}; concluding that the relevant operator $L$, whose spectral stability we seek to study is:
    \begin{equation}
        L= i\begin{pmatrix}
        0&1 \\
        L_1\left[\partial_\mathscr{\rho}^2,\partial_\mathscr{\rho};\boldsymbol{\mathfrak{q}},\rho\right]&L_2\left[\partial_\rho;\boldsymbol{\mathfrak{q}},\rho\right]   
        \end{pmatrix}\,.
    \end{equation}

    With respect to the boundary conditions, we note that, as indicated in section \ref{subsection:Construction of the Eigenvalue Problem}, the outgoing behavior is automatically satisfied as long as we demand regularity on the horizon. On the other hand, normalizable boundary conditions have to be imposed manually, discarding the leading mode on the AdS boundary. In order to identify the aforementioned mode, we solve the equation of motion \eqref{eq:EOM real scalar} near the AdS boundary using a Frobenius series, obtaining the following solution:
    \begin{equation}\label{eq:Near boundary asymptotics scalar}
        \phi(\mathscr{t},\mathbf{k},\rho)=e^{-i\omega \mathscr{t}} \left\{ c_- (\rho-1)^{\Delta_-}(1+...) + c_+ (\rho-1)^{\Delta_+}(1+...)\right\}\,,
    \end{equation}
    with $\Delta_\pm=2\pm \sqrt{4+m^2l^2}$. Imposing normalizable boundary conditions corresponds to fixing $c_-=0$.
    
   Regarding the energy norm, we have the following expression for the inner product:\footnote{We have eliminated all constant prefactors from the energy norm \eqref{eq:scalar energy norm 1} as they play no role in the operator norm.}
    \begin{equation}\label{eq:scalar energy norm 1} 
        \expval{\phi_1,\phi_2}_E=\int_{\mathscr {t}=\text{const}}d\rho\,\, 
        \begin{pmatrix}
            \overline{\phi_1}(\mathscr{t},\mathbf{k},\rho) &\overline{\psi_1}(\mathscr{t},\mathbf{k},\rho)
        \end{pmatrix}
        \mathcal{G}\left[\overleftarrow{\partial}_\rho,\overrightarrow{\partial}_\rho;\boldsymbol{\mathfrak{q}},\rho\right]
        \begin{pmatrix}
            \phi_2(\mathscr{t},\mathbf{k},\rho)\\[\medskipamount]
            \psi_2(\mathscr{t},\mathbf{k},\rho)
        \end{pmatrix}\,,
    \end{equation}
    \begin{equation}\label{eq:scalar energy norm 2} 
        \mathcal{G}\left[\overleftarrow{\partial}_\rho,\overrightarrow{\partial}_\rho;\boldsymbol{\mathfrak{q}},\rho\right]=
        \begin{pmatrix}
            \frac{m^2l^2}{(1-\rho)^5}+
            \frac{\boldsymbol{\mathfrak{q}}^2}{(1-\rho)^3}+
        \overleftarrow{\partial}_\rho
        \frac{\mathscr{f}(\rho)}{(1-\rho)^3}
        \overrightarrow{\partial}_\rho&0\\[\medskipamount]
            0 & \frac{2-\mathscr{f}(\rho)}{(1-\rho)^3}
        \end{pmatrix}\,,
    \end{equation}
    derived from the subsequent energy-momentum tensor: 
    \begin{equation}
        T_{MN}=\nabla_M\phi \nabla_N \phi-\frac{1}{2}g_{MN}\left[\nabla_S\phi \nabla^S\phi+ m^2\phi^2\right] \,.
    \end{equation}

    In order to ensure that the energy norm \eqref{eq:scalar energy norm 1} is positive definite, while also enforcing the adequate asymptotic behavior for the QNMs, we chose to work in the space of regular functions with the following behavior on the AdS boundary:
    \begin{equation}\label{eq:function space scalar}
        \phi(\mathscr{t},\mathbf{k},\rho\rightarrow 1) = A(\mathscr{t},\mathbf{k})\, (\rho-1)^{n}\,, \qquad n>2\,.
    \end{equation}
    With this choice, for masses above the BF bound, the leading mode in \eqref{eq:Near boundary asymptotics scalar} is automatically discarded and the energy norm is positive definite \cite{Breitenlohner:1982bm}.
    It is worth pointing out that the chosen function space, besides being mathematically convenient, is also the physically relevant one as it contains all possible asymptotic behaviors for scalar QNMs with masses above the BF bound in the standard quantization.
    
     In view of the above, in order to select functions belonging to the chosen Hilbert space, it is convenient to work with the rescaled fields
    \begin{equation}\label{eq:rescaled scalar}
        \begin{pmatrix}
            \tilde{\phi}\\
            \tilde{\psi}
        \end{pmatrix}=(1-\rho)^{-2}\begin{pmatrix}
            \phi\\
            \psi
        \end{pmatrix}\,,
    \end{equation}
    for which imposing the asymptotic behavior \eqref{eq:function space scalar} amounts to fixing Dirichlet boundary conditions on the AdS boundary. 
    
    Lastly, note that, given the expression of the adjoint of $L$ with respect to the energy norm $L^\dagger$:\footnote{The explicit computation of $L^\dagger$ can be found in appendix \ref{subannex:Computation of Ldagger scalar}.}
    \begin{equation}\label{eq:Ldag Scalar}
        L^\dagger=L+\begin{pmatrix}
            0&0\\[\medskipamount]
            0&-2i\,\delta(\rho)\,\frac{\mathscr{f}(\rho)-1}{\mathscr{f}(\rho)-2}
        \end{pmatrix}\,,
    \end{equation}
    we can conclude that $L$ is non-normal and that this non-normality is associated with the existence of a horizon. This matches our initial physical intuition, the system is non-conservative because the fluctuations eventually fall into the horizon and die out.
    
    \subsection{Transverse Gauge Field}
    \label{ssec:gaugefield}

    The linearized equations of motion for a gauge field with action \eqref{eq:Action gauge field} are:
    \begin{equation}\label{eq:eoms Gauge Field}
        \nabla^M F_{MN}=0\,.
    \end{equation}

    Decomposing into Fourier modes $A_\mu(\mathscr{t},\mathbf{k},\rho)$ and assuming the momentum $\mathbf{k}$ to be oriented in the $\mathscr{x}_3$ direction, the equations of motion decouple into two sectors: transverse $\{A_1,A_2\}$ and longitudinal $\{A_\mathscr{t},A_3,A_\mathscr{r}\}$. The members of the transverse channel transform as vectors under the unbroken $O(2)$ symmetry, while the components of the longitudinal sector are invariants \cite{Kovtun:2005ev}.
    
    In the present work, we only study the transverse sector, \textit{i.e.}, we consider QNMs $a(\mathscr{t},\mathbf{k},\rho)$ satisfying the following eigenvalue problem:
    \begin{align}\label{eq:Eigenvalue problem Transverse Guage field 1}
        \mathfrak{w}\begin{pmatrix}
        a(\mathscr{t},\mathbf{k},\rho)\\
        \alpha(\mathscr{t},\mathbf{k},\rho)
        \end{pmatrix} &= i
        \begin{pmatrix}
        0&1 \\
        L_1\left[\partial_\mathscr{\rho}^2,\partial_\mathscr{\rho};\boldsymbol{\mathfrak{q}},\rho\right]&L_2\left[\partial_\rho;\boldsymbol{\mathfrak{q}},\rho\right]   
        \end{pmatrix}\begin{pmatrix}
        a(\mathscr{t},\mathbf{k},\rho)\\
        \alpha(\mathscr{t},\mathbf{k},\rho)
        \end{pmatrix}\,,\\[\medskipamount]\label{eq:Eigenvalue problem Transverse Guage field 2}
        L_1\left[\partial_\rho^2,\partial_\rho;\boldsymbol{\mathfrak{q}},\rho\right]&=\left[\mathscr{f}(\rho)-2\right]^{-1}\left[\boldsymbol{\mathfrak{q}}^2-(1-\rho)\left(\frac{\mathscr{f}(\rho) }{1-\rho}\right)'\partial_{\rho} -\mathscr{f}(\rho)\partial_{\rho}^2 \right]\,,\\[\medskipamount]\label{eq:Eigenvalue problem Transverse Guage field 3}
        L_2\left[\partial_\rho;\boldsymbol{\mathfrak{q}},\rho\right]&=\left[\mathscr{f}(\rho)-2\right]^{-1}\left[(1-\rho)\left(\frac{\mathscr{f}(\rho)-1}{(1-\rho)}\right)' +2 \left(\mathscr{f}(\rho)-1\right)\partial_{\rho}  \right]\,,
    \end{align}
    where we used again the dimensionless variables $\mathfrak{w}=z_h\omega$, $\boldsymbol{\mathfrak{q}}=z_h\mathbf{k}$ and introduced the auxiliary field $\alpha(\mathscr{t},\mathbf{k},\rho)=z_h\partial_\mathscr{t}a(\mathscr{t},\mathbf{k},\rho)=-i\mathfrak{w} a(\mathscr{t},\mathbf{k},\rho)$.
    
    It is important to stress that the transverse sector is gauge invariant, as we can eliminate the $\{\mathscr{x}_1,\mathscr{x}_2\}$ dependence of all functions given the symmetries of the problem. Explicitly, we have that the gauge transformation of the transverse gauge field is the following:
    \begin{equation}
        a\rightarrow a+d \chi=a
    \end{equation}
    where, as indicated, we disregard any possible dependence of $\mathscr{\chi}$ on $\{\mathscr{x}_1,\mathscr{x}_2\}$. This ensures that the pseudospectrum and the condition numbers only probe stability under gauge-invariant perturbations.
    In fact, the gauge invariant electric field is simply $E_{1,2} = i \mathfrak{w} A_{1,2}$.

    As with the scalar, outgoing boundary conditions are automatically satisfied demanding regularity on the horizon, while normalizable boundary conditions have to be imposed explicitly. Solving the equations of motion \eqref{eq:eoms Gauge Field} with a Frobenius series, we obtain the following near-AdS boundary behavior:
    \begin{equation}
        a(\mathscr{t},\mathbf{k},\rho)=e^{-i\omega \mathscr{t}}\left\{c_- (1+...)+c_+(1-\rho)^{2}(1+...)\right\}\,,
    \end{equation}
    and thus, imposing normalizable boundary conditions corresponds to fixing $c_-=0$.
    
    With regard to the energy norm, we recall that the energy-momentum tensor for a gauge field with action \eqref{eq:Action gauge field} is given by:
    \begin{equation}
        T_{MN}=F_{MS} \tensor{F}{_N^S}-\frac{1}{4}g_{MN}F_{SP} \tensor{F}{^S^P}\,,
    \end{equation}
    which yields the following inner product:
    \begin{equation}\label{eq:Transverse GF energy norm 1} 
        \expval{a_1,a_2}_E=\int_{\mathscr {t}=\text{const}}d\rho\,\, 
        \begin{pmatrix}
            \overline{a_1}(\mathscr{t},\mathbf{k},\rho) &\overline{\alpha_1}(\mathscr{t},\mathbf{k},\rho)
        \end{pmatrix}
        \mathcal{G}\left[\overleftarrow{\partial}_\rho,\overrightarrow{\partial}_\rho;\boldsymbol{\mathfrak{q}},\rho\right]
        \begin{pmatrix}
            a_2(\mathscr{t},\mathbf{k},\rho)\\[\medskipamount]
            \alpha_2(\mathscr{t},\mathbf{k},\rho)
        \end{pmatrix}\,,
    \end{equation}
    \begin{equation}\label{eq:Transverse GF energy norm 2} 
        \mathcal{G}\left[\overleftarrow{\partial}_\rho,\overrightarrow{\partial}_\rho;\boldsymbol{\mathfrak{q}},\rho\right]=\begin{pmatrix}
            \frac{\boldsymbol{\mathfrak{q}}^2}{1-\rho}+\overleftarrow{\partial}_\rho\frac{\mathscr{f}(\rho)}{1-\rho}\overrightarrow{\partial}_\rho&0\\[\medskipamount]
            0 & \frac{2-\mathscr{f}(\rho)}{1-\rho}
        \end{pmatrix}\,.
    \end{equation}

    In this case, as we have no mass term, the energy norm is always positive definite. Thus, we consider the function space comprised of functions satisfying outgoing boundary conditions on the horizon with the following behavior on the AdS boundary: 
    \begin{equation}\label{eq:function space GF}
        a(\mathscr{t},\mathbf{k},\rho\rightarrow 1) = A(\mathscr{t},\mathbf{k}) (\rho-1)^{n}\,, \qquad n>1\,.
        \end{equation}
    This ensures the desired asymptotic behavior for the QNMs while also guaranteeing that the norm is non-divergent. Similarly to the case of the scalar, we find it more convenient to work with the rescaled fields:
    \begin{equation}\label{eq:rescaled gauge field}
        \begin{pmatrix}
            \tilde{a}\\
            \tilde{\alpha}
        \end{pmatrix}=(1-\rho)^{-1}\begin{pmatrix}
            a\\
            \alpha
        \end{pmatrix}\,,
    \end{equation}
    for which imposing the asymptotic behavior \eqref{eq:function space GF} amounts to fixing Dirichlet boundary conditions on the AdS boundary. 
    
    To conclude, we note that, identically to what we found for the scalar field, $L$ is non-normal with respect to the energy norm. Its adjoint with respect to the aforementioned norm is given by
    \begin{equation}\label{eq:Ldag GF}
        L^\dagger=L+\begin{pmatrix}
            0&0\\[\medskipamount]
            0&-2i\delta(\rho)\frac{\mathscr{f}(\rho)-1}{\mathscr{f}(\rho)-2}
        \end{pmatrix}\,,
    \end{equation}
    which once again, relates the non-normality to the existence of a horizon.\footnote{The explicit computation of $L^\dagger$ can be found in appendix \ref{subannex:Computation of Ldagger GF}.}
        
    \section{Numerical method}\label{section:Numerical method}

    We approach the stability analysis numerically; discretizing the radial coordinate $\rho$ in a Chebyshev grid with points:
    \begin{equation}\label{eq:Chebyshev grid}
        \rho_j=\frac{1}{2}\left[1-\cos\left(\frac{j\pi}{N}\right)\right]\,,\qquad j=0,1,...,N\,,
    \end{equation}
    which in turn allows us to discretize the differential operators using the corresponding Chebyshev differentiation matrices \cite{Trefethen:2000}. This choice is equivalent to approximating the QNMs by a series of Chebyshev polynomials $T_n$ with $n=\{0,1,...,N\}$. As a rule of thumb, this allows us to approximate at most $N/2$ QNFs \cite{Trefethen:2000,boyd}.

    Regarding the energy norm we proceed as indicated in \cite{Jaramillo:2020tuu} and construct a $G_E$ matrix defined as the discretized version of the original energy norm. Labelling by $u$ the rescaled scalar doublet $\left(\tilde{\phi},\,\tilde{\psi}\right)^T$ and the rescaled gauge field doublet $\left(\tilde{a},\,\tilde{\alpha}
    \right)^T$,  we construct $G_E$ such that:
    \begin{eqnarray}
        \lim_{N\rightarrow\infty}u_N^* G_E u_N=\expval{u,u}_E\,,
    \end{eqnarray}
    where $u_N$ is the vector arising form the discretization of $u$.\footnote{There is an important subtlety in the construction of $G_E$ that should be addressed. One needs to construct $G_E$ in a grid of at least twice the size of the original one and, at the end, interpolate back. This ensures consistency with the discretization process: when working in a grid with $N+1$ points the maximum resolution is given by polynomials of degree $N+1$. Thus, for the discretized norm to be exact for polynomials of such degree, one needs to construct it on a grid with at least $2(N+1)$ points. A detailed discussion on the construction of the $G_E$ matrix can be found in appendix~\ref{app:Extra Details on the Discretization.}.}

    To conclude the discretization process, we need to numerically impose the function spaces introduced in the previous section. Note that, in any numerical method, we can only get regular solutions. Consequently, the discretization immediately selects the space of regular functions, and thus we only need to ensure the adequate behavior on the AdS boundary. This corresponds to imposing Dirichlet boundary conditions for the rescaled scalar and gauge fields, which can be achieved by removing the rows and columns corresponding to the AdS boundary from all discretized operators, including the $G_E$ matrix.\footnote{This is equivalent to reducing the space in which the matrices act to that of vectors vanishing on the AdS boundary.}

    With the original system fully discretized, we can proceed to study the spectral stability. We use \textit{Wolfram Engine} to compute condition numbers and pseudospectra as indicated in theorem \ref{th:psuedospectrum in different norms}. It is particularly convenient to rewrite the problem in terms of the matrix $\ell^2$-norm as it reduces the computation of the pseudospectrum to obtaining the smallest eigenvalue of a Hermitian matrix; which we locate using Arnoldi iteration (see \textit{e.g.} ch. 28 of \cite{Trefethen:2005}). With this procedure, we achieve $\mathcal{O}\left(N^2\right)$ runtime for each point $z$ in the complex plane where we compute the norm of the resolvent.

    In order to gain greater insight into the nature of the (in)stability, we also explore the selective pseudospectra associated with local potential perturbations to the original equations of motion. Concretely, we consider perturbations to equations \eqref{eq:EOM real scalar} and \eqref{eq:eoms Gauge Field} of the form:
    \begin{equation}\label{eq:eoms perturbed}
       \left[-\nabla_M \nabla^M + m^2 + \frac{V(\rho)}{l^2}\right]\phi=0\,,\qquad \nabla^M F_{MN}-\frac{V(\rho)}{l^2}A_N=0\,,
    \end{equation}
    where, in order to preserve the asymptotic behavior on the AdS boundary, we choose potentials vanishing on the aforementioned boundary $V(1)=0$. Note that the potential term added to the gauge field dynamics does, in general, break gauge invariance. Nonetheless, as a transverse gauge field $A_M$ is gauge invariant, the chosen potential term preserves the gauge symmetry.

    The main difference between full and selective pseudospectra lies in the kind of stability each of them probes. The former explores the stability under generic bounded perturbations, while the latter only considers local perturbations. Physically, these local perturbations can be interpreted as effective interactions arising from small deviations from the perfect SAdS$_{4+1}$ background. Therefore, probing the stability under them is akin to determining the underlying model dependence of the system.\footnote{In general, we consider that corrections arising from deformations of the perfect SAdS$_{4+1}$ background (associated, for instance, with the existence of extra fields or non-minimal couplings in the action) appear as non-local interactions. Nonetheless, we assume that the leading order behavior can be modelled with a local potential (a very common picture in effective field theories). Under this interpretation, and recalling that we have no control over the nature of the perturbations entering the full pseudospectrum, we consider the selective pseudospectrum to be slightly more physically motivated.}
           
    The selective pseudospectrum is computed using definition \ref{def:Pseudo Definition 2} with randomly generated potential perturbations constructed as diagonal matrices with random entries and normalized to a given size. This analysis is complemented with the computation of the QNFs for the perturbed system with the following deterministic potentials:
 \begin{subequations}
 \begin{align}
        V_1(\rho)&=A_1(1-\rho)\cos(2\pi\rho)\label{eq:determinisiticV1}\,,\\
        V_2(\rho)&=A_2(1-\rho)\cos(90\pi\rho)\label{eq:determinisiticV2}\,, \\
        V_3(\rho)&=A_3(1-\rho)\left\{1-\tanh\left[20\rho\right]\right\}\label{eq:determinisiticV3}\,,\\
        V_4(\rho)&=A_4(1-\rho)\left\{1-\tanh\left[20(1-\rho)\right]\right\}\label{eq:determinisiticV4}\,,
\end{align}\label{eq:deterministicpots}%
\end{subequations}  
    which shed light on a few interesting regimes. With $V_1$ and $V_2$, we probe the effect of long and short $\rho$-wavelength (wavelength in the $\rho$ direction) perturbations; while with $V_3$ and $V_4$, we analyze the stability under localized perturbations near the horizon (IR of the dual QFT) and the boundary (UV of the dual QFT). The $\{A_i\}$ are normalization constants to fix the magnitude of the perturbation.
    We plot these potentials in figure~\ref{fig:Deterministic Potentials plot}.
    \begin{figure}[htb!]
        \centering
            \includegraphics[width=0.6\linewidth]{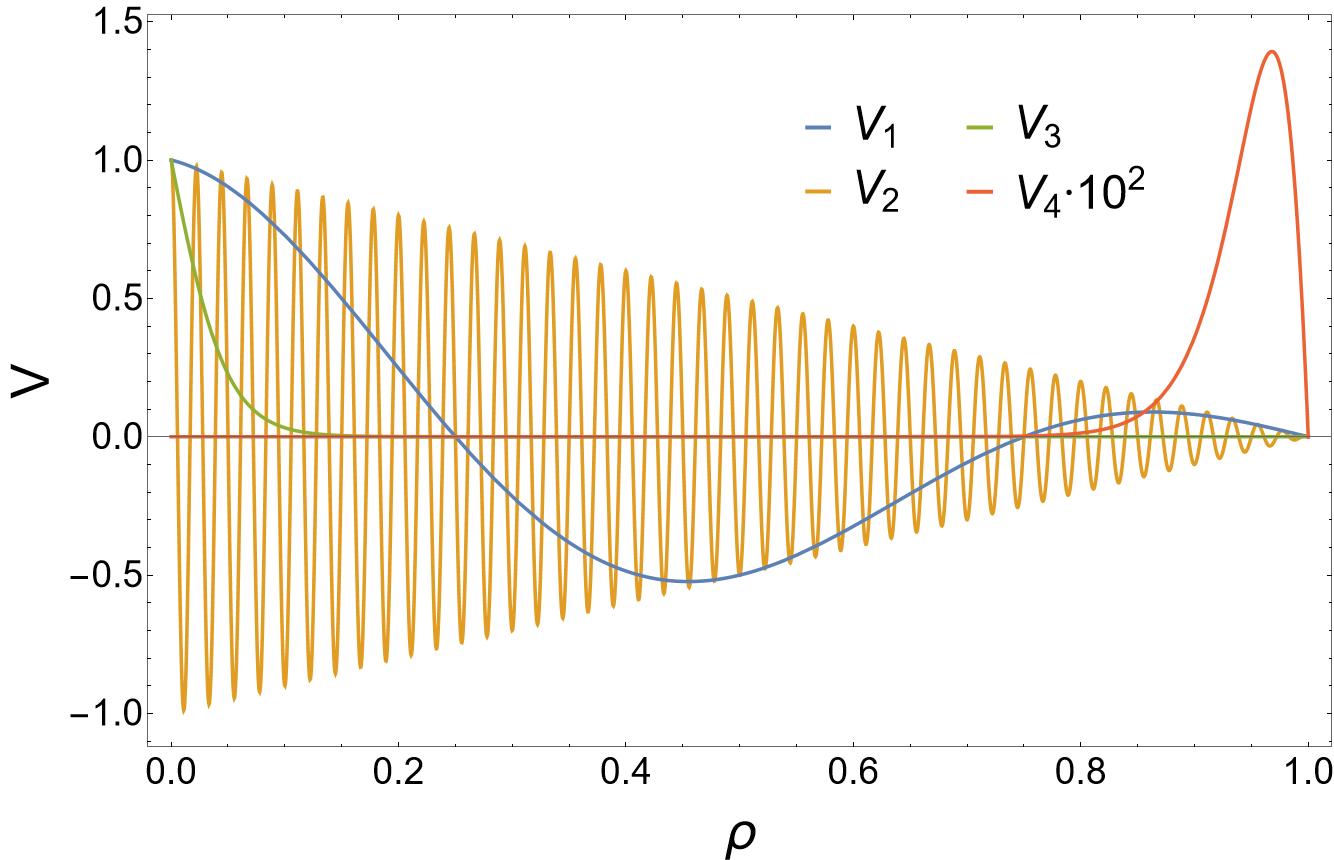}
          \caption{Deterministic potentials
          \eqref{eq:deterministicpots}
          with $A_i=1$.
          Recall that the horizon is at $\rho=0$ and the boundary at $\rho=1$.}
            \label{fig:Deterministic Potentials plot}
    \end{figure}

   The grid establishes a cutoff in the $\rho$-wavelength of the perturbations. Thus, as a byproduct of the discretization, we are only truly sensitive to the stability under perturbations above the cutoff.

    Finally, to stress the nontrivial nature of the norm, we also consider the pseudospectrum in the $L^2$-norm
    \begin{equation}
        \norm{u}_{L^2}=\int d\rho\,\, u^*(\mathscr{t},\mathbf{k},\rho)u(\mathscr{t},\mathbf{k},\rho)\,.
    \end{equation}
    
    Although  the physical relevance of this norm for the present problem is less clear, it might be informative to compare the general features of the pseudospectra in the energy norm to those in the $L^2$-norm. We direct the interested reader to Appendix~\ref{app:l2norm} where we present the pseudospectrum analysis in the $L^2$-norm.

    \section{Results}\label{section:Results}
    Here, we present the results of the analysis described in the previous sections. Our numerical simulations are performed in a grid of 120 points with a precision $5\times$MachinePrecision.
    
        \begin{figure}[h]
        \centering
        \begin{subfigure}[b]{0.46\linewidth}
            \centering
            \includegraphics[width=.935\linewidth]{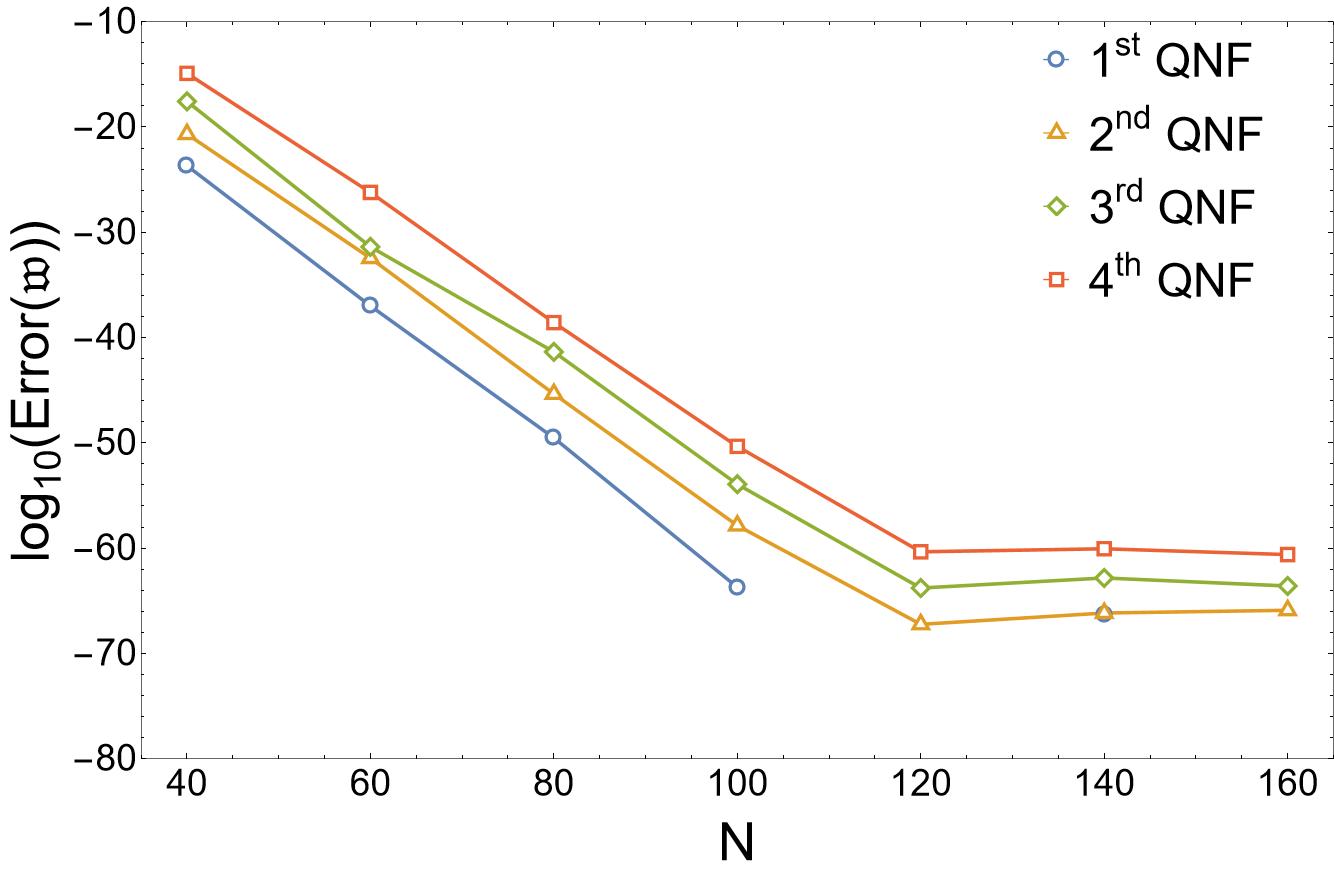}
            \captionsetup{justification=centering}
            \caption{Scalar with $m^2l^2=0$ and $\boldsymbol{\mathfrak{q}}^2=0$.}
            \label{fig:ErrorMasslessScalar5}
        \end{subfigure}\hfill
        \begin{subfigure}[b]{0.46\linewidth}
            \centering
            \includegraphics[width=.935\linewidth]{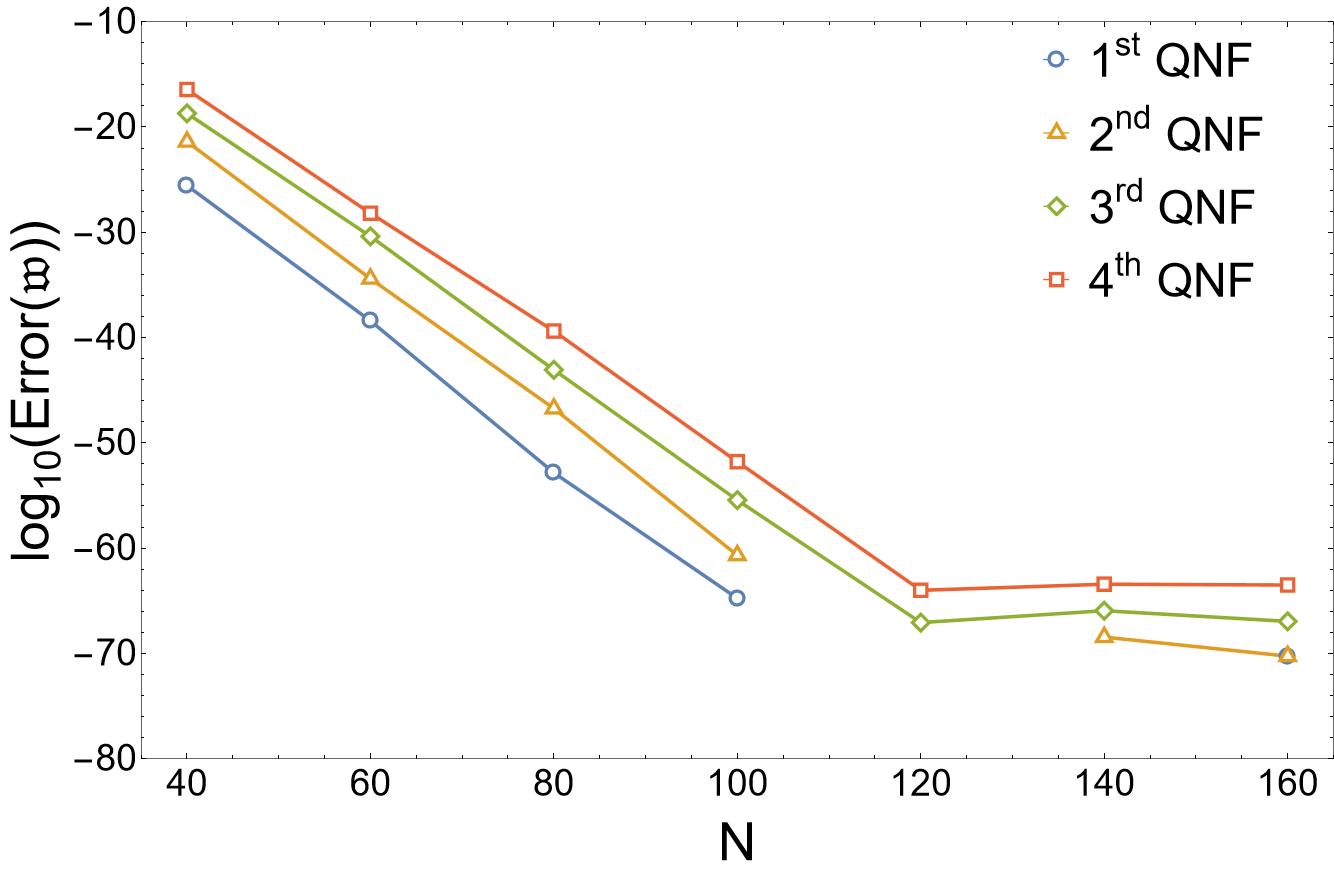}
            \captionsetup{justification=centering}
            \caption{Transverse Gauge Field with $\boldsymbol{\mathfrak{q}}^2=0$.}
            \label{fig:ErrorTransvGaugeField}
        \end{subfigure}
        \caption{Convergence test for QNFs computed with $5\times$MachinePrecision. The error is defined as $\left|1-\frac{|\mathfrak{w}|}{|\mathfrak{w}_{\text{ref}}|}\right|$, with $\mathfrak{w}_{\text{ref}}$ the reference value obtained with $10\times$MachinePrecision in grid of $400$ points. %
        The missing points are such that their error is below the numerical accuracy.}
        \label{fig:Error}
    \end{figure}

    \begin{figure}[h]
        \centering
        \begin{subfigure}[b]{0.46\textwidth}
            \centering
            \includegraphics[width=.935\linewidth]{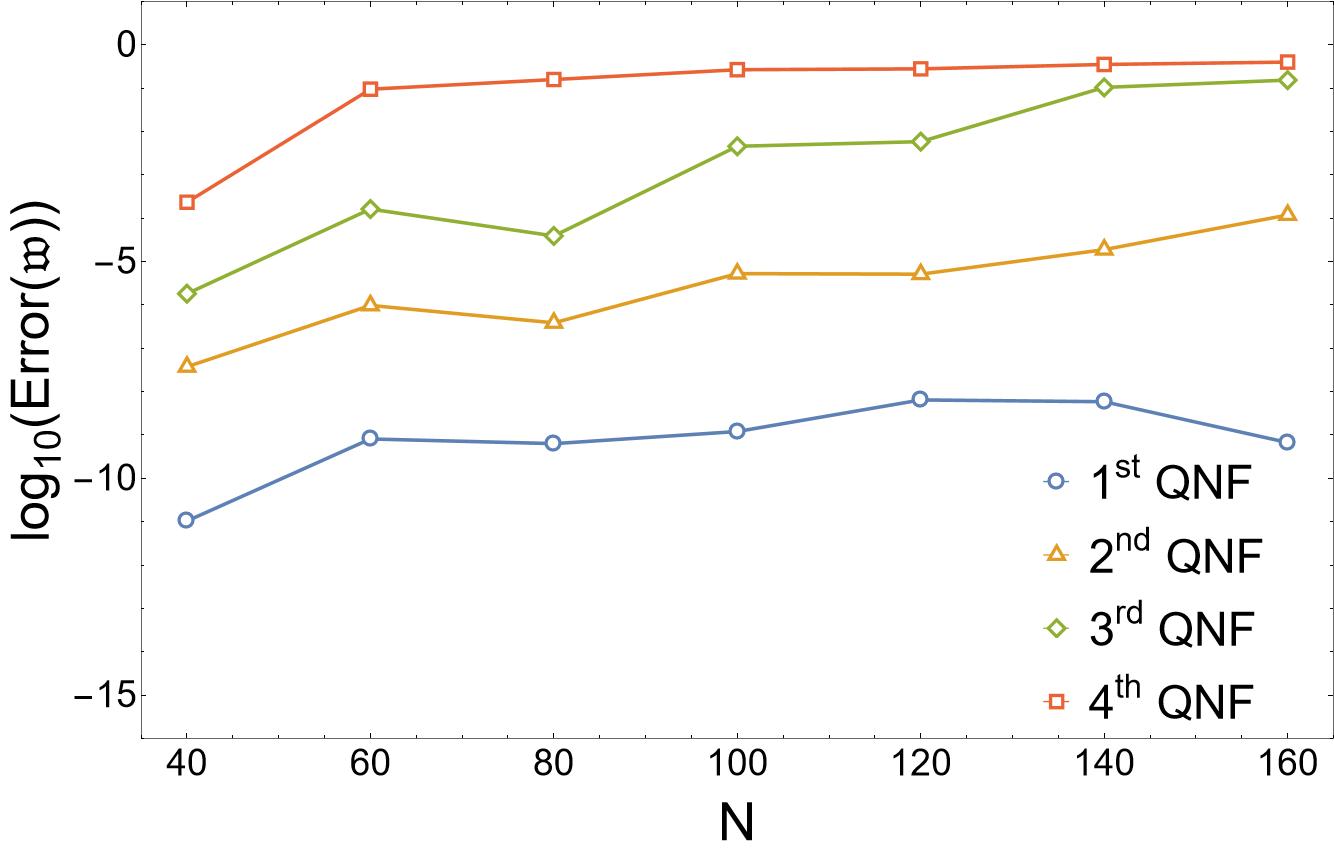}
            \captionsetup{justification=centering}
            \caption{Scalar with $m^2l^2=0$ and $\boldsymbol{\mathfrak{q}}^2=0$.}
            \label{fig:ErrorMasslessScalarMP}
        \end{subfigure}\hfill
        \begin{subfigure}[b]{0.46\textwidth}
            \centering
            \includegraphics[width=.935\linewidth]{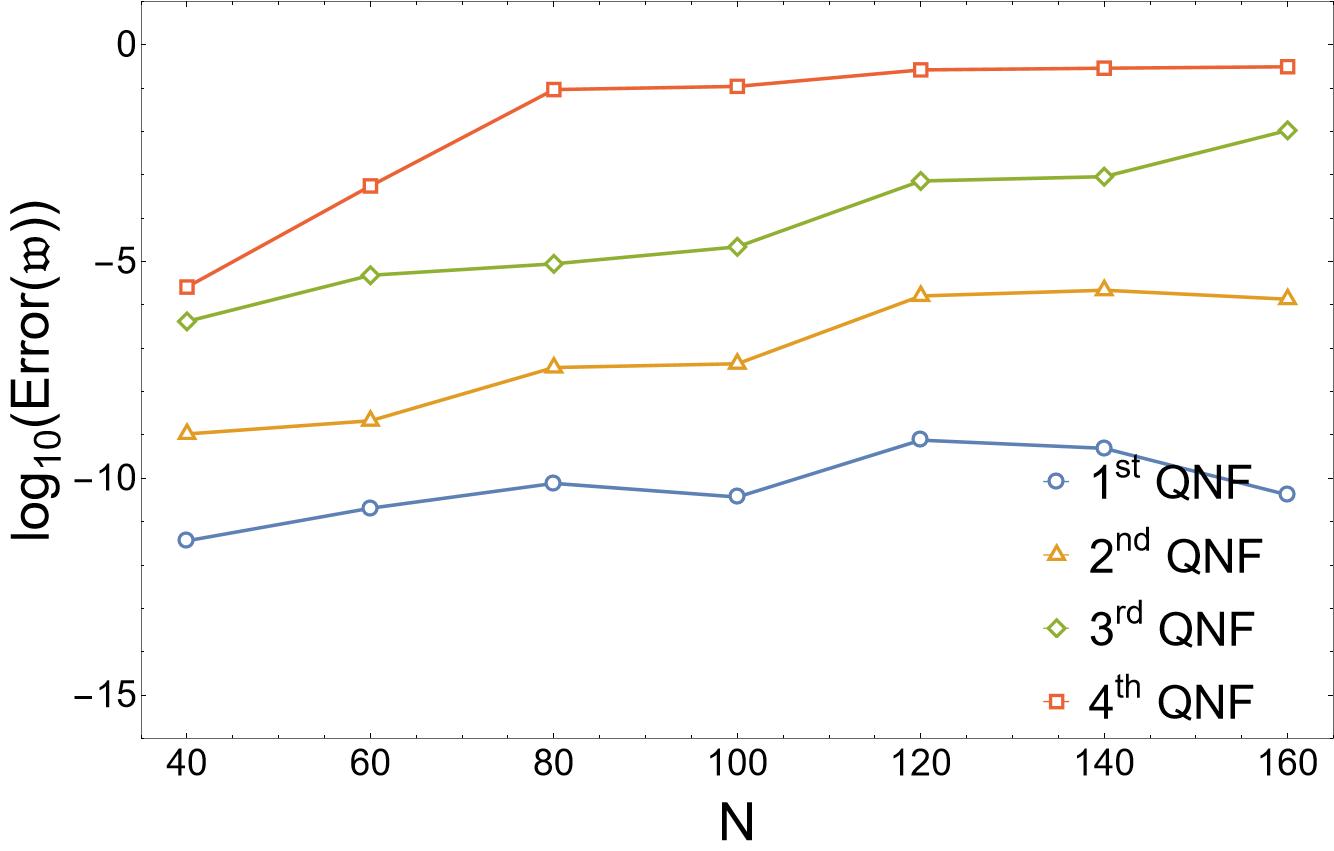}
            \captionsetup{justification=centering}
            \caption{Transverse Gauge Field with $\boldsymbol{\mathfrak{q}}^2=0$.}
            \label{fig:ErrorTransvGaugeFieldMP}
        \end{subfigure}
        \caption{Convergence test for QNFs computed with MachinePrecision. The error is defined as $\left|1-\frac{|\mathfrak{w}|}{|\mathfrak{w}_{\text{ref}}|}\right|$, with $\mathfrak{w}_{\text{ref}}$ the reference value obtained with $10\times$MachinePrecision in grid of $400$ points. The large error of the fourth QNF hints towards spectral instability.}
        \label{fig:ErrorMP}
    \end{figure}

   We require large grids and high precision to ensure that the numerical procedure does not heavily affect our results. Physically, the discretization and numerical round-offs can be understood as perturbations to the original problem and, as we want to analyze the stability of the latter, we need to be especially careful and avoid undesired effects associated with them. 
    As we are mainly interested in the stability of the three or four lowest-lying QNFs, in figure \ref{fig:Error} we test the validity of our numerics by computing those QNFs in our setup and comparing them to the ones obtained using $10\times$MachinePrecision in a grid of 400 points.\footnote{Henceforth, we abuse notation and denote by QNF the pair $\{\mathfrak{w},-\mathfrak{w}^*\}$. As the equations of motion are invariant under simultaneous conjugation and multiplication by $-1$, the spectrum (and pseudospectrum) is symmetric under reflections along the imaginary axis and consequently the eigenvalues always appear in pairs $\{\mathfrak{w},-\mathfrak{w}^*\}$. Furthermore, we choose to order the QNFs by their imaginary part such the first QNF is the one with the largest imaginary part.} Thus, we conclude that with our choice of parameters we are ensuring that the effect of the numerics on the first four QNFs is smaller than $10^{-50}$ in all cases. It is interesting to note that the need for such large grids and precisions is a good indicator of the underlying instability of the problem. Small perturbations associated with the numerics have effects much larger than their typical scale (see figure \ref{fig:ErrorMP}).

    \subsection{Real Scalar in \texorpdfstring{$\text{SAdS}_{4+1}$}{SAdS4+1}}

    \begin{figure}[h]
        \centering
        \begin{subfigure}[b]{0.49\linewidth}
            \centering
            \includegraphics[width=\linewidth]{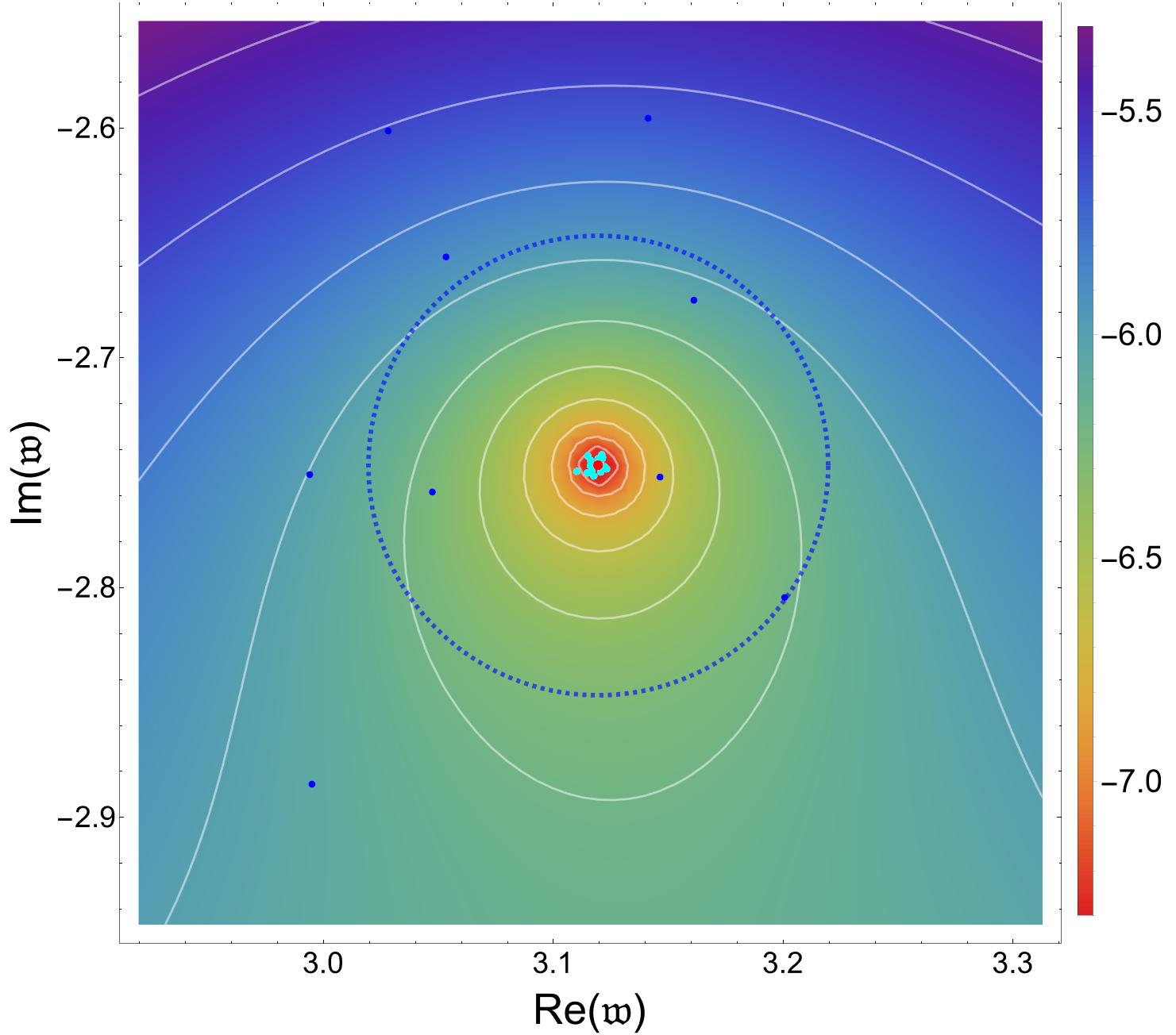}
            \captionsetup{justification=centering}
            \caption{$m^2l^2=0$, $\mathfrak{q}=0$.}
            \label{fig:CloseupPseudoMasslessScalarq0}
        \end{subfigure}\hfill
        \begin{subfigure}[b]{0.49\linewidth}
            \centering
            \includegraphics[width=\linewidth]{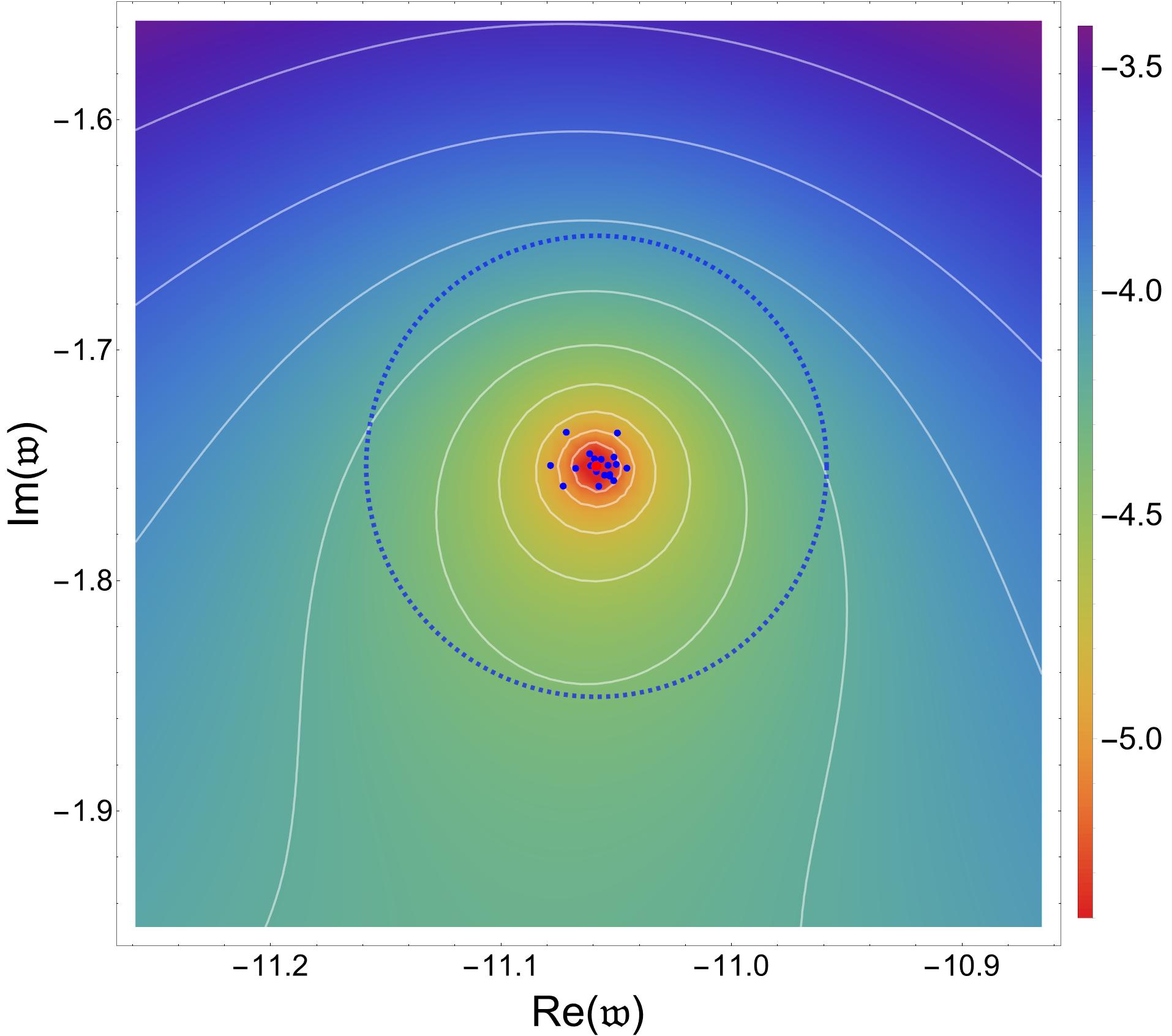}
            \captionsetup{justification=centering}
            \caption{$m^2l^2=0$, $\mathfrak{q}=10$.}
            \label{fig:CloseupPseudoMasslessScalarq10}
        \end{subfigure}
        \begin{subfigure}[b]{0.49\linewidth}
            \centering
            \includegraphics[width=\linewidth]{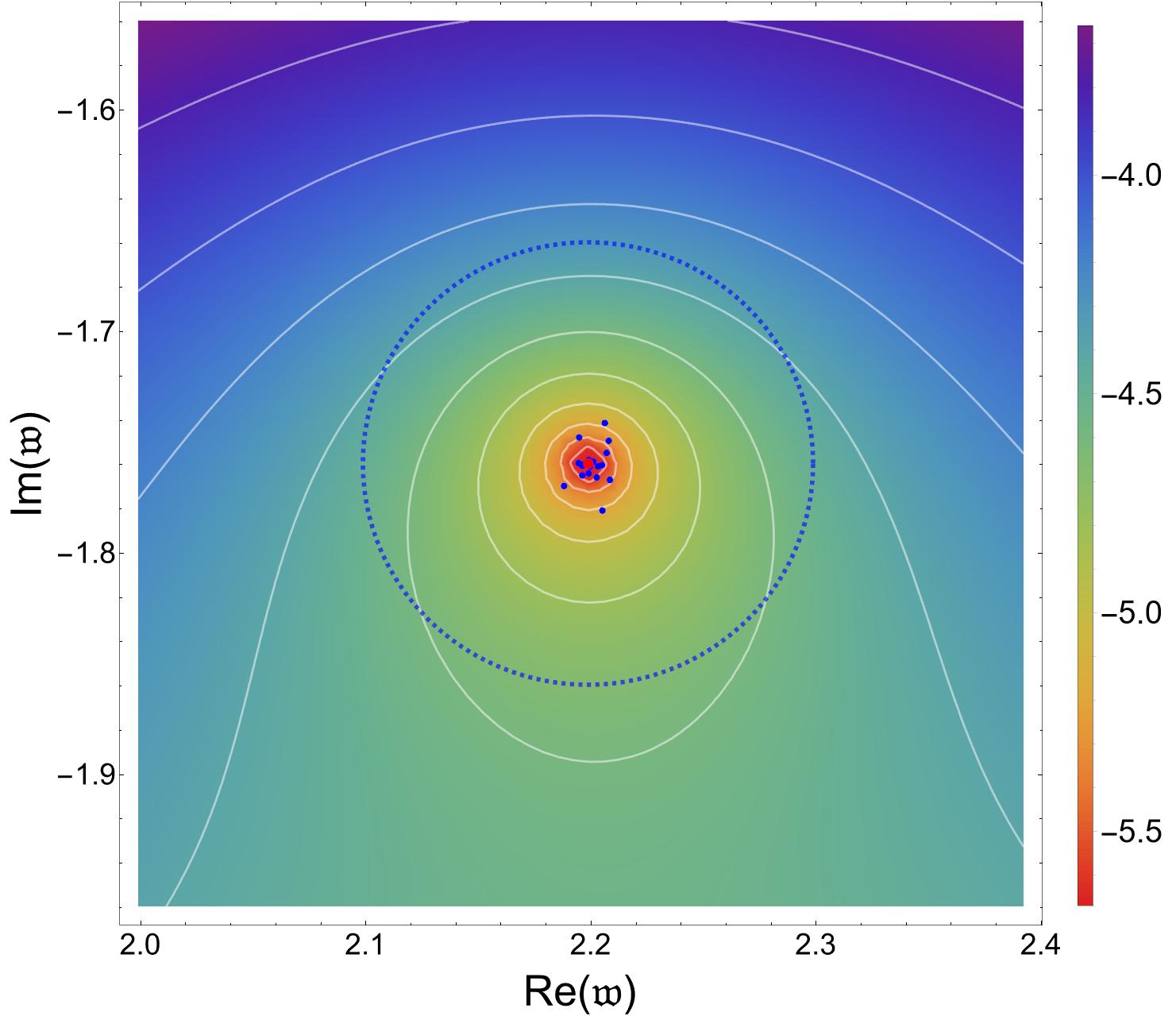}
            \captionsetup{justification=centering}
            \caption{$m^2l^2=-3$, $\mathfrak{q}=0$.}
            \label{fig:CloseupPseudoMassiveScalarq0}
        \end{subfigure}\hfill
        \begin{subfigure}[b]{0.49\linewidth}
            \centering
            \includegraphics[width=\linewidth]{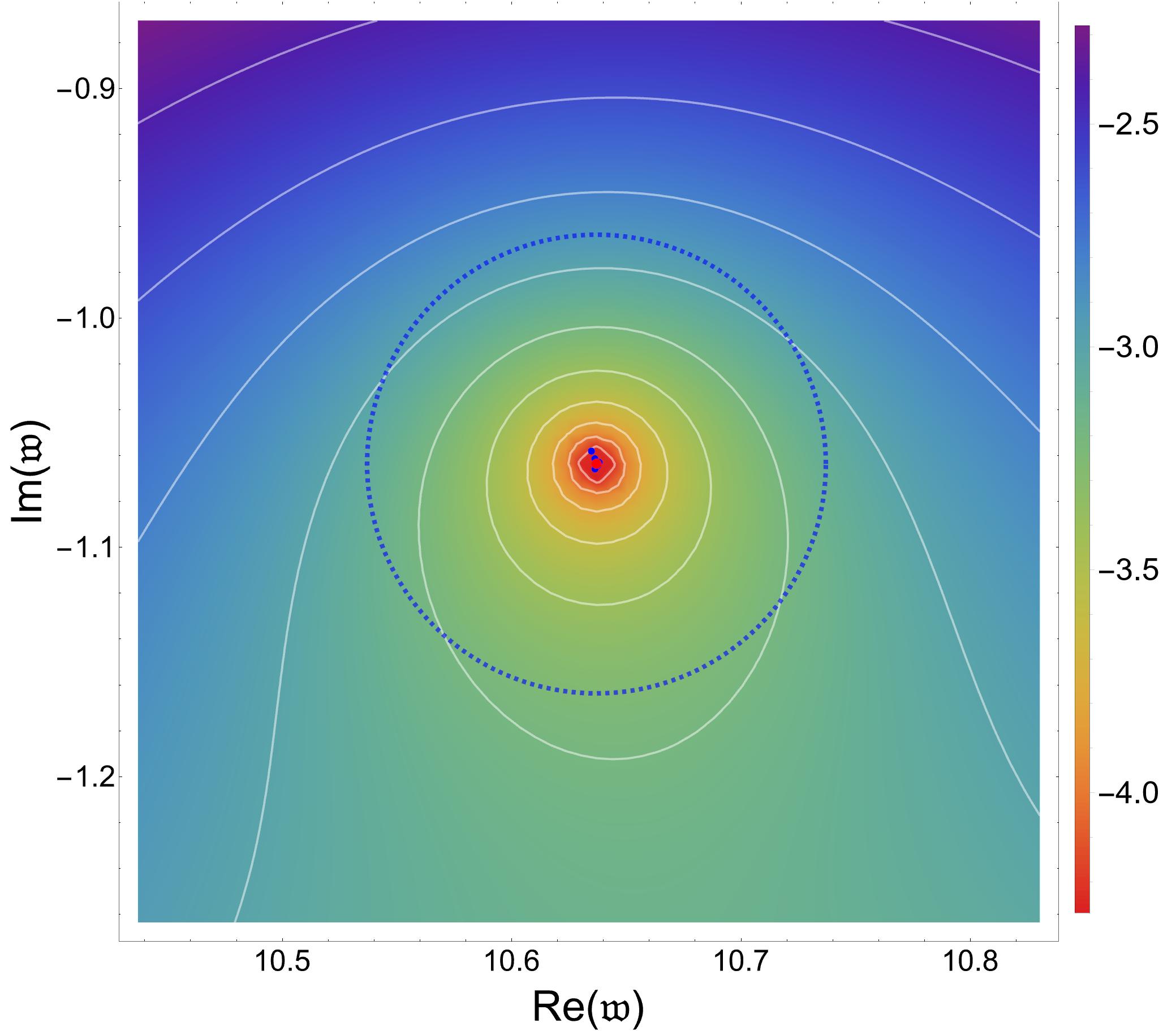}
            \captionsetup{justification=centering}
            \caption{$m^2l^2=-3$, $\mathfrak{q}=10$.}
            \label{fig:CloseupPseudoMassiveScalarq10}
        \end{subfigure}
        \caption{Close-up of the scalar  pseudospectrum in the energy norm around the first QNF for different values of $\mathfrak{q}$ and $m^2l^2$. The red dot corresponds to the QNF, the white lines represent the boundaries of various full $\varepsilon$-pseudospectra, and the dashed blue circle symbolizes a circle with a radius of $10^{-1}$ centered on the QNF. The heat map corresponds to the logarithm in base 10 of the inverse of the resolvent, while the blue and cyan dots indicate selective $\varepsilon$-pseudospectra computed with random local potential perturbations of size $10^{-1}$ and $10^{-3}$, respectively. Remarkably, only in (a) do we observe instability of the first QNF under local potential perturbations.}
        \label{fig:CloseupPseudospectraScalar}
    \end{figure}

  In figures \ref{fig:CloseupPseudospectraScalar} and \ref{fig:LargePseudospectraScalar} we present full and selective pseudospectra and the corresponding condition numbers in the energy norm for different values of $m^2l^2$ and $\mathfrak{q}=\sqrt{\boldsymbol{\mathfrak{q}}^2}$.
  Recall that throughout this section we are working in units of $z_h=(\pi\,T)^{-1}$.
  The $\varepsilon$-pseudospectra exhibit extended open regions denoting instability.\footnote{Note that all QNFs are contained within some region of the $\varepsilon$-pseudospectra. However, due to limited resolution in the plots, not all regions of the $\varepsilon$-pseudospectrum are observable. For example, in the full pseudospectrum shown in figure \ref{fig:LargePseudoMasslessScalarq0}, we cannot appreciate the $10^{-7}$-pseudospectrum around the first QNF observable in figure \ref{fig:CloseupPseudoMasslessScalarq0}. The same limitation applies to the selective pseudospectra, where small regions are covered by the red dots representing the QNFs (\textit{e.g.}, the selective $10^{-3}$-pseudospectrum observed around the first QNF in figure \ref{fig:CloseupPseudoMasslessScalarq0} is covered by a red dot in figure \ref{fig:LargePseudoMasslessScalarq0}).}
As observed in asymptotically flat \cite{Jaramillo:2020tuu,Destounis:2021lum} and de Sitter spacetimes \cite{Sarkar:2023rhp}, the instability increases the further away the QNFs are from the real axis.
In the dual quantum field theory this implies that excitations are increasingly spectrally unstable the more short-lived they are.

Interestingly, we find that one needs perturbations of size $\sim 10^{-0.4}\sim0.4$ to drive the QNFs to the upper half of the complex plane, as indicated by the full pseudospectra in figure \ref{fig:LargePseudospectraScalar}. In our units the typical distance between eigenvalues is $\sim2$, this implies that all backgrounds arising as a small deviation from SAdS$_{4+1}$ are stable, \textit{i.e.}, they do not have exponentially growing QNMs. 

\begin{figure}[htb!]
        \centering
        \begin{subfigure}[b]{0.46\linewidth}
            \centering
            \includegraphics[width=\linewidth]{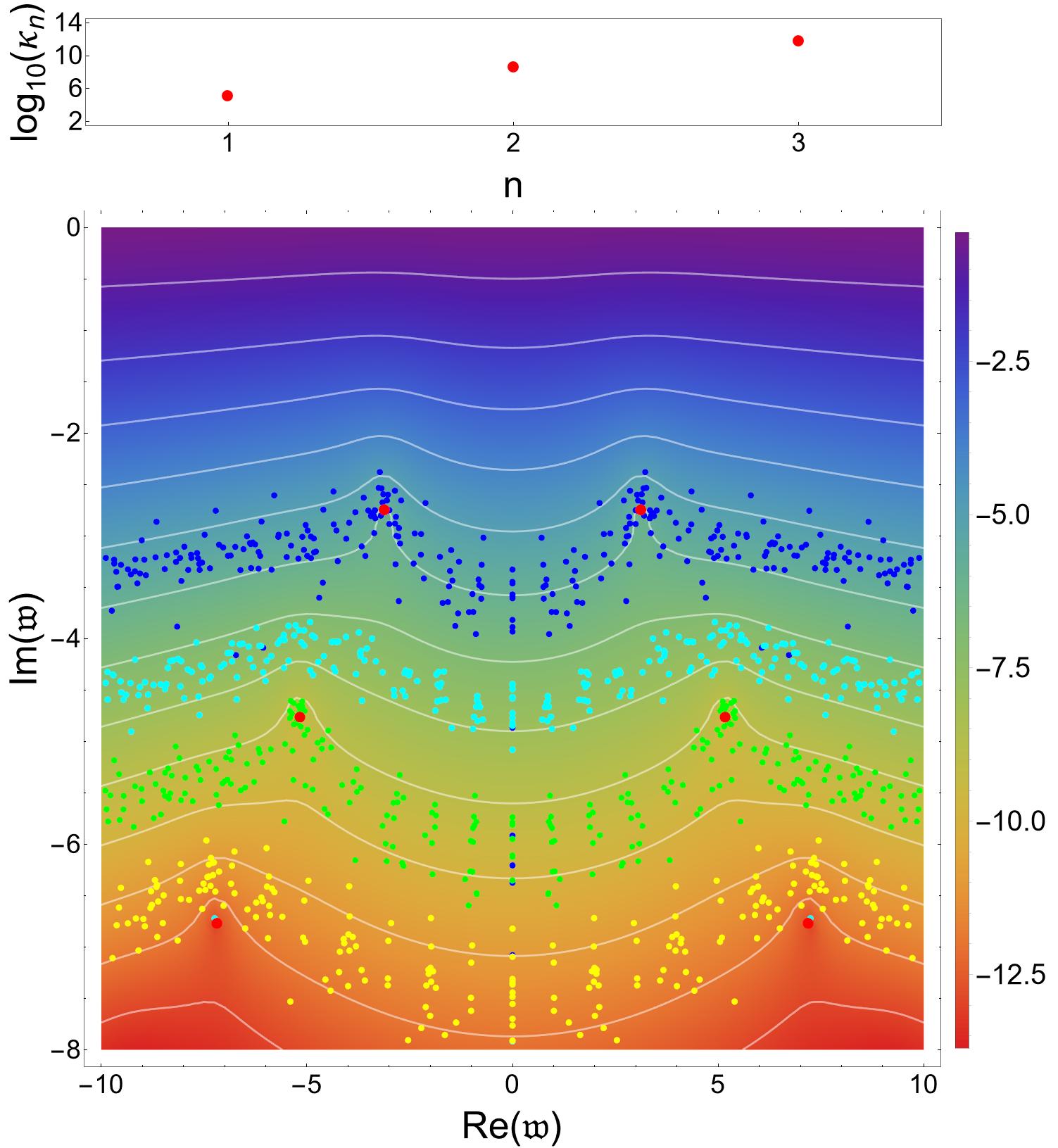}
            \captionsetup{justification=centering}
            \caption{$m^2l^2=0$, $\mathfrak{q}=0$.}
            \label{fig:LargePseudoMasslessScalarq0}
        \end{subfigure}\hfill
        \begin{subfigure}[b]{0.46\linewidth}
            \centering
            \includegraphics[width=\linewidth]{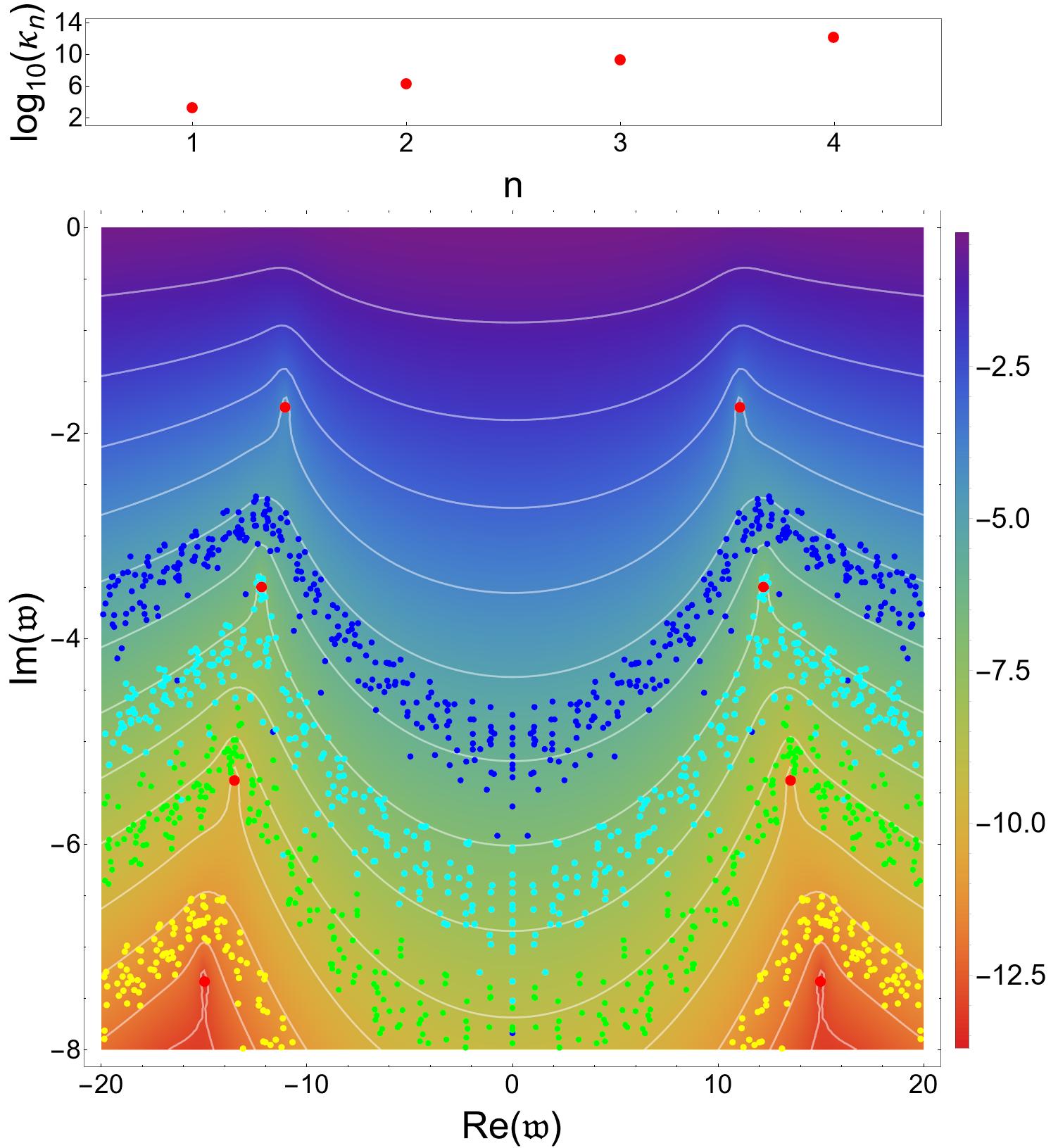}
            \captionsetup{justification=centering}
            \caption{$m^2l^2=0$, $\mathfrak{q}=10$.}
            \label{fig:LargePseudoMasslessScalarq10}
        \end{subfigure}
        \begin{subfigure}[b]{0.46\linewidth}
            \centering
            \includegraphics[width=\linewidth]{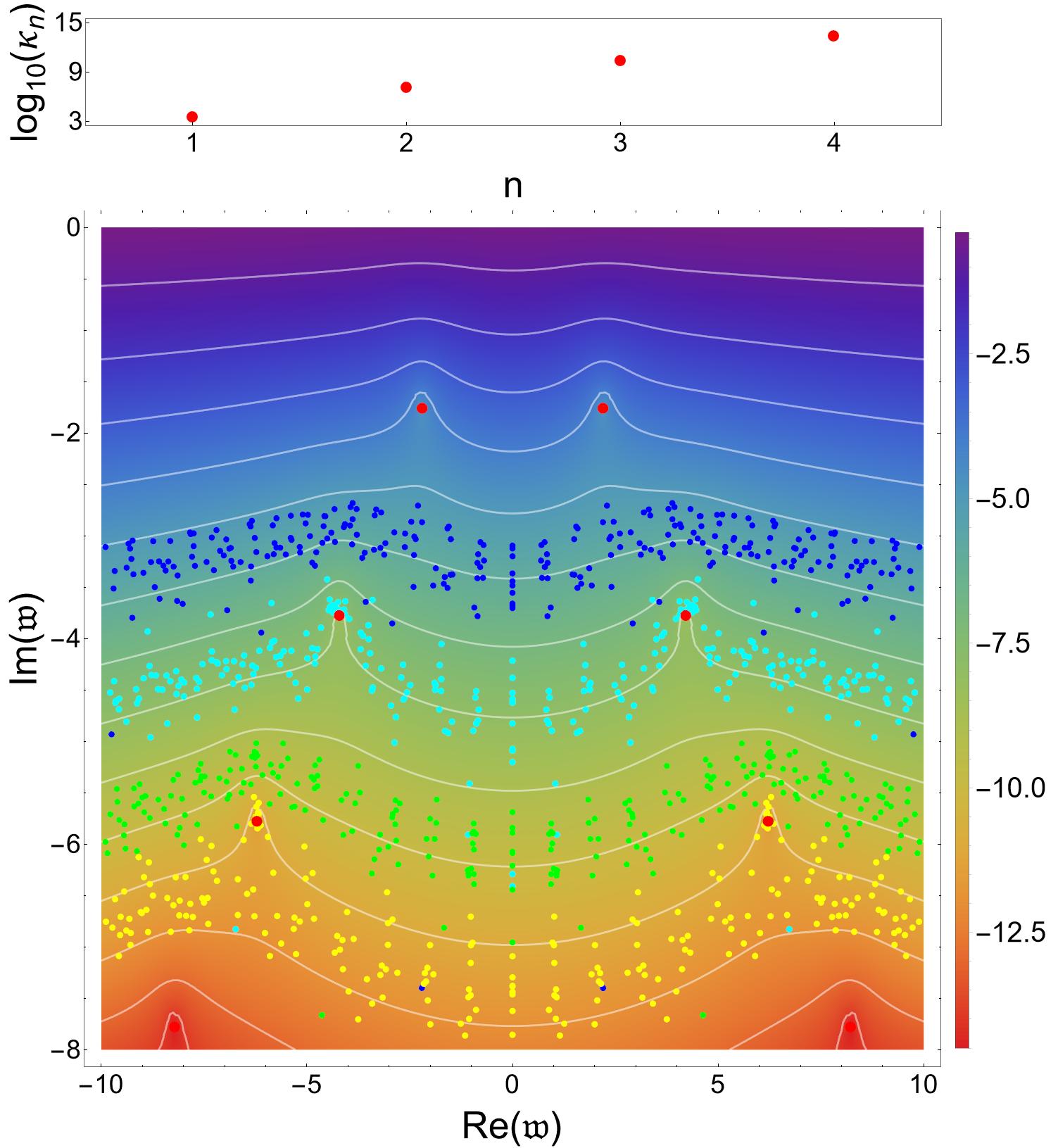}
            \captionsetup{justification=centering}
            \caption{$m^2l^2=-3$, $\mathfrak{q}=0$.}
            \label{fig:LargePseudoMassiveScalarq0}
        \end{subfigure}\hfill
        \begin{subfigure}[b]{0.46\linewidth}
            \centering
            \includegraphics[width=\linewidth]{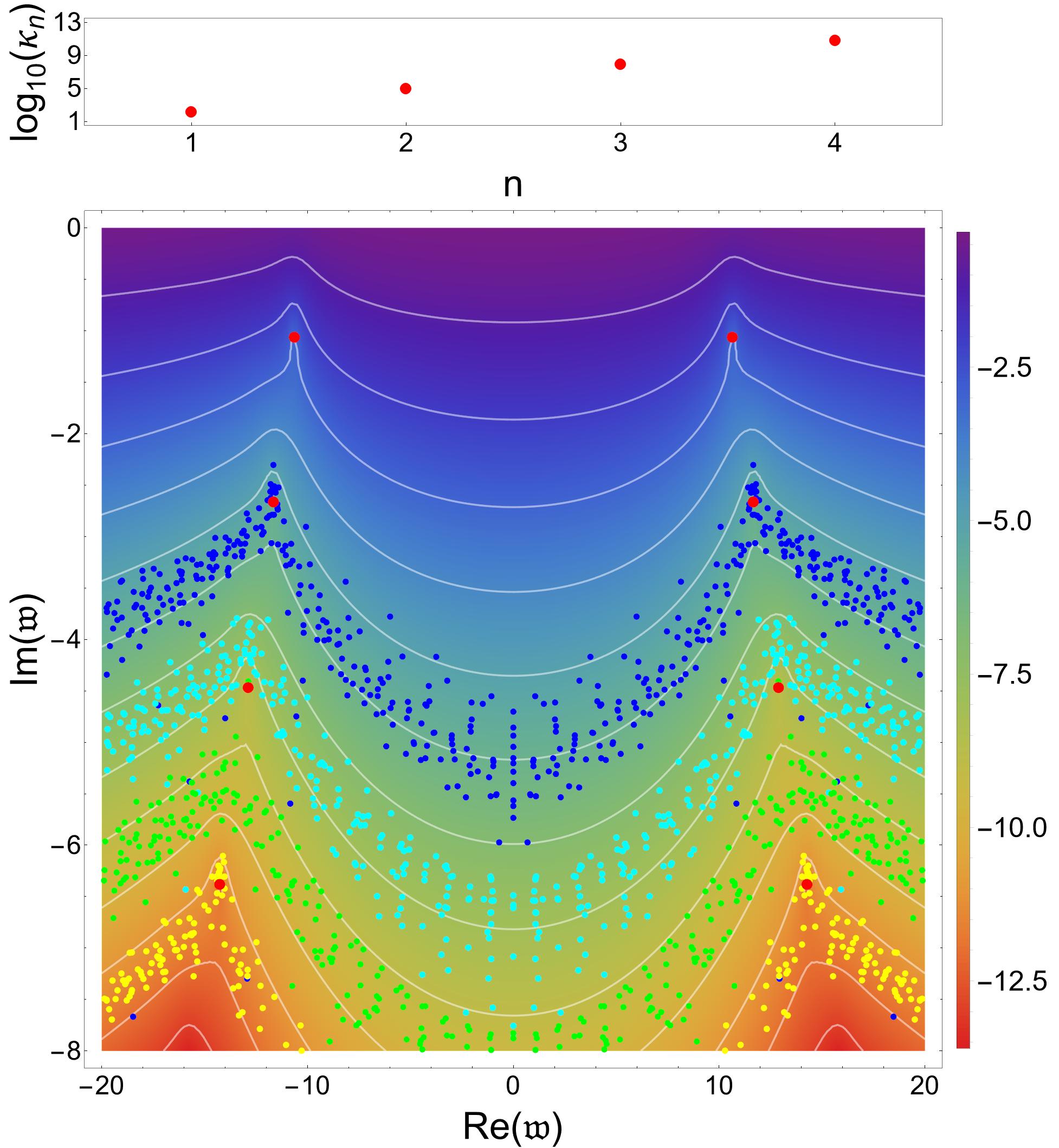}
            \captionsetup{justification=centering}
            \caption{$m^2l^2=-3$, $\mathfrak{q}=10$.}
            \label{fig:LargePseudoMassiveScalarq10}
        \end{subfigure}
        \caption{Scalar  pseudospectrum in the energy norm for different values of $\mathfrak{q}$ and $m^2l^2$. In the lower panels, we present selective and full pseudospectra. The red dots represent the (unperturbed) QNFs in the typical ``Christmas Tree'' configuration.
        The white lines denote the boundaries of different full $\varepsilon$-pseudospectra. The heat map corresponds to the logarithm in base 10 of the inverse of the resolvent, while the blue, cyan, green, and yellow dots indicate different selective $\varepsilon$-pseudospectra computed with random local potential perturbations of size $10^{-1}$, $10^{-3}$, $10^{-5}$, and $10^{-7}$; respectively. In the upper panels, we represent the condition numbers. Most notably, for small values of $\varepsilon$, the full $\varepsilon$-pseudospectra present open regions containing multiple QNFs, which signals spectral instability.}
        \label{fig:LargePseudospectraScalar}
    \end{figure}

           \begin{figure}[htb!]
        \centering
        \begin{subfigure}[b]{0.48\linewidth}
            \centering
            \includegraphics[width=.9\linewidth]{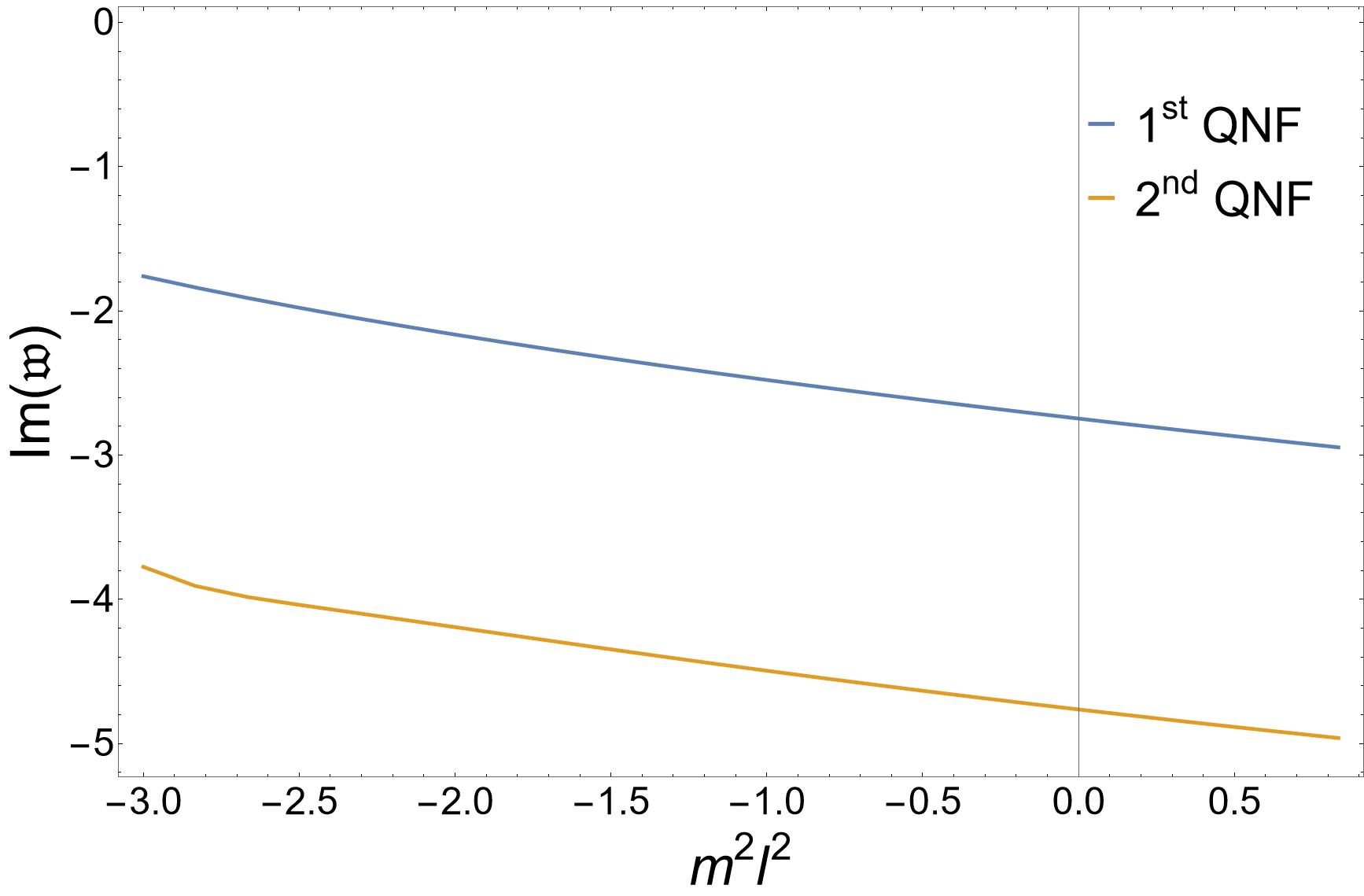}
            \captionsetup{justification=centering}
            \caption{Imaginary part of the QNFs for $\mathfrak{q}=0$.}
            \label{fig:ImQNFScalarq0}
        \end{subfigure}\hfill
        \begin{subfigure}[b]{0.48\linewidth}
            \centering
            \includegraphics[width=.9\linewidth]{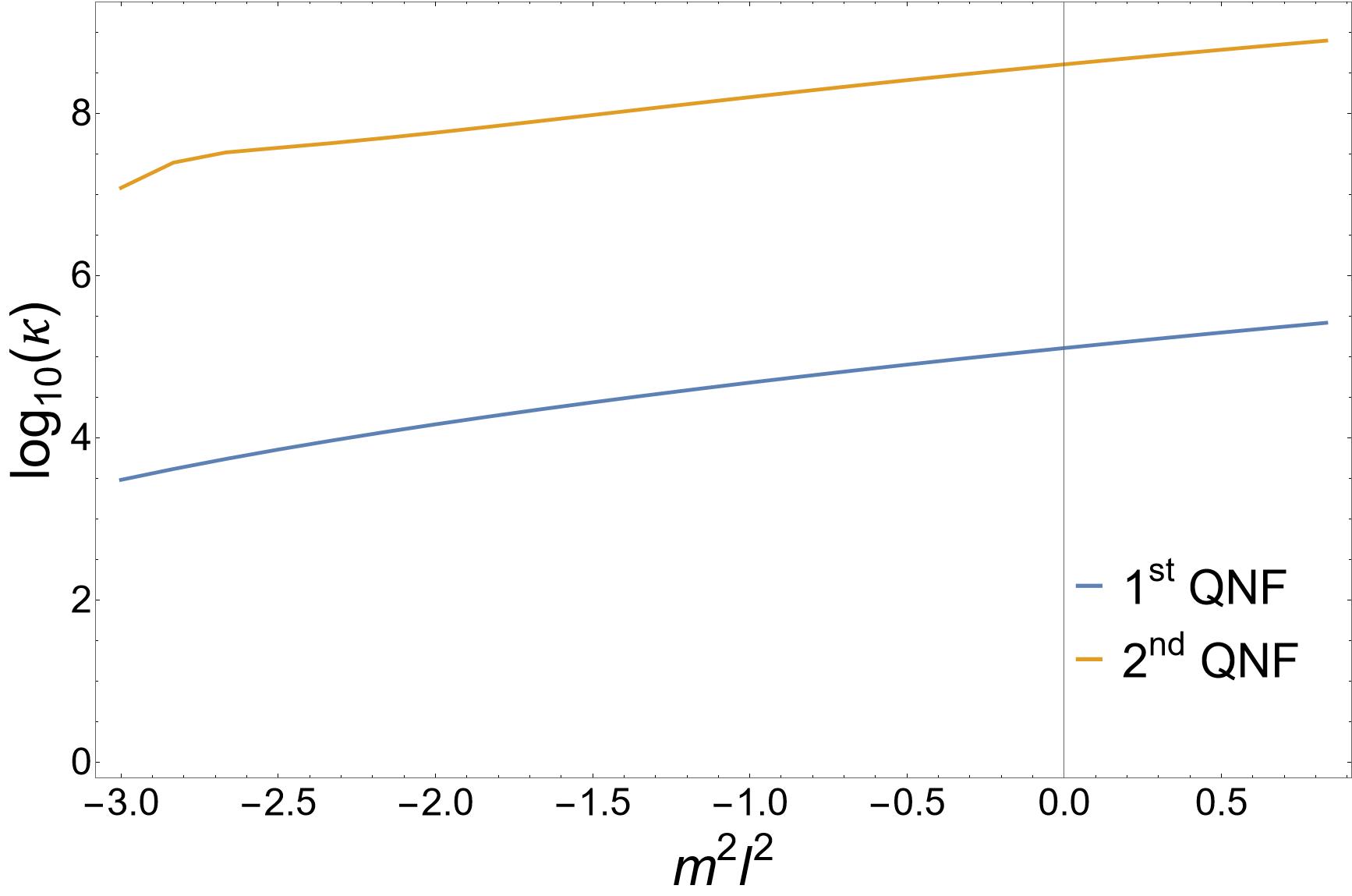}
            \captionsetup{justification=centering}
            \caption{Condition numbers for $\mathfrak{q}=0$.}
            \label{fig:ConditionNumbersScalarq0}
        \end{subfigure}
    \begin{subfigure}[b]{0.48\linewidth}
            \centering
            \includegraphics[width=.9\linewidth]{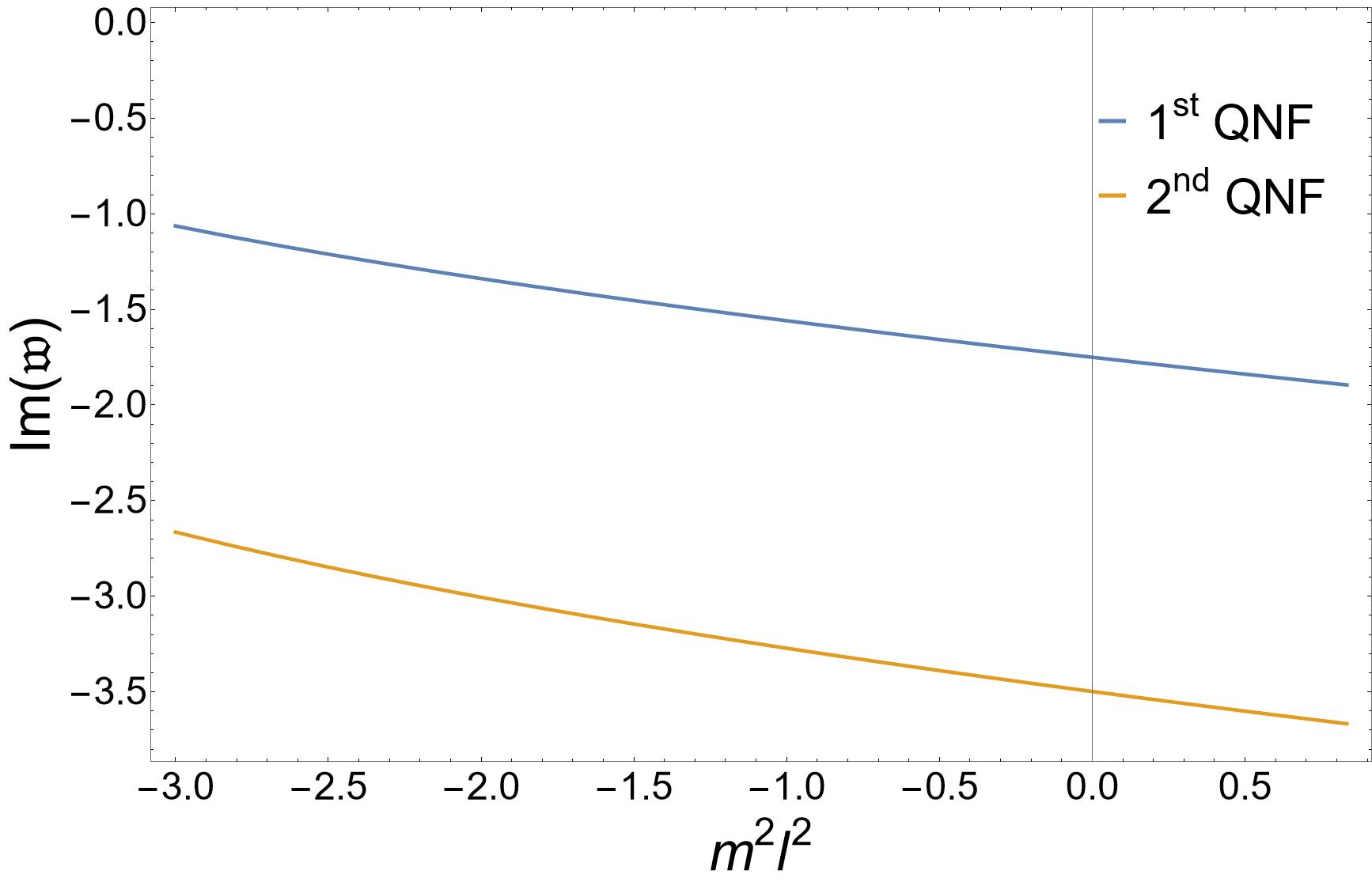}
            \captionsetup{justification=centering}
            \caption{Imaginary part of the QNFs for $\mathfrak{q}=10$.}
            \label{fig:ImQNFScalarq10}
        \end{subfigure}\hfill
        \begin{subfigure}[b]{0.48\linewidth}
            \centering
            \includegraphics[width=.9\linewidth]{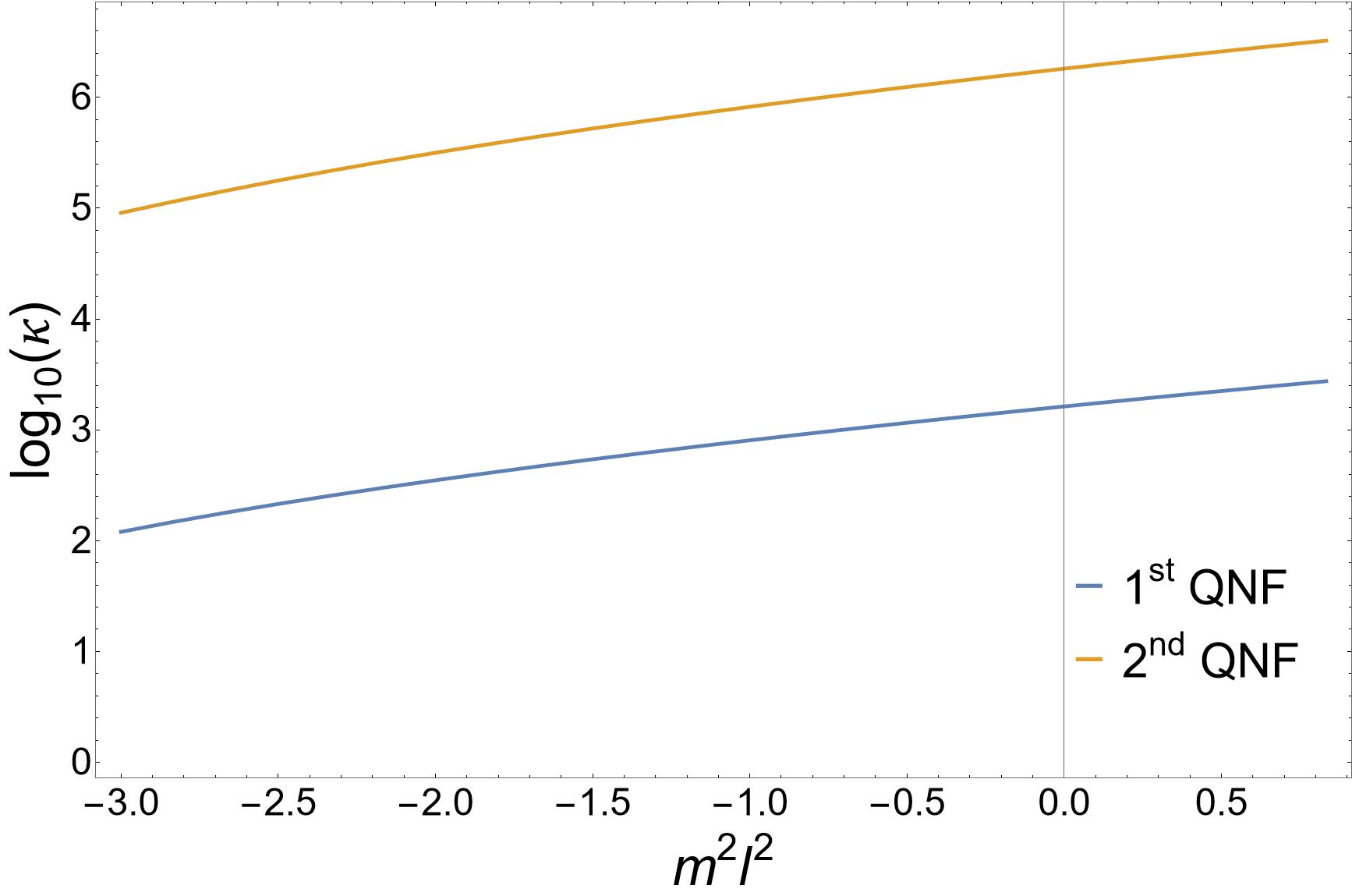}
            \captionsetup{justification=centering}
            \caption{Condition numbers for $\mathfrak{q}=10$.}
            \label{fig:ConditionNumbersScalarq10}
        \end{subfigure}
        \caption{Mass dependence in the energy norm for the scalar fluctuations.}
        \label{fig:massDependenceScalar}
    \end{figure}

    \begin{figure}[htb!]
        \centering
        \begin{subfigure}[b]{0.48\linewidth}
            \centering
            \includegraphics[width=.9\linewidth]{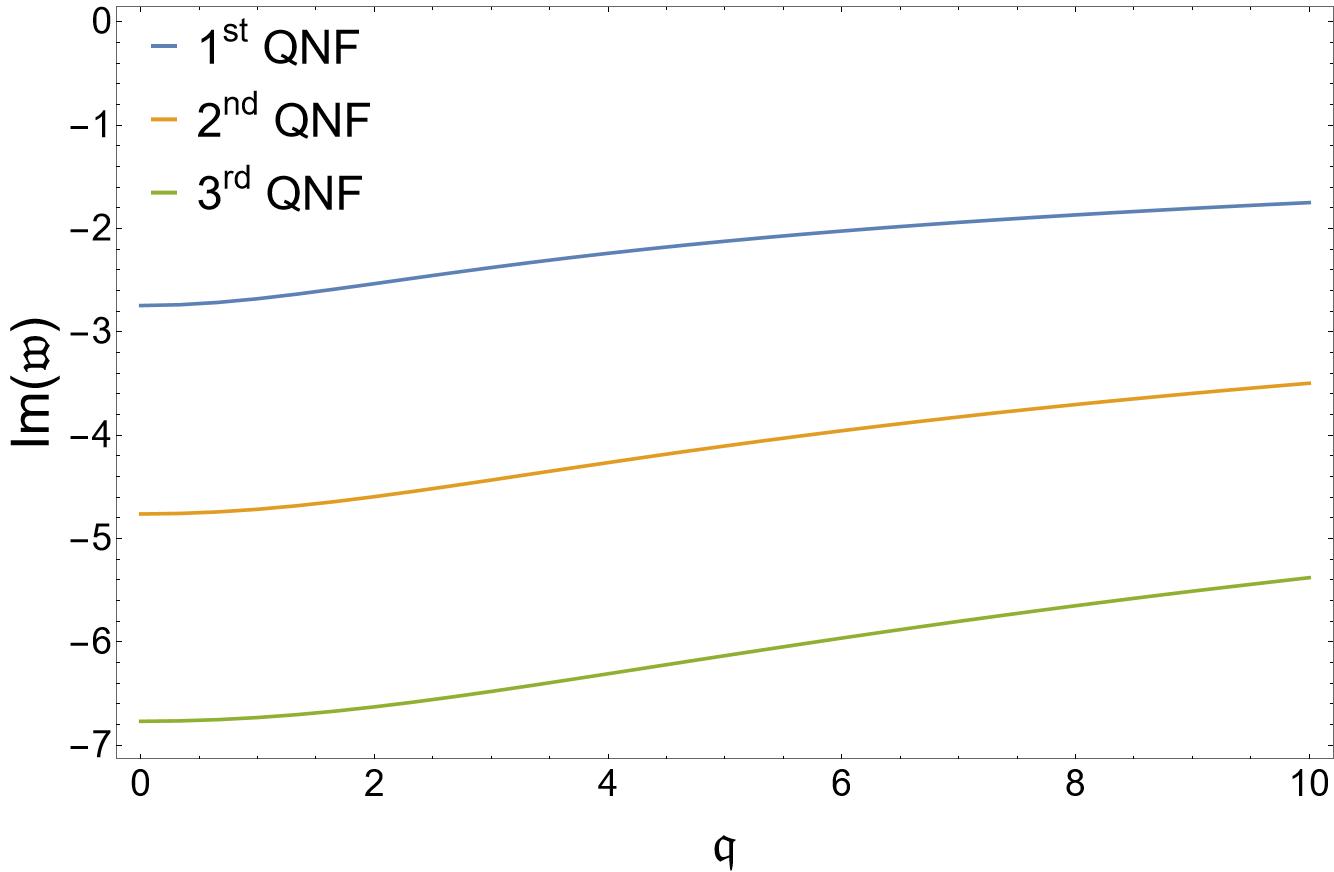}
            \captionsetup{justification=centering}
            \caption{Imaginary part of the QNFs for $m^2l^2=0$.}
            \label{fig:ImQNFMasslessScalar}
        \end{subfigure}\hfill
        \begin{subfigure}[b]{0.48\linewidth}
            \centering
            \includegraphics[width=.9\linewidth]{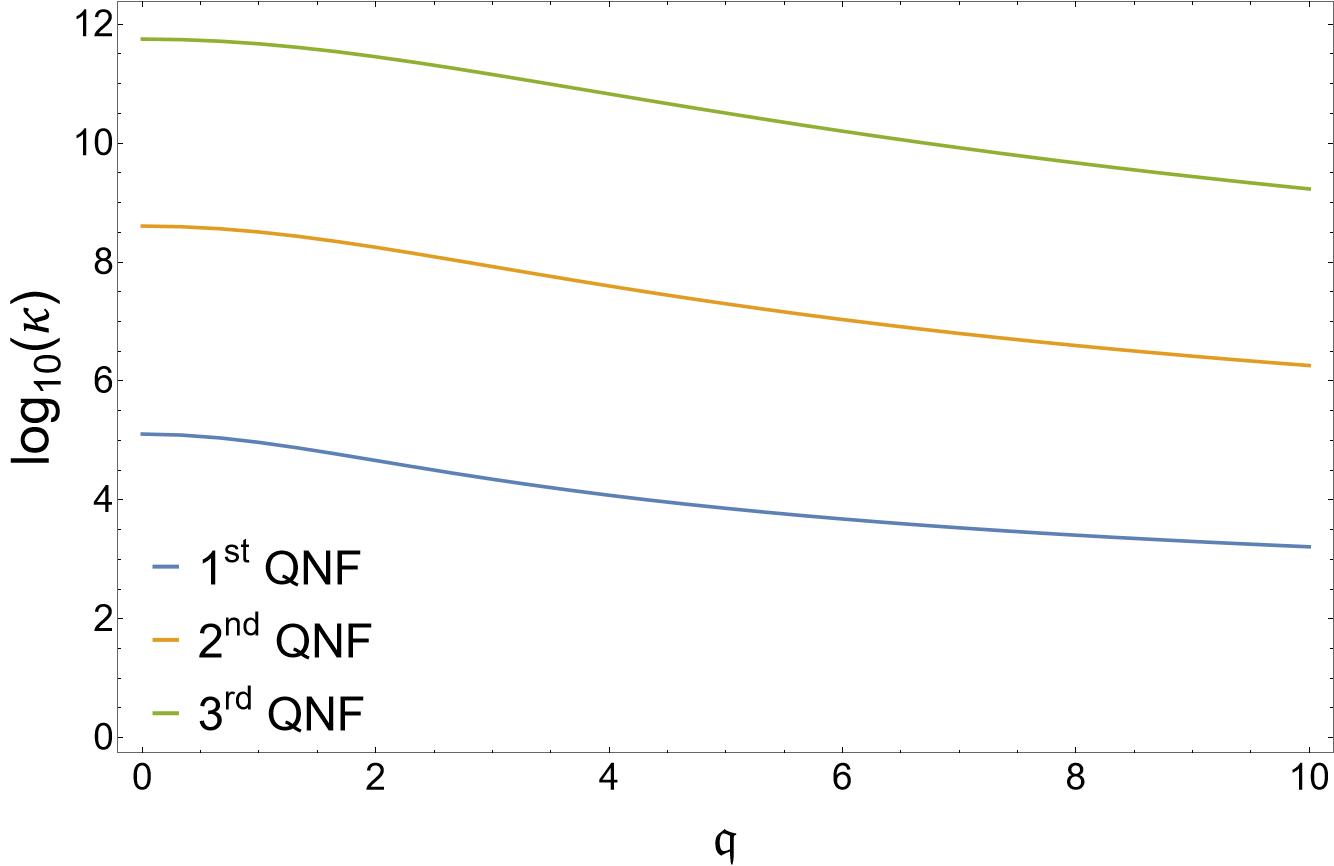}
            \captionsetup{justification=centering}
            \caption{Condition numbers for $m^2l^2=0$.}
            \label{fig:ConditionNumbersMasslessScalar}
        \end{subfigure}
        \begin{subfigure}[b]{0.48\linewidth}
            \centering
            \includegraphics[width=.9\linewidth]{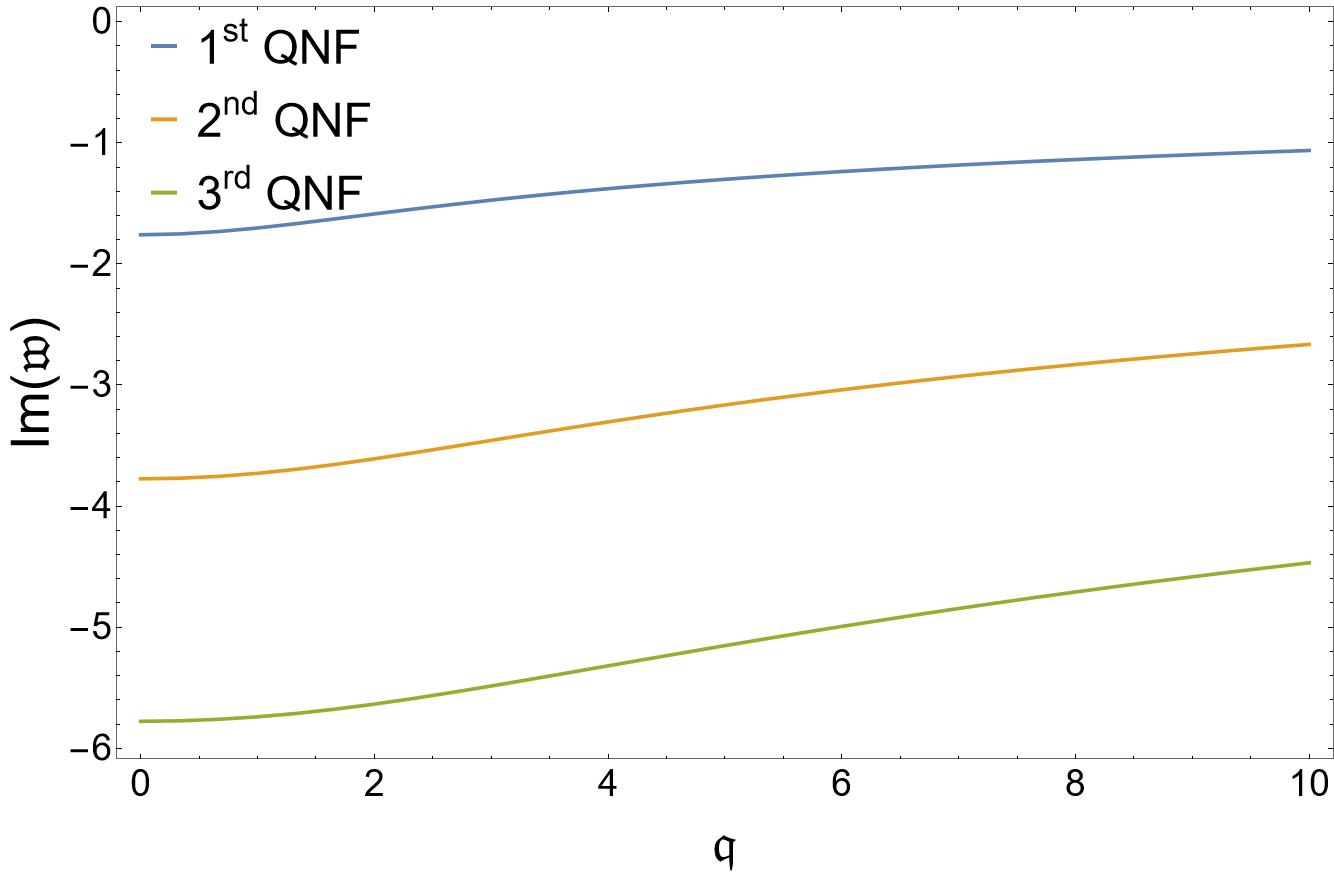}
            \captionsetup{justification=centering}
            \caption{Imaginary part of the QNFs for $m^2l^2=-3$.}
            \label{fig:ImQNFMassiveScalar}
        \end{subfigure}\hfill
        \begin{subfigure}[b]{0.48\linewidth}
            \centering
            \includegraphics[width=.9\linewidth]{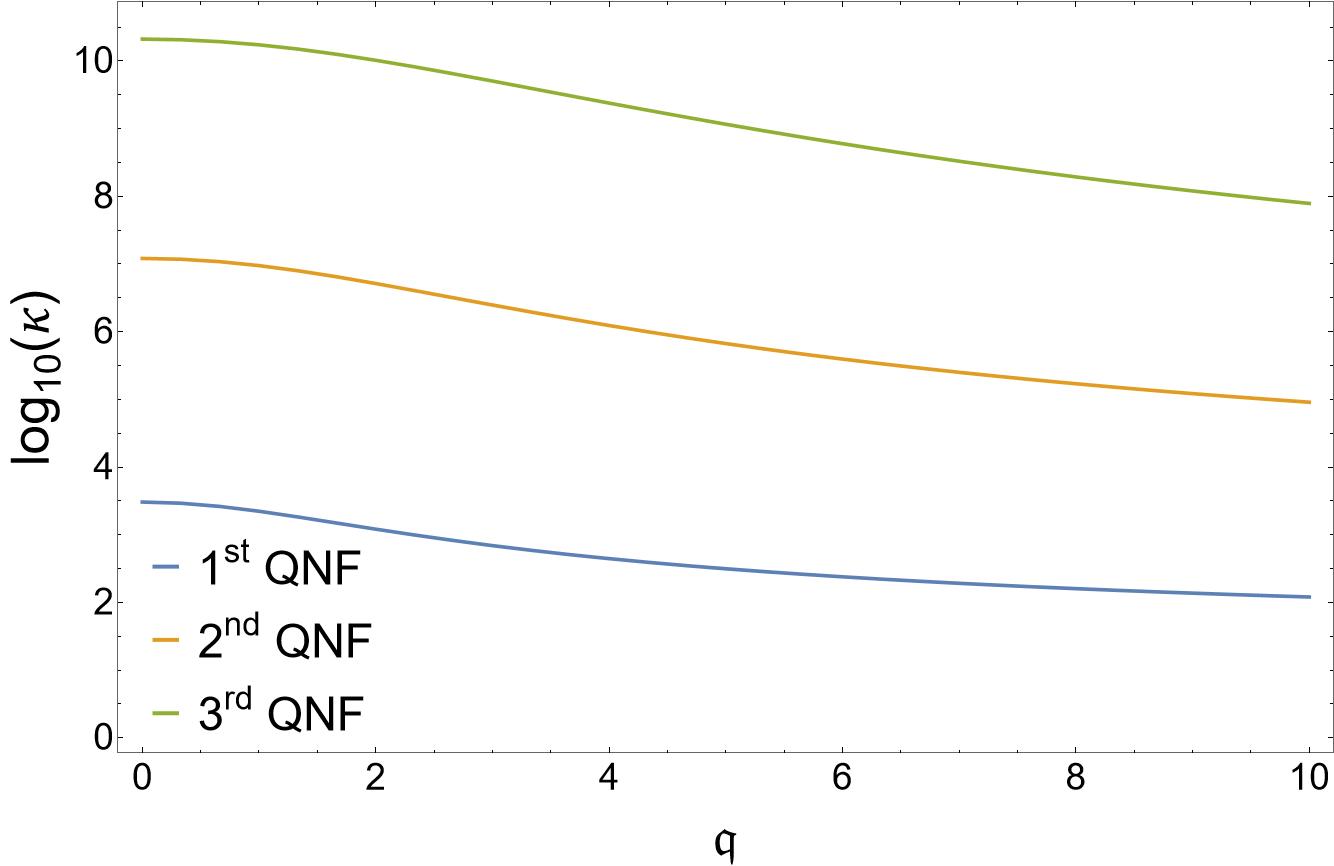}
            \captionsetup{justification=centering}
            \caption{Condition numbers for $m^2l^2=-3$.}
            \label{fig:ConditionNumbersMassiveScalar}
        \end{subfigure}
        \caption{Momentum dependence in the energy norm for the scalar fluctuations.}
        \label{fig:qDependenceScalar}
    \end{figure}
    
  Focusing now on the selective pseudospectrum, we note that its structure seems to resemble the structure of the full pseudospectrum. Nonetheless, the scales clearly do not match; the selective pseudospectrum shows significantly more stability. Hence, we can conclude that most of the instability arises from perturbations which cannot be interpreted as local potential perturbations.

  Regarding the mass and momentum dependence, stability decreases with mass and increases with momentum (see figures \ref{fig:massDependenceScalar} and \ref{fig:qDependenceScalar}). Such behavior can be understood in the context of the instability being related to the distance to the real axis: the imaginary part of the QNFs is reduced as mass (momentum) decreases (increases), and thus we observe more stability. For the QFT this implies that high-momentum fluctuations of operators with small scaling dimension are more stable.

     \begin{figure}[h!]
        \centering
        \begin{subfigure}[b]{0.49\linewidth}
            \centering
            \includegraphics[width=\linewidth]{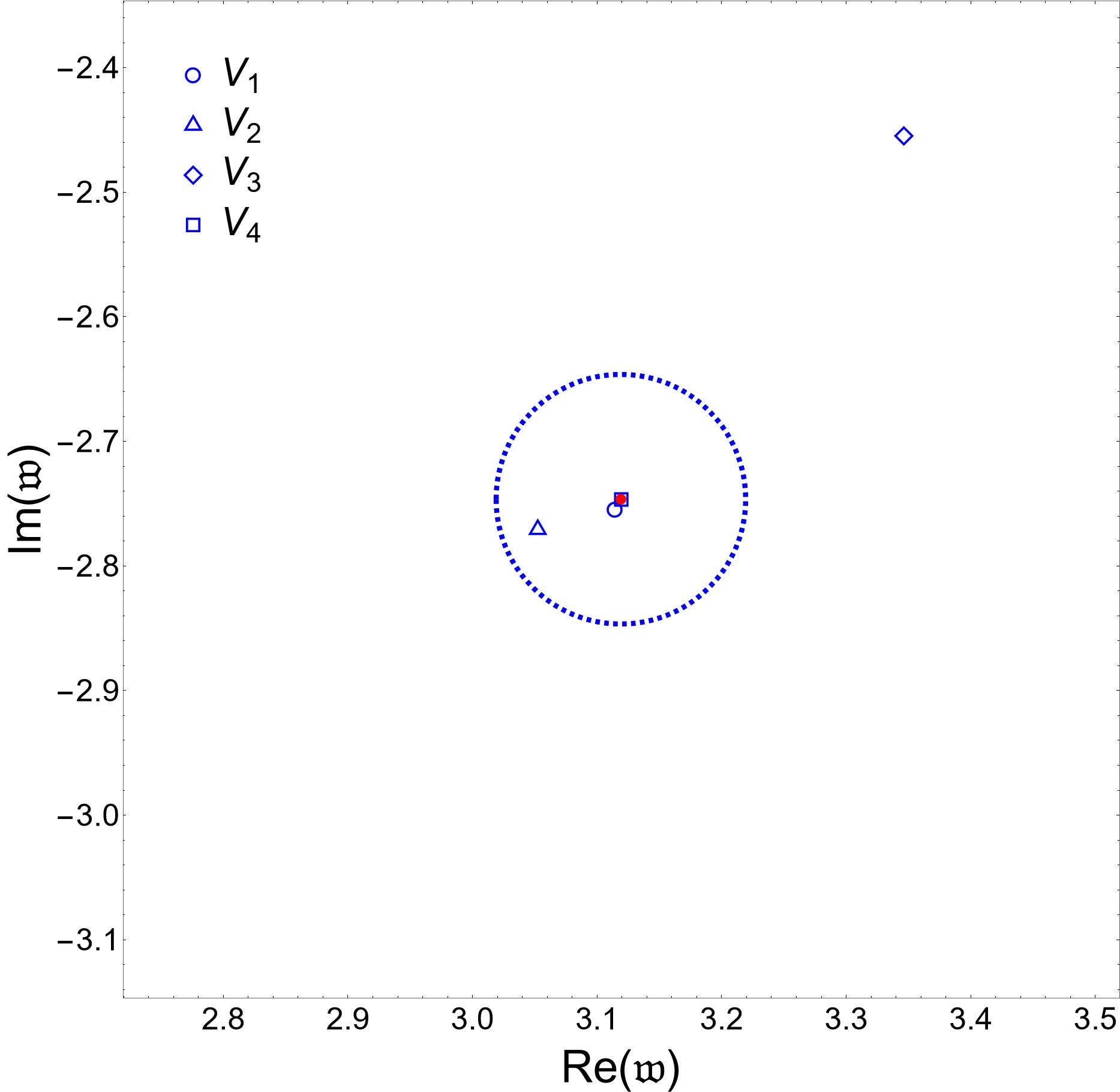}
            \captionsetup{justification=centering}
            \caption{$m^2l^2=0$, $\mathfrak{q}=0$.}
            \label{fig:CloseupDeterministicMasslessScalarq0}
        \end{subfigure}\hfill
        \begin{subfigure}[b]{0.49\linewidth}
            \centering
            \includegraphics[width=\linewidth]{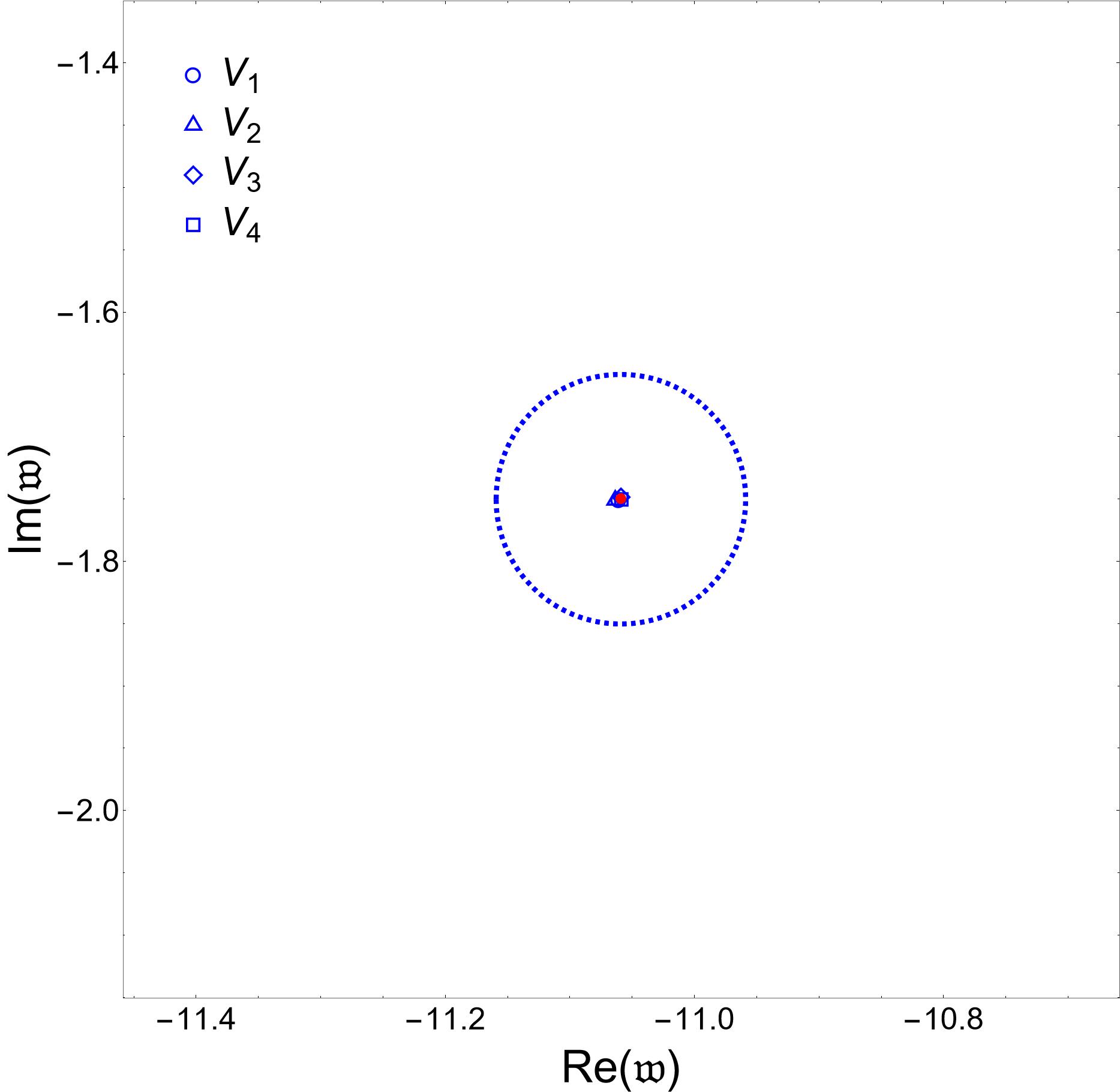}
            \captionsetup{justification=centering}
            \caption{$m^2l^2=0$, $\mathfrak{q}=10$.}
            \label{fig:CloseupDeterministicMasslessScalarq10}
        \end{subfigure}
        \begin{subfigure}[b]{0.49\linewidth}
            \centering
            \includegraphics[width=\linewidth]{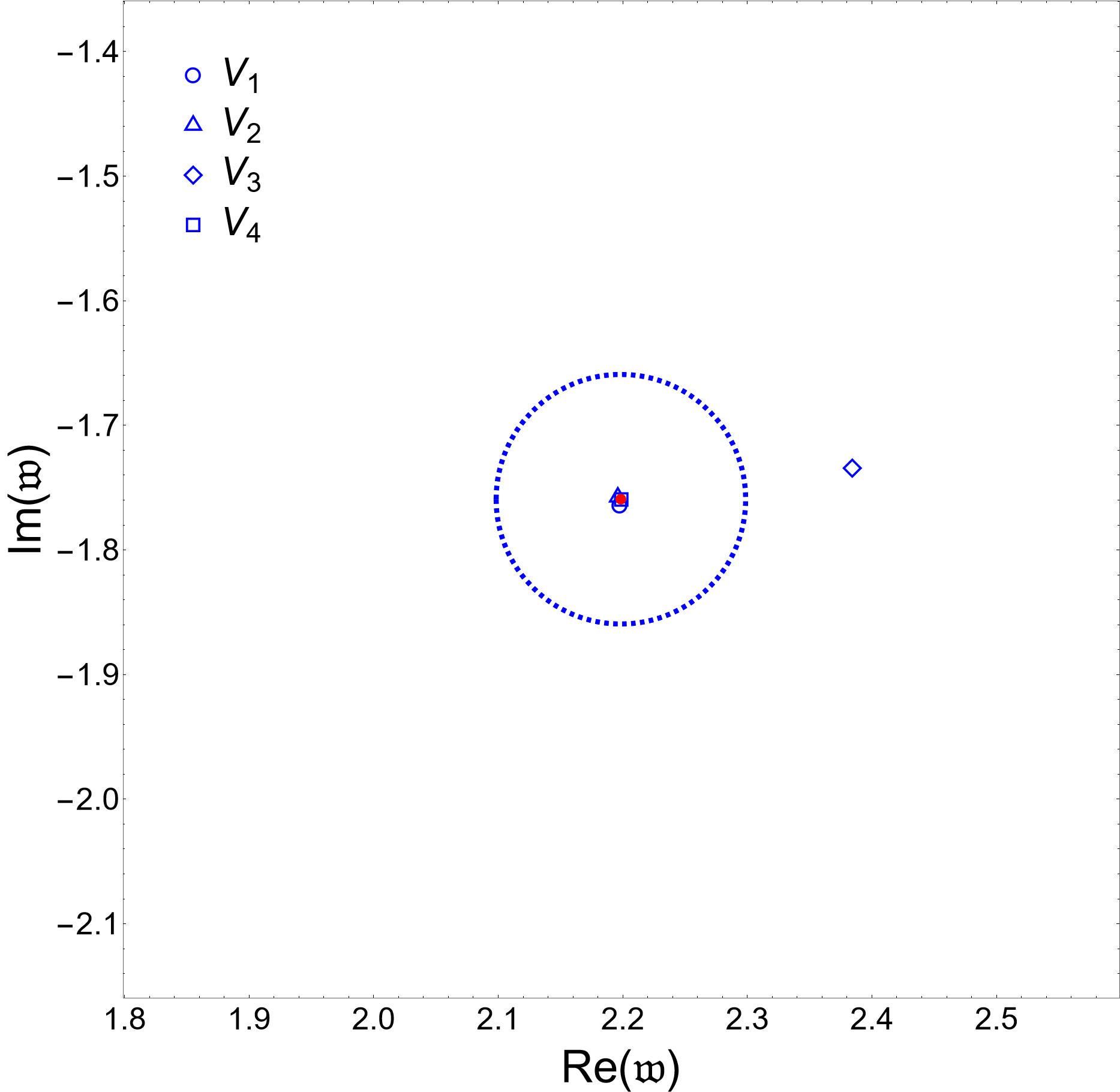}
            \captionsetup{justification=centering}
            \caption{$m^2l^2=-3$, $\mathfrak{q}=0$.}
            \label{fig:CloseupDeterministicMassiveScalarq0}
        \end{subfigure}\hfill
        \begin{subfigure}[b]{0.49\linewidth}
            \centering
            \includegraphics[width=\linewidth]{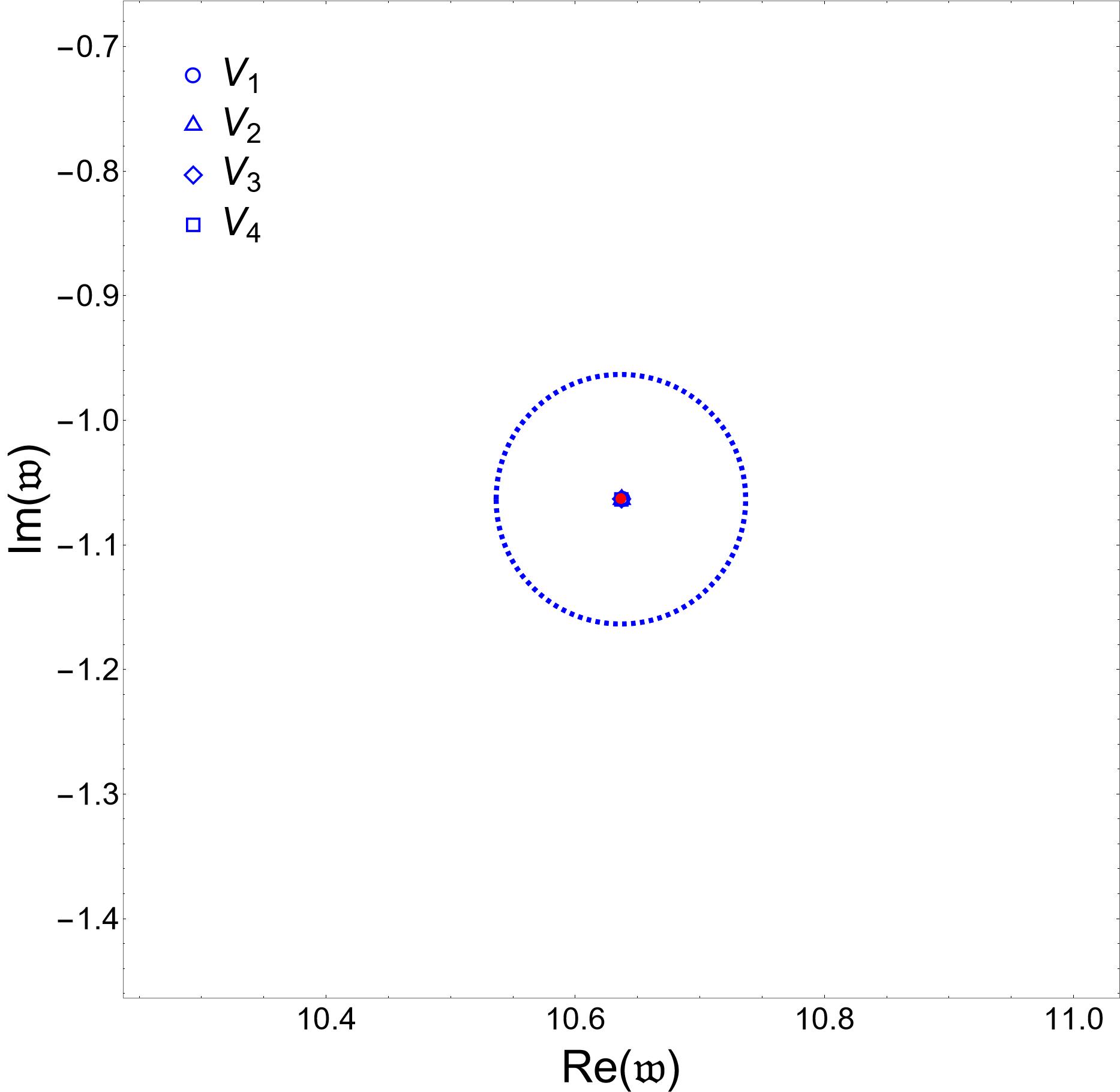}
            \captionsetup{justification=centering}
            \caption{$m^2l^2=-3$, $\mathfrak{q}=10$.}
            \label{fig:CloseupDeterministicMassiveScalarq10}
        \end{subfigure}
        \caption{Effect on the first scalar QNF  of the deterministic perturbations %
        \eqref{eq:deterministicpots}
        with size $\norm{V_i}_E=10^{-1}$. The unperturbed QNF is shown in red, while the perturbed QNFs are depicted in blue. The dashed blue line represents the circle of radius $10^{-1}$ centered in the unperturbed QNF. Remarkably, only for $\mathfrak{q}=0$ do we observe instability under near-horizon perturbations ($V_3$).}
        \label{fig:CloseupDeterministicScalar}
    \end{figure}   

    \begin{figure}[h!]
        \centering
        \begin{subfigure}[b]{0.44\linewidth}
            \centering
            \includegraphics[width=\linewidth]{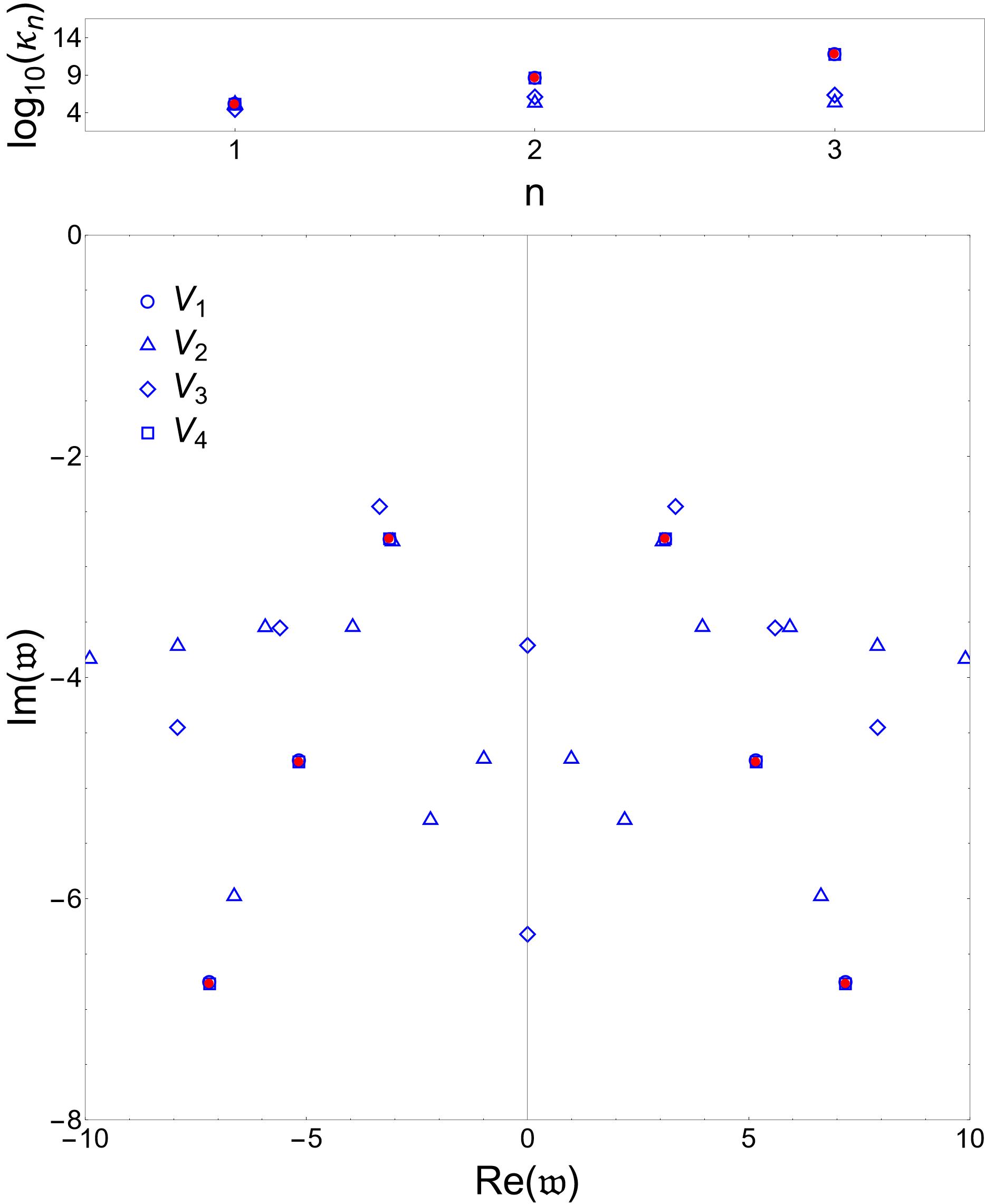}
            \captionsetup{justification=centering}
            \caption{$m^2l^2=0$, $\mathfrak{q}=0$.}
            \label{fig:LargeDeterministicMasslessScalarq0}
        \end{subfigure}\hfill
        \begin{subfigure}[b]{0.44\linewidth}
            \centering
            \includegraphics[width=\linewidth]{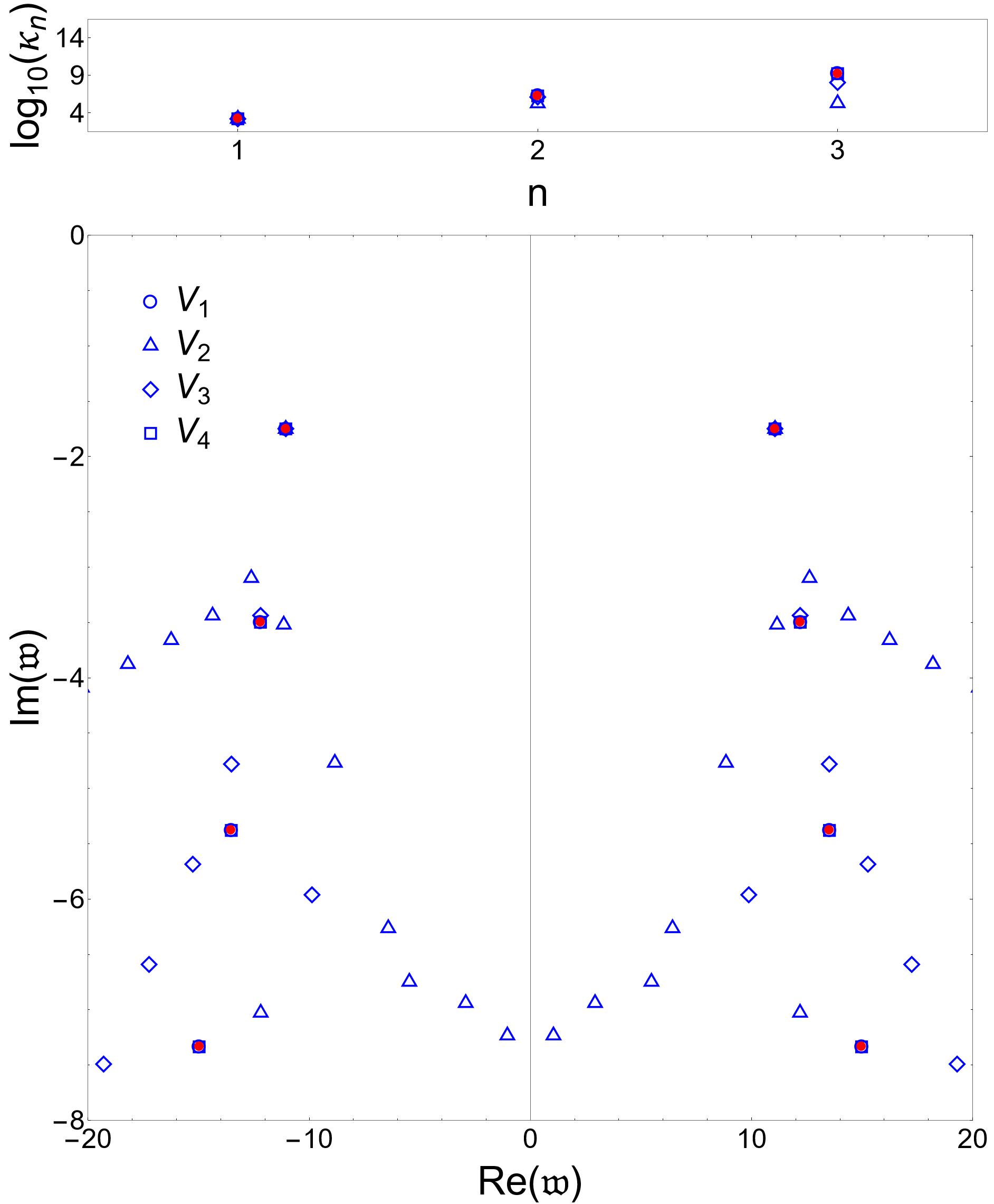}
            \captionsetup{justification=centering}
            \caption{$m^2l^2=0$, $\mathfrak{q}=10$.}
            \label{fig:LargeDeterministicMasslessScalarq10}
        \end{subfigure}
        \begin{subfigure}[b]{0.44\linewidth}
            \centering
            \includegraphics[width=\linewidth]{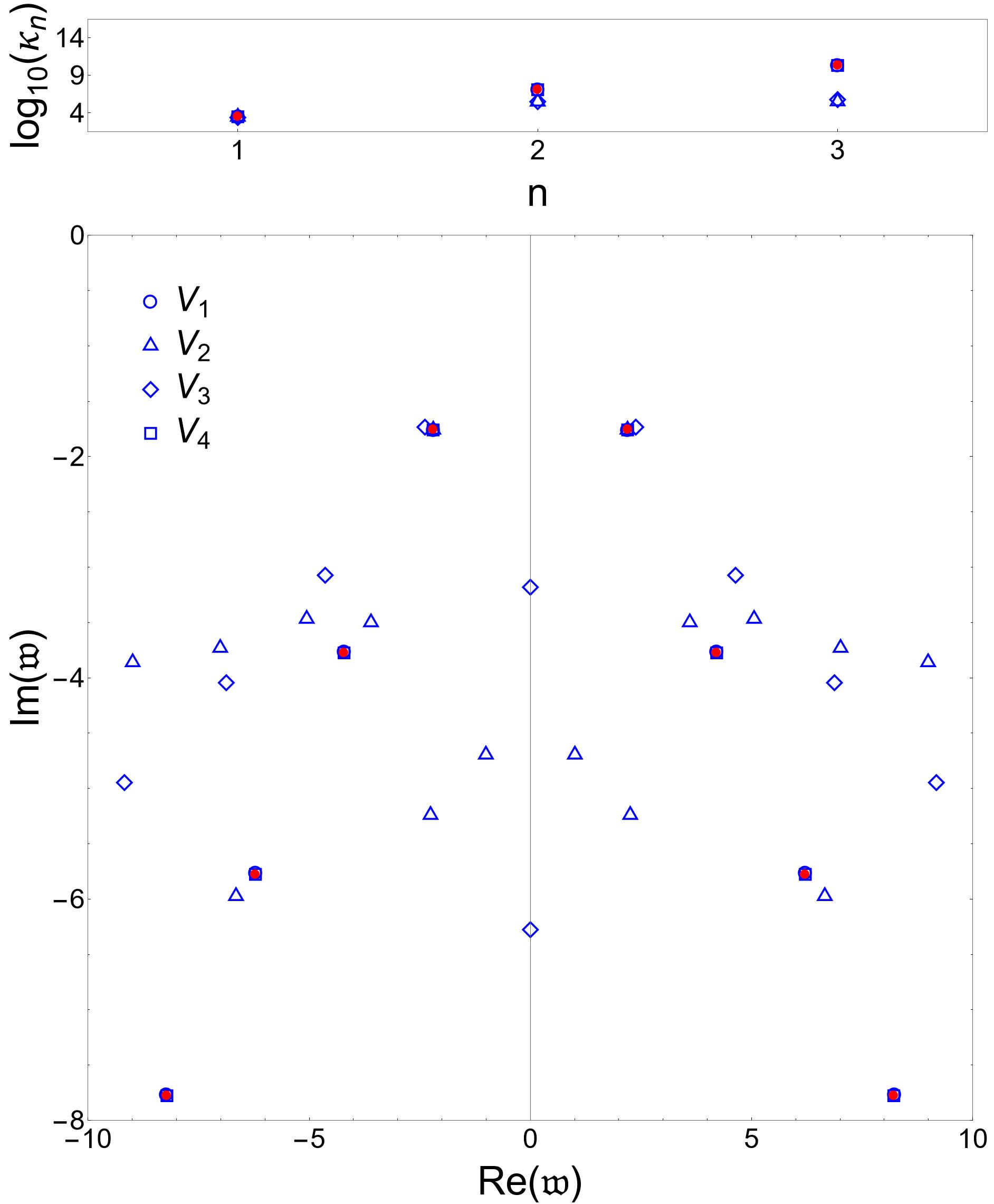}
            \captionsetup{justification=centering}
            \caption{$m^2l^2=-3$, $\mathfrak{q}=0$.}
            \label{fig:LargeDeterministicMassiveScalarq0}
        \end{subfigure}\hfill
        \begin{subfigure}[b]{0.44\linewidth}
            \centering
            \includegraphics[width=\linewidth]{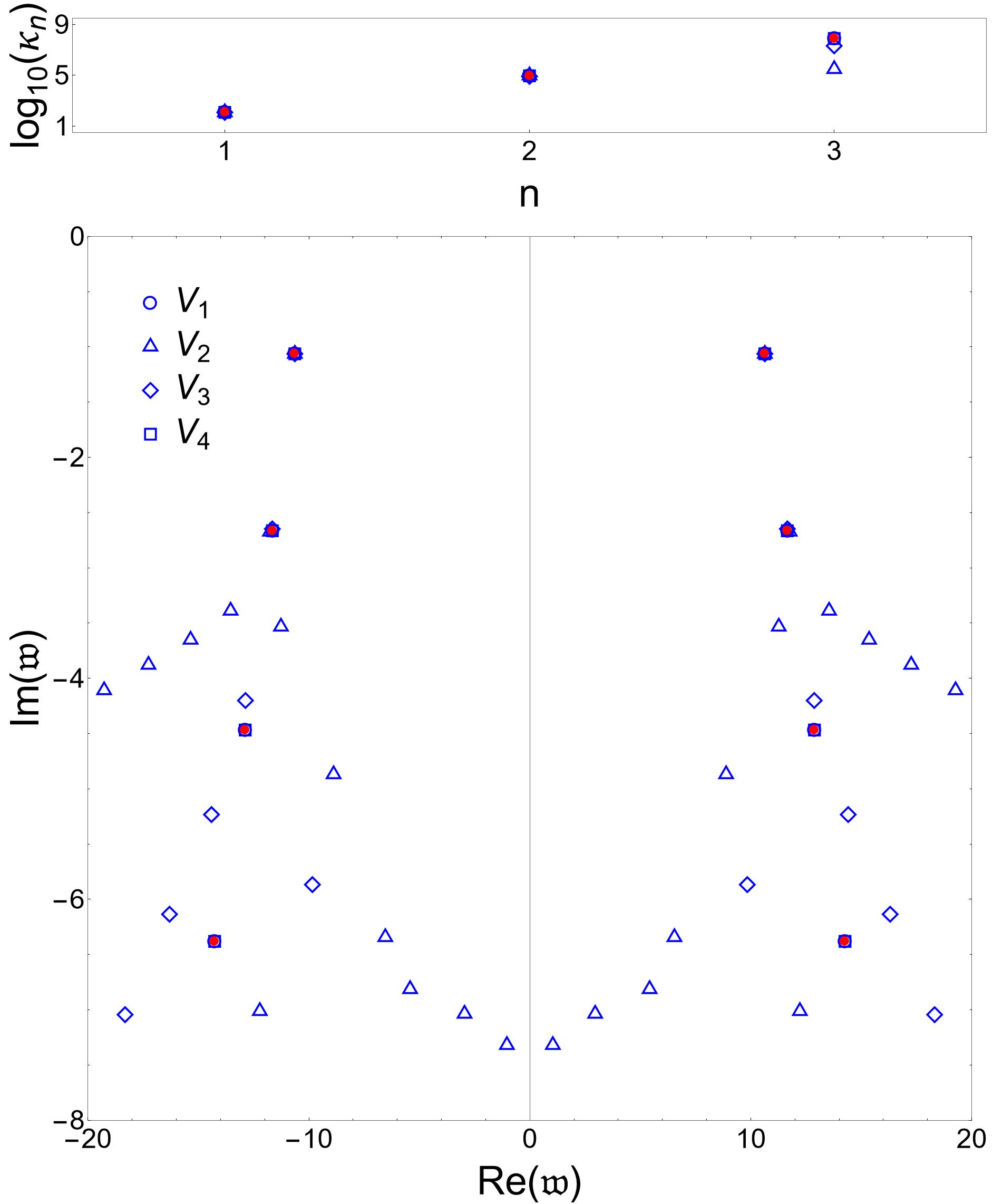}
            \captionsetup{justification=centering}
            \caption{$m^2l^2=-3$, $\mathfrak{q}=10$.}
            \label{fig:LargeDeterministicMassiveScalarq10}
        \end{subfigure}
        \caption{Effect on the spectrum of the scalar of the deterministic perturbations %
        \eqref{eq:deterministicpots}
        with size $\norm{V_i}_E=10^{-1}$. In the lower panels we present the spectra and in the upper ones the condition numbers for the lowest QNFs. The unperturbed QNFs are shown in red, while the perturbed ones are depicted in blue. The plotted region of the spectra is stable under long $\rho$-wavelength perturbations ($V_1$) and near-boundary perturbations ($V_4$). Note also that, as indicated by the condition numbers, the perturbed QNFs associated with $V_2$ and $V_3$ are more stable than the unperturbed QNFs.}
        \label{fig:LargeDeterministicScalar}
    \end{figure}

 We now focus on 
 the first QNF (see figure \ref{fig:CloseupPseudospectraScalar}), which has the smallest imaginary part and thus dominates the low energy spectrum of the QFT.
 Remarkably, we find that for $m^2l^2=0$ at $\mathfrak{q}=0$
 that QNF is unstable under local potential perturbations as the selective $\varepsilon$-pseudospectra extends beyond the circle of radius $\varepsilon$ centered around the QNF. To further explore the nature of this instability, in figure \ref{fig:CloseupDeterministicScalar} we analyze the effect of the deterministic potential perturbations
 \eqref{eq:deterministicpots}
 on the first QNF. We observe that the instability is only present for near-horizon perturbations, concluding that it is mainly associated with deformations of the IR of the dual QFT. Moreover, we also find that, although we did not have statistical evidence in the selective pseudospectra, the first QNF is actually unstable under local potential perturbations for $m^2l^2=-3$ (figure \ref{fig:CloseupDeterministicMassiveScalarq0}). At $\mathfrak{q}=10$ we find that the QNM is stable, as expected by the enhanced stability at large $\mathfrak{q}$.

\clearpage
    
    Continuing our exploration of the effect of the deterministic potentials, in figure \ref{fig:LargeDeterministicScalar} we present the spectrum and the condition numbers for the first 3 perturbed QNFs. We observe that large $\rho$-wavelength ($V_1$) and near-boundary ($V_4$) perturbations have very little effect on the spectrum. On the other hand, short $\rho$-wavelength ($V_2$) and near-horizon perturbations ($V_3$) displace the QNFs significantly, generating new branching structures. Remarkably, these perturbed QNFs are more stable than the unperturbed ones as indicated by the condition numbers plotted in the upper panels of figure \ref{fig:LargeDeterministicScalar}. Thus, the spectra of poles of Green's functions in strongly coupled holographic field theories might have some preference for branching structures over the unperturbed ``Christmas Tree'' configuration.

We end this section by briefly addressing the results for the pseudospectrum of the scalar field in the $L^2$-norm.
We refer the reader to Appendix~\ref{app:l2norm} and here 
we simply point out that the pseudospectra are qualitatively similar to those in the energy norm. However, in the $L^2$-norm we find that the stability is enhanced under local potential perturbations and reduced under generic perturbations.

    \subsection{Transverse Gauge Field in \texorpdfstring{$\text{SAdS}_{4+1}$}{SAdS4+1}}

In this section we present the pseudospectrum of the fluctuations of the transverse gauge field introduced in section~\ref{ssec:gaugefield}.
As we detail in the following, we observe the same generic features as for the real scalar. In particular we find the QNFs of the gauge field to be unstable. 
    
\begin{figure}[h]
        \centering
        \begin{subfigure}[b]{0.48\linewidth}
            \centering
            \includegraphics[width=\linewidth]{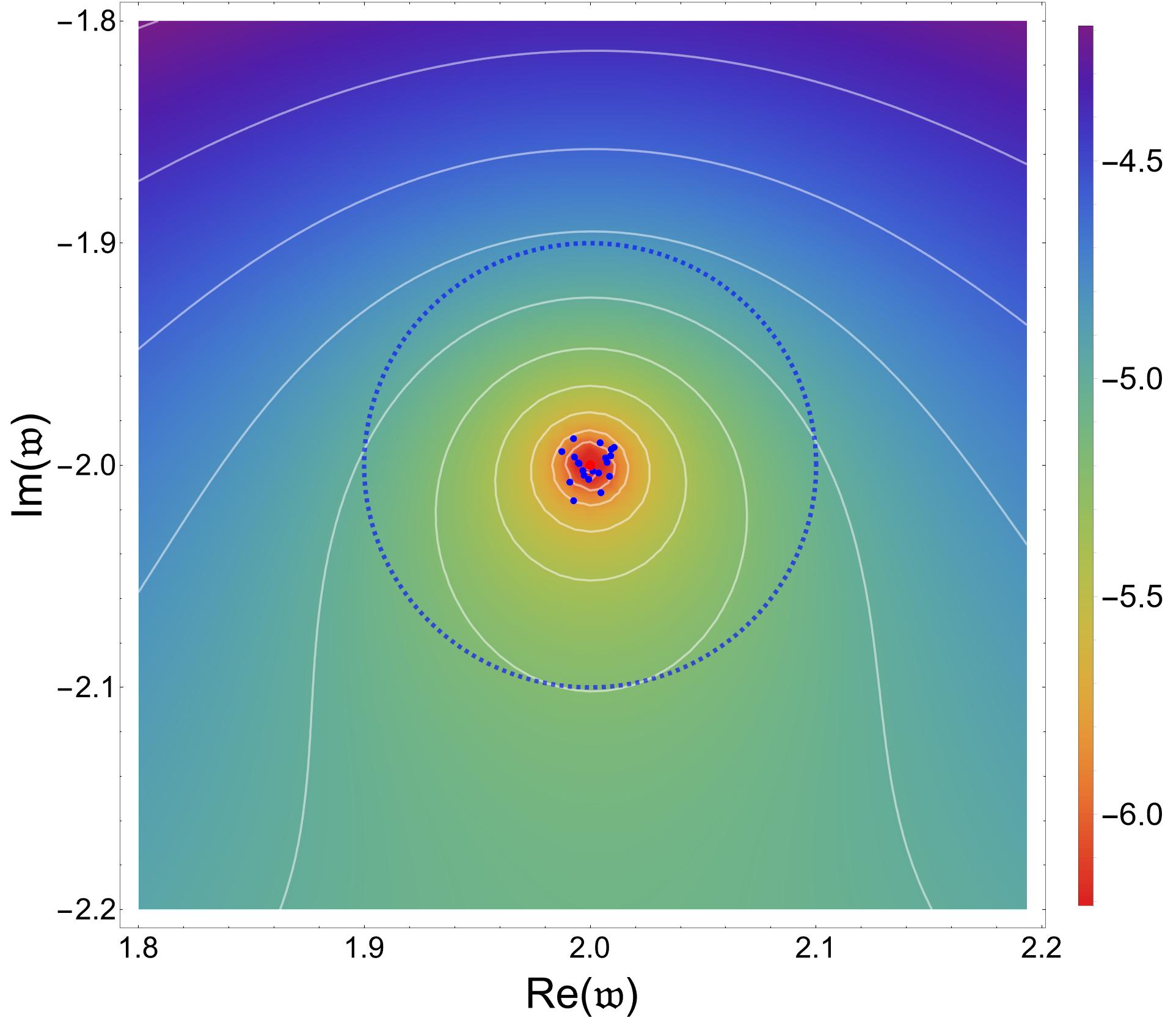}
            \captionsetup{justification=centering}
            \caption{$\mathfrak{q}=0$.}
            \label{fig:CloseupPseudoGFq0}
        \end{subfigure}\hfill
        \begin{subfigure}[b]{0.48\linewidth}
            \centering
            \includegraphics[width=\linewidth]{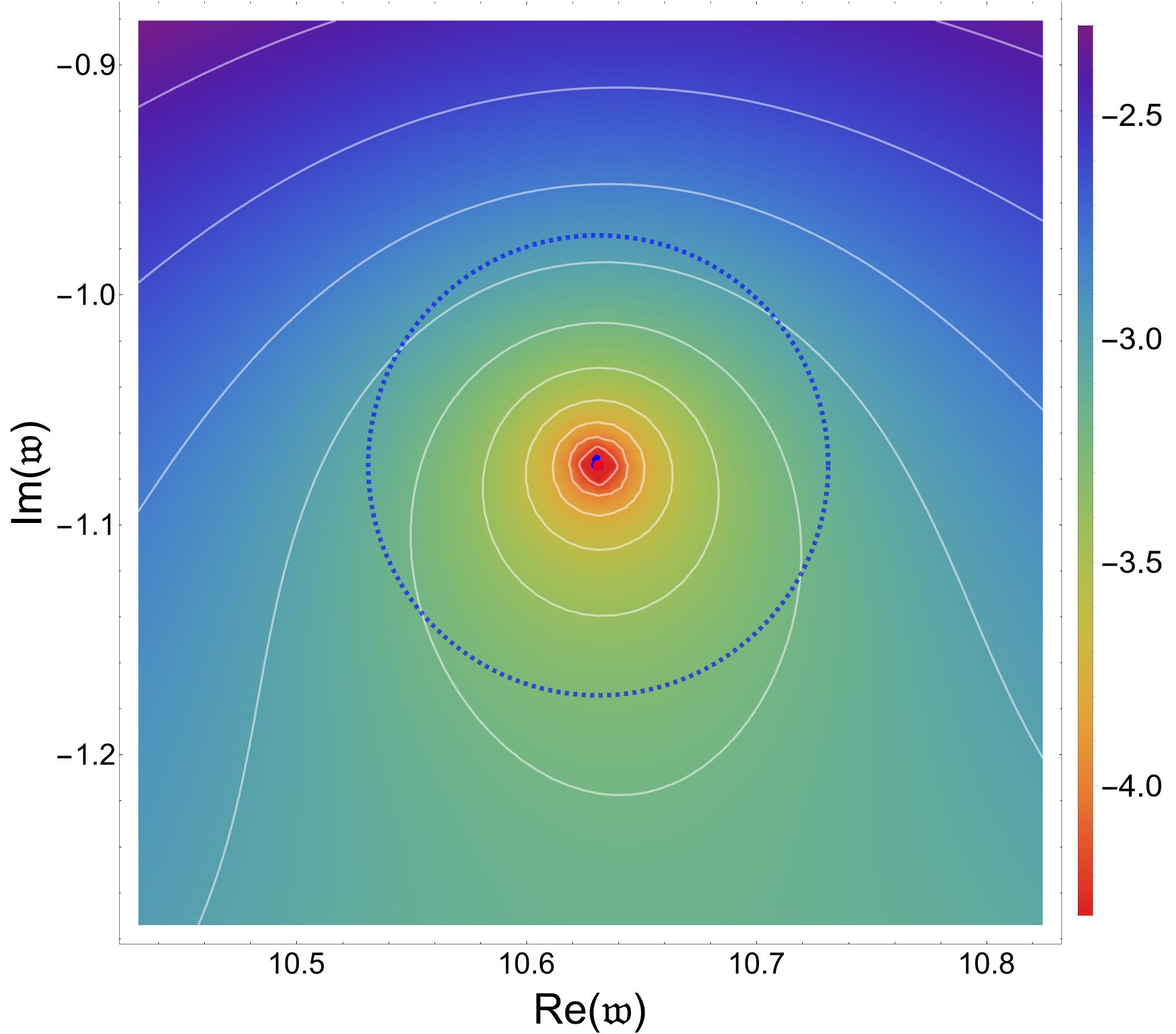}
            \captionsetup{justification=centering}
            \caption{$\mathfrak{q}=10$.}
            \label{fig:CloseupPseudoGFq10}
        \end{subfigure}
        \caption{Close-up of the transverse gauge field pseudospectrum in the energy norm around the first QNF for different values of $\mathfrak{q}$. The red dot corresponds to the QNF, the white lines represent the boundaries of various full $\varepsilon$-pseudospectra, and the dashed blue circle symbolizes a circle with a radius of $10^{-1}$ centered on the QNF. The heat map corresponds to the logarithm in base 10 of the inverse of the resolvent, while the blue dots indicate selective $\varepsilon$-pseudospectra computed with random local potential perturbations.}%
        \label{fig:CloseupPseudospectraGF}
    \end{figure}

    \begin{figure}[htb]
        \centering
        \begin{subfigure}[b]{0.48\linewidth}
            \centering
            \includegraphics[width=\linewidth]{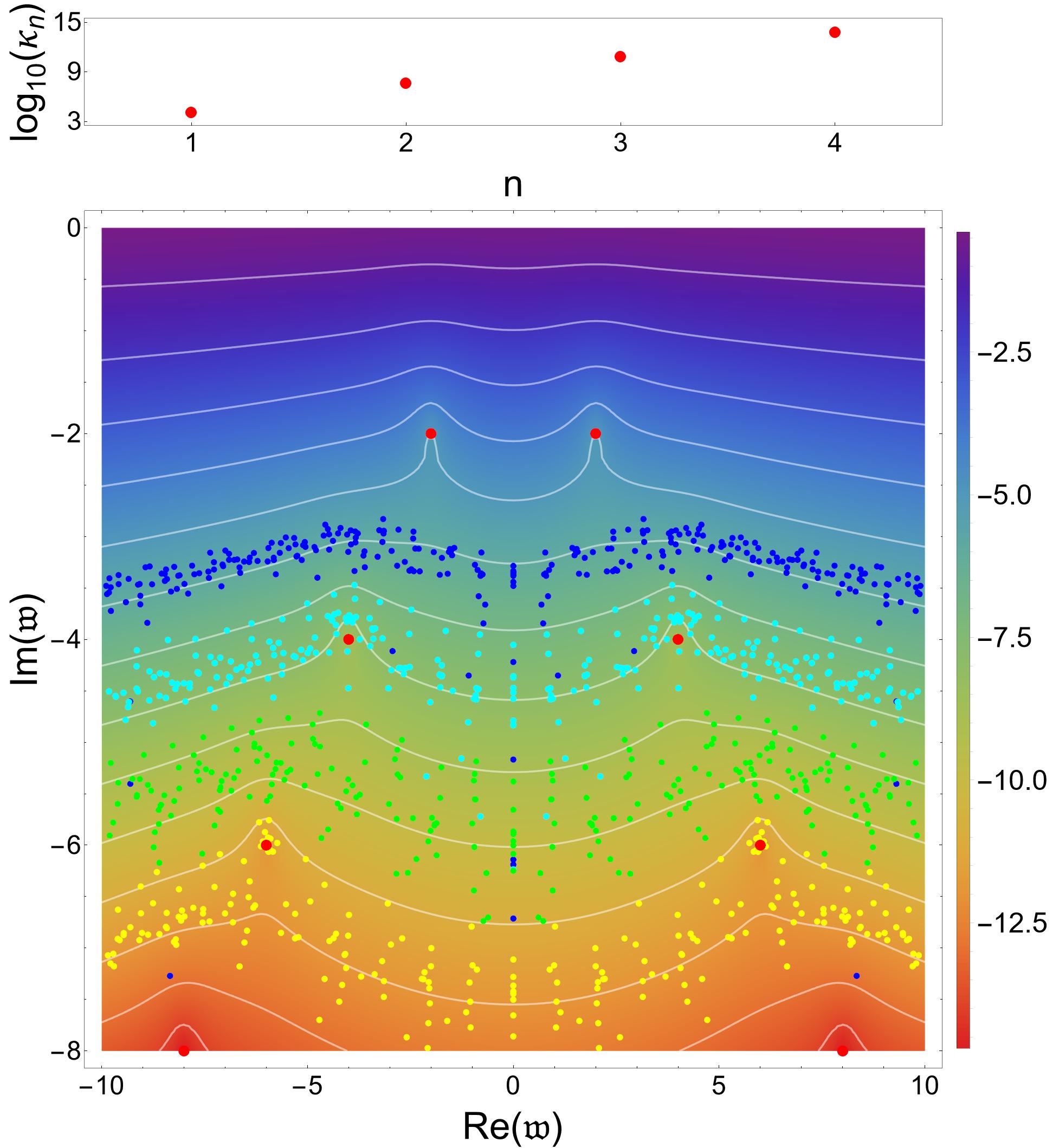}
            \captionsetup{justification=centering}
            \caption{$\mathfrak{q}=0$.}
            \label{fig:LargePseudoGFq0}
        \end{subfigure}\hfill
        \begin{subfigure}[b]{0.48\linewidth}
            \centering
            \includegraphics[width=\linewidth]{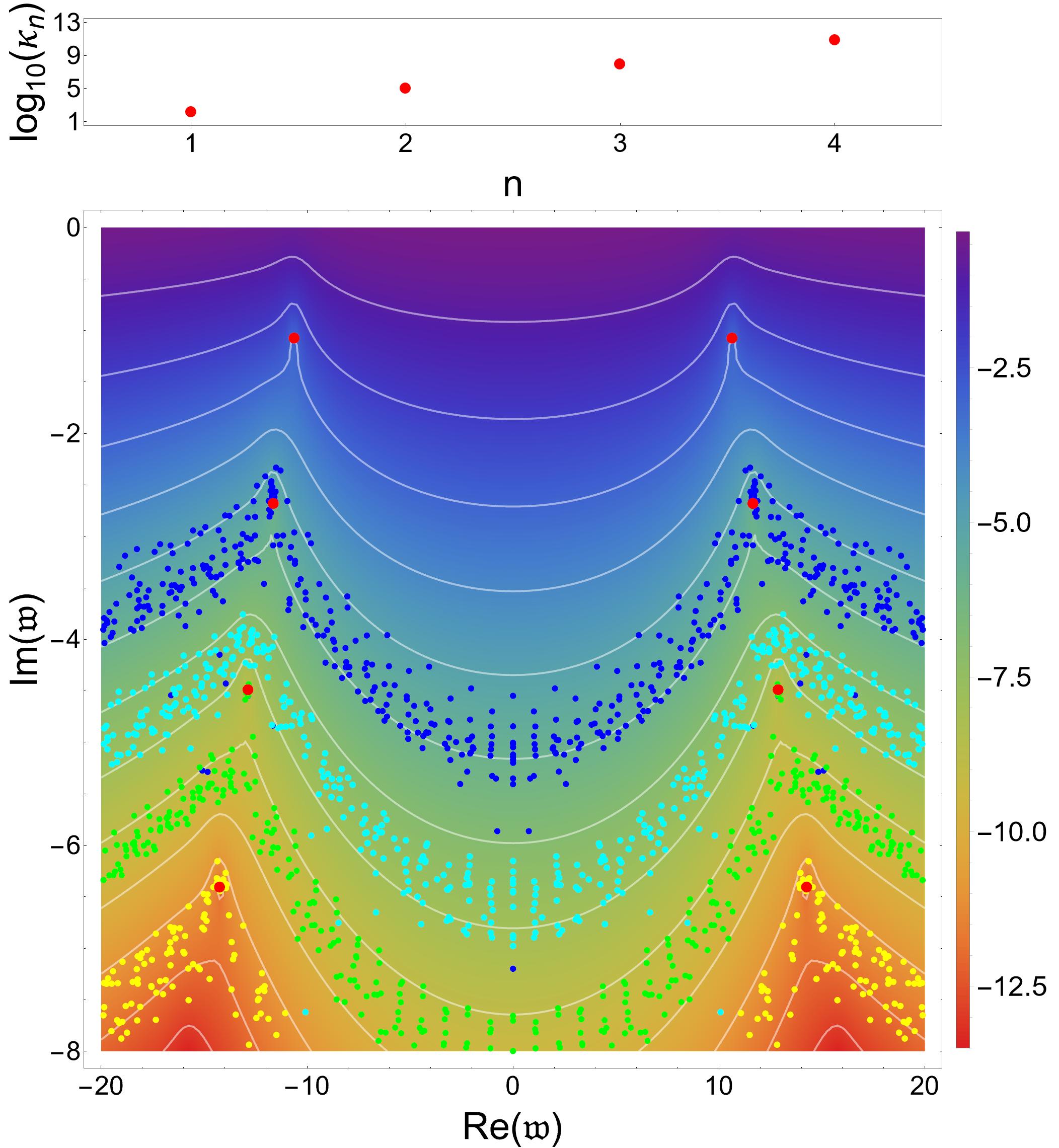}
            \captionsetup{justification=centering}
            \caption{$\mathfrak{q}=10$.}
            \label{fig:LargePseudoGFq10}
        \end{subfigure}
        \caption{Transverse gauge field pseudospectrum in the energy norm for different values of $\mathfrak{q}$. In the lower panels, we present selective and full pseudospectra. The red dots represent the QNFs, and the white lines denote the boundaries of different full $\varepsilon$-pseudospectra. The heat map corresponds to the logarithm in base 10 of the inverse of the resolvent, while the blue, cyan, green, and yellow dots indicate different selective $\varepsilon$-pseudospectra computed with random local potential perturbations of size $10^{-1}$, $10^{-3}$, $10^{-5}$, and $10^{-7}$; respectively. In the upper panels, we represent the condition numbers.}
        \label{fig:LargePseudospectraGF}
    \end{figure}

    \begin{figure}[htb]
        \centering
        \begin{subfigure}[b]{0.49\linewidth}
            \centering
            \includegraphics[width=\linewidth]{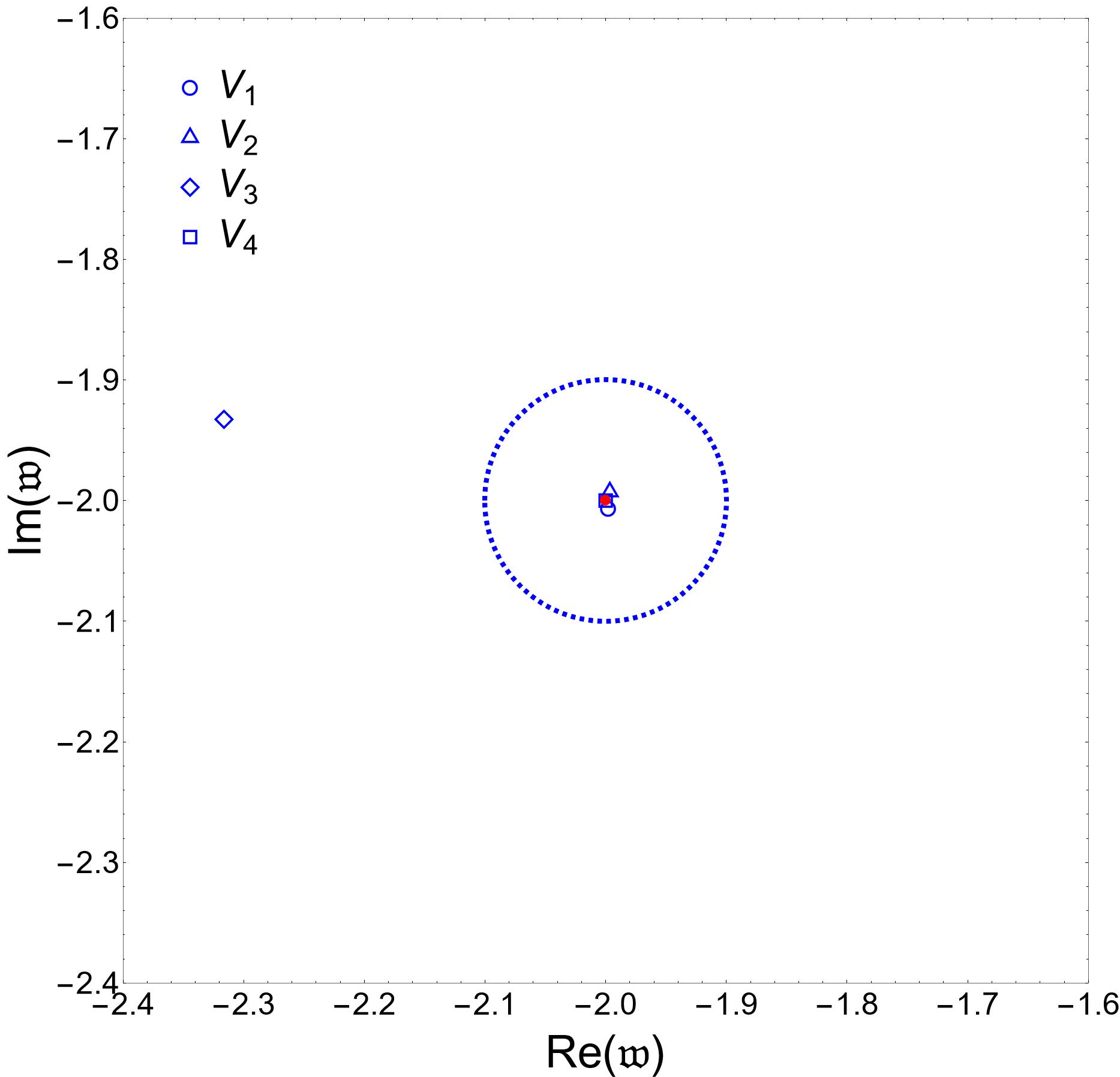}
            \captionsetup{justification=centering}
            \caption{$\mathfrak{q}=0$.}
            \label{fig:CloseupDeterministicGFq0}
        \end{subfigure}\hfill
        \begin{subfigure}[b]{0.49\linewidth}
            \centering
            \includegraphics[width=\linewidth]{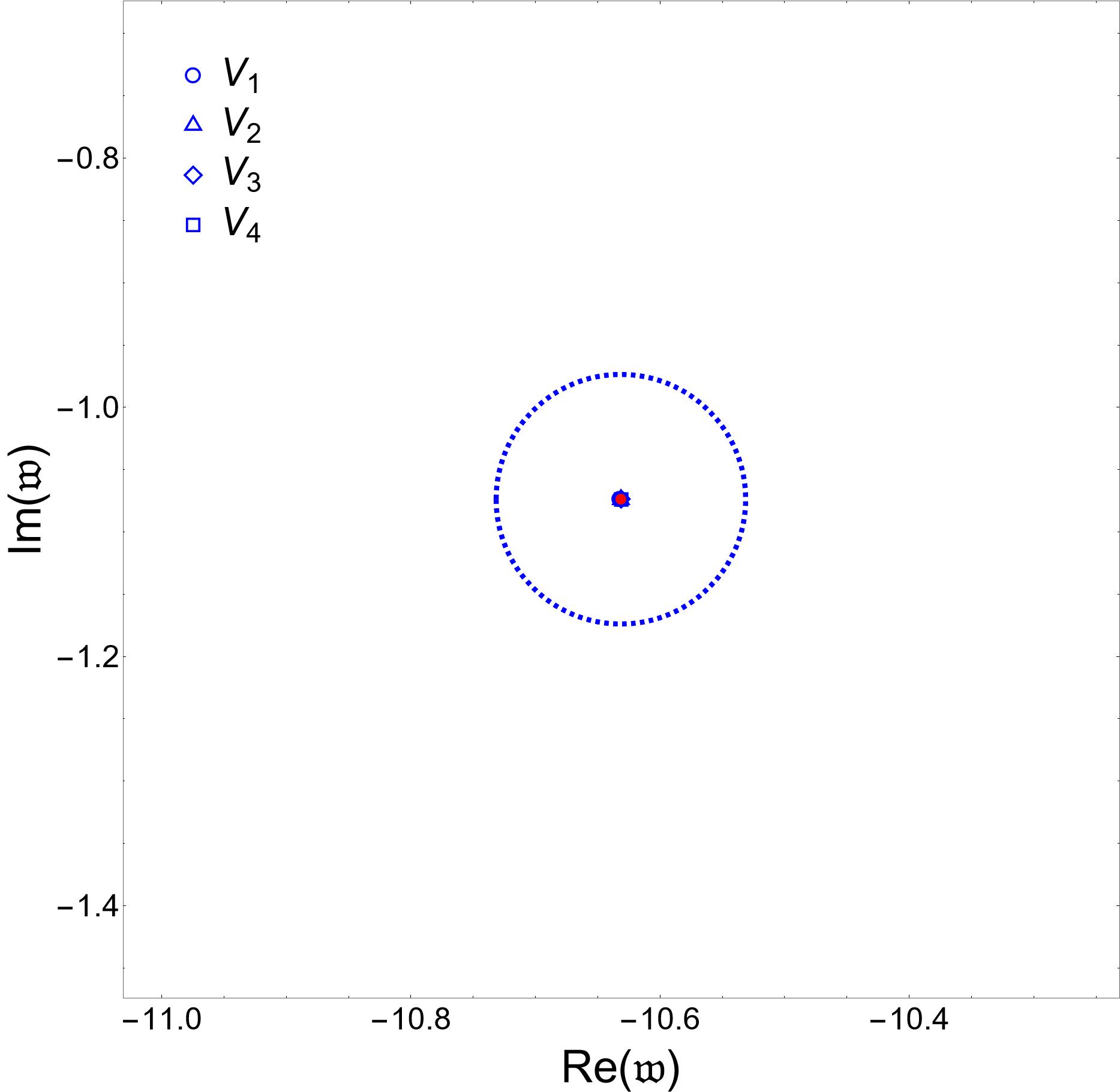}
            \captionsetup{justification=centering}
            \caption{$\mathfrak{q}=10$.}
            \label{fig:CloseupDeterministicGFq10}
        \end{subfigure}
        \begin{subfigure}[b]{0.49\linewidth}
            \centering
            \includegraphics[width=\linewidth]{Images/LargeDeterministicMasslessScalarq0.jpg}
            \captionsetup{justification=centering}
            \caption{$\mathfrak{q}=0$.}
            \label{fig:LargeDeterministicGFq0}
        \end{subfigure}\hfill
        \begin{subfigure}[b]{0.49\linewidth}
            \centering
            \includegraphics[width=\linewidth]{Images/LargeDeterministicMasslessScalarq10.jpg}
            \captionsetup{justification=centering}
            \caption{$\mathfrak{q}=10$.}
            \label{fig:LargeDeterministicGFq10}
        \end{subfigure}
        \caption{Effect on the transverse gauge field spectrum of the deterministic perturbations %
        \eqref{eq:deterministicpots}
        with size $\norm{V_i}_E=10^{-1}$. The unperturbed QNF is shown in red, while the perturbed QNFs are depicted in blue. In (a) and (b) the dashed blue line represents the circle of radius $10^{-1}$ centered in the unperturbed QNF. In the upper panels of (c) and (d) we present the condition numbers for the lowest QNFs.}
        \label{fig:DeterministicGF}
    \end{figure}

    \begin{figure}[htb]
        \centering
        \begin{subfigure}[b]{0.48\linewidth}
            \centering
            \includegraphics[width=\linewidth]{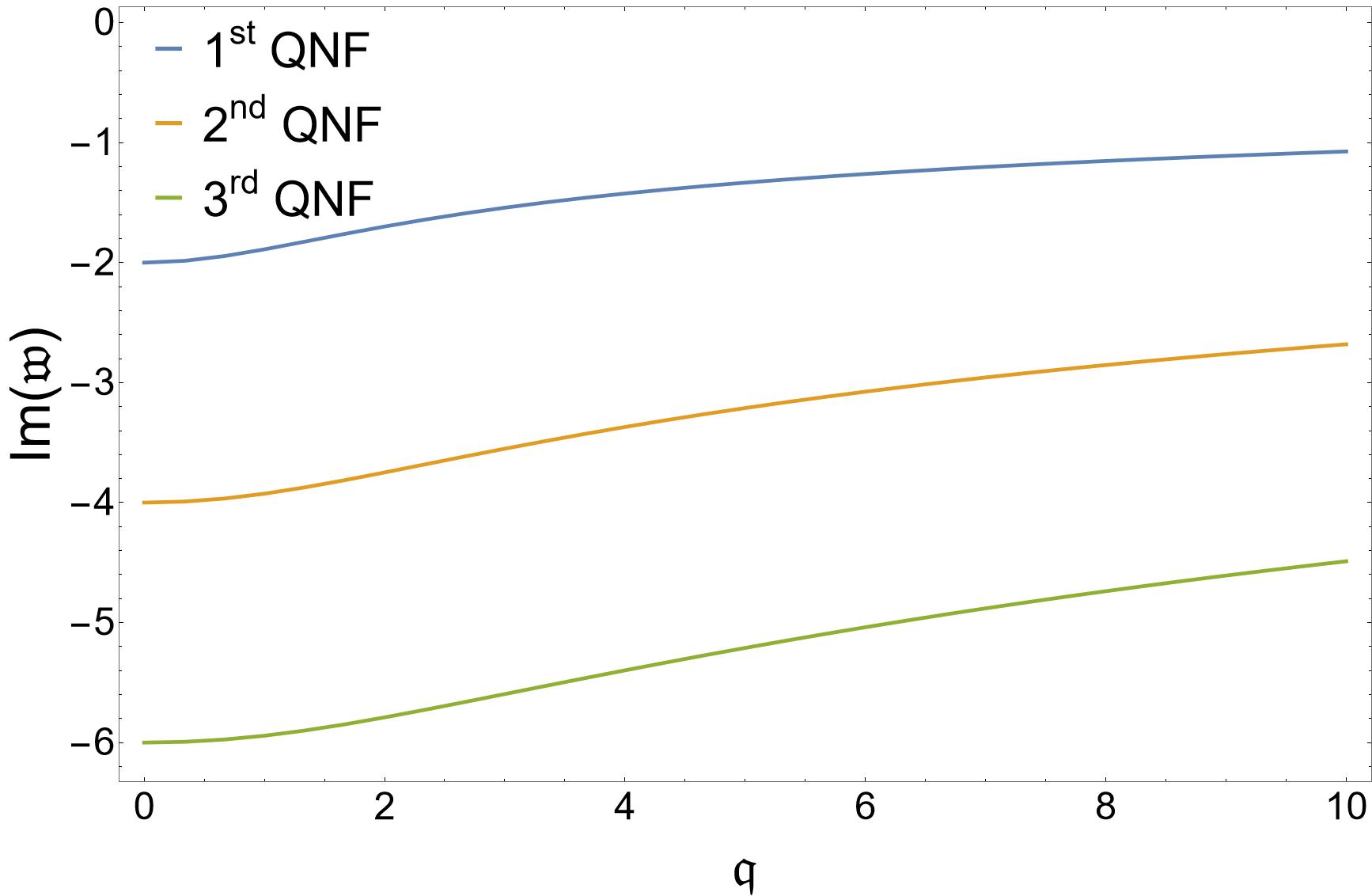}
            \captionsetup{justification=centering}
            \caption{Imaginary part of the QNFs.}
            \label{fig:ImQNFGF}
        \end{subfigure}\hfill
        \begin{subfigure}[b]{0.48\linewidth}
            \centering
            \includegraphics[width=\linewidth]{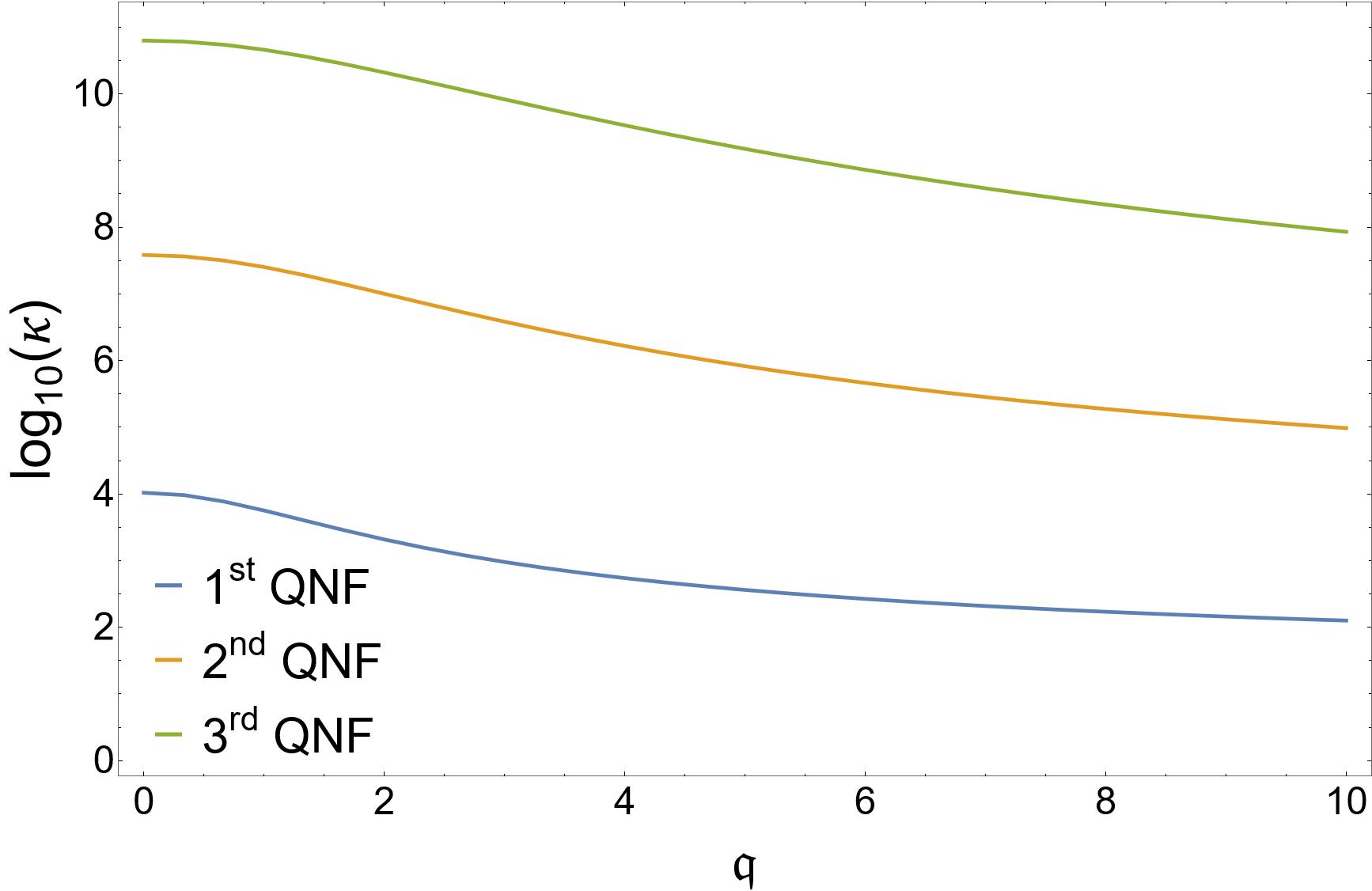}
            \captionsetup{justification=centering}
            \caption{Condition numbers.}
            \label{fig:ConditionNumbersGF}
        \end{subfigure}
        \caption{Momentum dependence in the energy norm for the transverse gauge field.}
        \label{fig:qDependenceGF}
    \end{figure}

First, in figures~\ref{fig:CloseupPseudospectraGF} and~\ref{fig:LargePseudospectraGF} we display the pseudospectrum  of the transverse gauge field. The presence of open contour lines around the QNFs indicates that they are unstable under generic perturbations. 
However, 
for the first QNF we have not found statistical evidence of its instability under random local potential perturbations (see the selective pseudospectra in figure~\ref{fig:CloseupPseudoGFq0}). It is interesting to note that this is similar to what we found for the scalar field of mass $m^2 l^2=-3$ (see figure~\ref{fig:CloseupPseudoMassiveScalarq0}). From the conformal field theory point of view both of these fields correspond to operators of conformal dimension three.

In order to further characterize the nature of the instability, in figure~\ref{fig:DeterministicGF} we show the effect of the deterministic perturbations \eqref{eq:deterministicpots} on the QNF spectrum. 
At zero momentum we find that the spectrum is unstable under the perturbations \eqref{eq:determinisiticV2} and \eqref{eq:determinisiticV3}. In particular in figure~\ref{fig:CloseupDeterministicGFq0} the first QNF is shown to be unstable under near-horizon perturbations \eqref{eq:determinisiticV3}. Moreover, from studying the evolution of the condition numbers shown in that figure, one can conclude that near-horizon and short $\rho$-wavelength perturbations result in less unstable QNFs.

As in the previous section, it is also interesting to analyze the dependence of the stability of the QNFs on the momentum of the fluctuation.
Therefore, in figure~\ref{fig:qDependenceGF} we plot the momentum dependence of the QNFs and their condition numbers. The QNFs get closer to the real axis as  momentum increases and, as for the real scalar, we observe that their condition numbers decrease in the process (although staying above one for the range of momentum studied).

    Despite the qualitative similarities, it is worth noting that the actual values of the pseudospectra do not match the cases studied in the previous subsection. Specifically, we observe that the gauge field is more stable than the massless scalar and less stable than the scalar with mass $m^2l^2=-3$. This agrees with the instability being related to the imaginary part of the QNFs. The gauge field QNFs are closer to the real axis than those of the massless scalar and further away than those of the scalar with mass $m^2l^2=-3$.

Finally, it is illustrative to analyze the pseudospectrum in the $L^2$-norm. We present our results in appendix~\ref{app:l2norm}. Again, we observe that in the $L^2$-norm the stability is enhanced under local potential perturbations and reduced under generic perturbations.

    \section{Conclusions}\label{section:Conclusions}

    In this work we have initiated the study of the spectral stability of quasinormal modes in asymptotically AdS black hole geometries. Gauge/gravity duality maps the frequencies of those modes (QNFs) to the spectrum of collective excitations of a strongly coupled quantum many-body system. Therefore, probing the spectral stability of QNFs is akin to probing the stability of the spectra of strongly coupled holographic field theories.

    Following previous works
    in asymptotically flat \cite{Jaramillo:2020tuu,Destounis:2021lum} and de Sitter spacetimes \cite{Sarkar:2023rhp}, we probed spectral stability through pseudospectra and condition numbers. To construct the pseudospectra we needed to choose a physically motivated norm and cast the equations of motion for the quasinormal modes in the form of an eigenvalue problem. Consequently, in section \ref{section:Quasinormal Modes and Pseudospectra} we proposed a generic approach valid for static black brane backgrounds, a family of spacetimes very prominent in the literature. Our approach relied on two cornerstones. First, the regular coordinates, which enabled us to write the equations of motion as an eigenvalue problem while also translating outgoing boundary conditions to regularity conditions on the horizon. And second, the energy norm, which represents the physically relevant notion of size for the QNMs.

    After discussing the general approach, we proceeded to study the stability of a real scalar field and a transverse gauge field in SAdS$_{4+1}$. As previously observed in asymptotically flat \cite{Jaramillo:2020tuu,Destounis:2021lum} and dS spacetimes \cite{Sarkar:2023rhp}, we found that the QNFs are unstable and that the instability is associated with the existence of a horizon. Moreover, we observed that QNFs are less unstable
    the closer they are to the real axis, a feature also present in the aforementioned spacetimes.

    Regarding the scalar, we found an increase in stability at high momenta and small masses. This implies that in the dual QFT the instability is milder 
    for high-momentum fluctuations of operators with small mass dimension.
     In order to further study the nature of the instability, we probed the spectrum of QNFs with local potential perturbations which could arise as effective interactions representing deviations from the exact SAdS$_{4+1}$ background. Remarkably, we inferred that only a small portion of the total instability is associated with these perturbations.
    Moreover, we concluded that small perturbations cannot drive the QNFs to the upper half of the complex plane.
    We also found that for small momenta, the first QNF is unstable, with the instability arising from localized perturbations near the horizon. 
    Thus, small deformations of the IR can destabilize even the lowest QNF.

    With respect to the gauge field, we concluded that its pseudospectrum is qualitatively very similar to that of the scalar field. This is even more pronounced for the scalar of mass $m^2l^2=-3$. In that case both the scalar and the transverse gauge field correspond to operators with the same conformal dimension. The resemblance between the pseudospectra of the scalar and the gauge field might be due to the fact that the spectra of both fields are very similar. Exploring this potential connection would indicate a degree of universality, which would be intriguing to investigate in future research.

        We have proven that AdS spacetimes also present spectral instability associated with the existence of an event horizon. This supports the physical intuition that the observed instability is a direct consequence of the non-conservative character of black holes and thus is independent of the asymptotic behavior of the spacetime. 
    
     With regard to gauge/gravity duality, our results imply that thermal excitation spectra of strongly coupled quantum field theories dual to SAdS$_{4+1}$ (or some small deviation from it) are unstable, \textit{i.e.}, small perturbations to the original model have a large effect on the spectra. This suggests that a mathematical model would not be able to accurately capture the actual spectra of the real physical system. That said, it is important to remark that the spectral instability decreases with the decay width of the fluctuations. We expect the imprint of the aforementioned instability to be less pronounced for modes near the real axis. In a practical setting it implies that, at most, one might hope to be able to capture the leading order behavior (dominated by the longest-lived excitation). 
    On the other hand, it might not be possible to model the subleading features of decay trustfully as the higher QNFs are increasingly unstable.

   In the realm of gauge/gravity duality, pseudospectral analysis in asymptotic AdS geometries has many applications with potentially deep implications. We list some of them in the following.

    \begin{itemize}
        \item Study the stability of hydrodynamic modes.
        As we have seen, the closer a QNF is to the real axis, the less unstable it becomes. In particular that might indicate that hydrodynamic modes enjoy 
        privileged stability properties in the pseudospectrum. On the other hand, if this effect is not sufficiently pronounced it would be interesting to see whether small local perturbations at zero momentum could drive the hydro mode's QNF to the upper half plane indicating instabilities of flow patterns. In fact such instabilities have been argued to arise in Navier-Stokes equations and explain the early onset of turbulence \cite{trefethen1993hydrodynamic}.
        
        \item Study the stability of the collision between the hydro mode's QNF and the first non-hydro mode's QNF. This would shed a new perspective on the validity of the hydrodynamic approximation. Currently, the hydrodynamic expansion is postulated to be valid up to the energy where the aforementioned collision happens \cite{Withers:2018srf,GrozdanovKovtun:2019,Grozdanov:2019}. Consequently, finding instability would indicate that the limit of validity itself is unstable.
        A similar comment applies to pole skipping points. These are points at which the residue of a pole of a holographic Green's function vanishes \cite{Amado:2007yr,Amado:2008ji}.
        Such points are of particular interest in relation to quantum chaos \cite{Grozdanov:2017ajz, Blake:2017ris, Blake:2018leo}.

        \item Study the spectral stability at the phase transition in the holographic superconductor model \cite{Hartnoll:2008vx}. It would be interesting to see whether the critical temperature is stable under the generic perturbations captured by the pseudospectrum.

        \item Study the spectral stability of AdS Reissner-Nordström black branes. The pseudospectrum analysis of charged black branes might have important consequences for the dual description of quantum critical phases of matter~\cite{zaanen2015holographic,Hartnoll:2018xxg}.
        
        \item Study the spectral stability of other solutions frequently used in holography. Beyond completing the current literature of QNFs in AdS spacetimes, it would be interesting to explore the existence of some degree of universality in the pseudospectrum. In particular, as spectra in AdS typically share the "Christmas Tree" structure observed in the current work (see, for instance, \cite{Kovtun:2005ev,Amado:2009ts}), this study would answer whether the overarching features observed here are fundamental to the structure of the spectrum or depend heavily on the particular setup.
    
        \item It has recently been suggested that quasinormal modes can serve as a signature for having successfully simulated (quantum) gravity on a (quantum) computer \cite{Biggs:2023sqw, Maldacena:2023acv}. In view of our results, this idea suffers from similar problems as the idea of black hole spectroscopy from gravitational wave signals. If such (quantum) simulations can be realized, the task of identifying the quasinormal modes will be subject to similar difficulties as for astrophysical black holes (see \cite{Jaramillo:2021tmt} for a relevant discussion of this issue).

        \item A different way of studying the poles of holographic Green's functions is to look for complex solutions for the momenta upon fixing the frequency to be a real number \cite{Amado:2007pv}. 
        This is particularly important for the diffusive mode and has implications for causality \cite{Landsteiner:2012gn} (see also the recent preprint \cite{Gavassino:2023mad}). It would be interesting to study the pseudospectra of these complex momentum modes.

    \end{itemize}

\section*{Acknowledgements}
We thank J.L. Jaramillo and V. Boyanov for valuable discussions, and the anonymous referee for the constructive comments which led to the inclusion of section \ref{ssec:Spectral and Spacetime Instability}.
The work of D.A and K.L.
is supported through the grants CEX2020-001007-S and PID2021-123017NB-100, PID2021-127726NB-I00 funded by MCIN/AEI/10.13039/501100011033 and by ERDF ``A way of making Europe''. The work of D.G.F. is supported by JAEIntroICU-2022-IFT-02.

\appendix
    
    \section{Computation of \texorpdfstring{$L^\dagger$}{Ldagger}}\label{app:Computation of Ldagger}
    Here we present a detailed computation of the adjoint operator $L^\dagger$ for the real scalar field (equation \eqref{eq:Ldag Scalar}) and transverse gauge field (equation \eqref{eq:Ldag GF}). 
    \subsection{Real Scalar Field}\label{subannex:Computation of Ldagger scalar}
    For the real scalar field, the differential operator $L$ is given by
    \begin{equation}
        L= i\begin{pmatrix}
        0&1 \\
        L_1\left[\partial_\mathscr{\rho}^2,\partial_\mathscr{\rho};\boldsymbol{\mathfrak{q}},\rho\right]&L_2\left[\partial_\rho;\boldsymbol{\mathfrak{q}},\rho\right]  
        \end{pmatrix}\,,
    \end{equation}
    with
    \begin{align}
        L_1\left[\partial_\rho^2,\partial_\rho;\boldsymbol{\mathfrak{q}},\rho\right]&=\left[\mathscr{f}(\rho)-2\right]^{-1}\left[\frac{m^2 l^2}{(1-\rho)^2}+\boldsymbol{\mathfrak{q}}^2-(1-\rho)^3\left(\frac{\mathscr{f}(\rho) }{(1-\rho)^3}\right)'\partial_{\rho} -\mathscr{f}(\rho)\partial_{\rho}^2 \right]\,,\\[\medskipamount]
        L_2\left[\partial_\rho;\boldsymbol{\mathfrak{q}},\rho\right]&=\left[\mathscr{f}(\rho)-2\right]^{-1}\left[(1-\rho)^3\left(\frac{\mathscr{f}(\rho)-1}{(1-\rho)^3}\right)' +2 \left(\mathscr{f}(\rho)-1\right)\partial_{\rho}  \right]\,,
    \end{align}
    and the inner product induced by the energy norm is:
    \begin{align}
        \expval{\phi_1,\phi_2}_E=&\int_0^1\frac{d\rho}{(1-\rho)^3}\Biggl[\left(\frac{m^2l^2}{(1-\rho)^2}+\boldsymbol{\mathfrak{q}}^2 \right)\overline{\phi_1}(\mathscr{t},\mathbf{k},\rho)\phi_2(\mathscr{t},\mathbf{k},\rho)+ \nonumber\\
        &+\mathscr{f}(\rho)\partial_\rho\overline{\phi_1}(\mathscr{t},\mathbf{k},\rho)\partial_\rho\phi_2(\mathscr{t},\mathbf{k},\rho)-\left(\mathscr{f}(\rho)-2\right)\overline{\psi_1}(\mathscr{t},\mathbf{k},\rho)\psi_2(\mathscr{t},\mathbf{k},\rho)\Biggr]\,.
    \end{align}

    For notational convenience, we drop the explicit dependence of all functions and denote $\rho$ derivatives with a prime. 

    Now, recalling the equality:
    \begin{equation}
        \expval{L^\dagger\left[\phi_1\right],\phi_2}_E=\expval{\phi_1,L\left[\phi_2\right]}_E\,,
    \end{equation} 
    we can compute $L^\dagger$ integrating by parts
    
    \begingroup
    \allowdisplaybreaks
    \begin{align}\label{eq:Ldag Scalar 1}
        \expval{\phi_1,L\left[\phi_2\right]}_E=&\,i\int_0^1\frac{d\rho}{(1-\rho)^3}\Biggl[\left(\frac{m^2l^2}{(1-\rho)^2}+\boldsymbol{\mathfrak{q}}^2 \right)\overline{\phi_1}\psi_2+\overline{\phi_1}'\mathscr{f} \psi_2'- \overline{\psi_1}\left\{\left(\frac{m^2 l^2}{(1-\rho)^2}+\boldsymbol{\mathfrak{q}}^2\right)\phi_2\right.\nonumber\\
        &\left. -(1-\rho)^3\left(\frac{\mathscr{f} \phi_2'}{(1-\rho)^3}\right)'+(1-\rho)^3\left(\frac{(\mathscr{f}-1)\psi_2}{(1-\rho)^3}\right)'+\left(\mathscr{f}-1\right)\psi_2' \right\}\Biggr]\nonumber\\
        =&-i\int_0^1\frac{d\rho}{(1-\rho)^3}\Biggl[\left(\frac{m^2l^2}{(1-\rho)^2}+\boldsymbol{\mathfrak{q}}^2 \right)\phi_2\overline{\psi_1}+\phi_2'\mathscr{f}\overline{\psi_1}' - \psi_2\left\{\left(\frac{m^2 l^2}{(1-\rho)^2}+\boldsymbol{\mathfrak{q}}^2\right)\overline{\phi_1}\right. \nonumber\\
        &\left. -(1-\rho)^3\left(\frac{\mathscr{f} \overline{\phi_1}'}{(1-\rho)^3}\right)'+(1-\rho)^3\left(\frac{(\mathscr{f}-1)\overline{\psi_1}}{(1-\rho)^3}\right)'+\left(\mathscr{f}-1\right)\overline{\psi_1}' \right\}\Biggr]\nonumber\\
        &+\left. i\frac{\mathscr{f}\overline{\phi_1}'\psi_2}{(1-\rho)^3}\right|_{\rho=0}^{\rho=1}+\left. i\frac{\mathscr{f}\overline{\psi_1}\phi_2'}{(1-\rho)^3}\right|_{\rho=0}^{\rho=1}-\left. 2i\frac{(\mathscr{f}-1)\overline{\psi_1}\psi_2}{(1-\rho)^3}\right|_{\rho=0}^{\rho=1}\nonumber\\
        =&\expval{L\left[\phi_1\right],\phi_2}_E+\left. i\frac{\mathscr{f}\overline{\phi_1}'\psi_2}{(1-\rho)^3}\right|_{\rho=0}^{\rho=1}+\left. i\frac{\mathscr{f}\overline{\psi_1}\phi_2'}{(1-\rho)^3}\right|_{\rho=0}^{\rho=1}-\left. 2i\frac{(\mathscr{f}-1)\overline{\psi_1}\psi_2}{(1-\rho)^3}\right|_{\rho=0}^{\rho=1}\nonumber\\
        =&\expval{L^\dagger\left[\phi_1\right],\phi_2}_E\,.
    \end{align}
    \endgroup
    
    The boundary terms vanish for $\rho=1$ as the chosen function space satisfies the boundary condition \eqref{eq:function space scalar}. However at $\rho=0$, we have $\mathscr{f}(0)=0$, and we get a non-zero contribution arising from the last term. Then, the final expression for \eqref{eq:Ldag Scalar 1} is given by:
    \begin{align}
        \expval{\phi_1,L\left[\phi_2\right]}_E=&\expval{L\left[\phi_1\right],\phi_2}_E-\left. 2i\frac{(\mathscr{f}-1)\overline{\psi_1}\psi_2}{(1-\rho)^3}\right|_{\rho=0}\nonumber\\
        =&\expval{L\left[\phi_1\right],\phi_2}_E-i\int_0^1\frac{d\rho}{(1-\rho)^3}(\mathscr{f}-2)\left[2\delta(\rho)\frac{\mathscr{f}(\rho)-1}{\mathscr{f}(\rho)-2}\right]\overline{\psi_1}\psi_2\nonumber\\
        =&\expval{L\left[\phi_1\right],\phi_2}_E + \expval{\delta L\left[\phi_1\right],\phi_2}_E =\expval{L^\dagger\left[\phi_1\right],\phi_2}_E\,,
    \end{align}
    from where we recover the expression for the adjoint given in equation \eqref{eq:Ldag Scalar}:
    \begin{equation}
        L^\dagger=L+\delta L=L+\begin{pmatrix}
        0&0\\[\medskipamount]
        0&-2i\delta(\rho)\frac{\mathscr{f}(\rho)-1}{\mathscr{f}(\rho)-2}
    \end{pmatrix}\,.
    \end{equation}

    \subsection{Transverse Gauge Field}\label{subannex:Computation of Ldagger GF}
    For the transverse gauge field, the differential operator $L$ is given by
    \begin{equation}
        L= i\begin{pmatrix}
        0&1 \\
        L_1\left[\partial_\mathscr{\rho}^2,\partial_\mathscr{\rho};\boldsymbol{\mathfrak{q}},\rho\right]&L_2\left[\partial_\rho;\boldsymbol{\mathfrak{q}},\rho\right]  
        \end{pmatrix}\,,
    \end{equation}
    with
    \begin{align}
    L_1\left[\partial_\rho^2,\partial_\rho;\boldsymbol{\mathfrak{q}},\rho\right]&=\left[\mathscr{f}(\rho)-2\right]^{-1}\left[\boldsymbol{\mathfrak{q}}^2-(1-\rho)\left(\frac{\mathscr{f}(\rho) }{1-\rho}\right)'\partial_{\rho} -\mathscr{f}(\rho)\partial_{\rho}^2 \right]\,,\\[\medskipamount]
    L_2\left[\partial_\rho;\boldsymbol{\mathfrak{q}},\rho\right]&=\left[\mathscr{f}(\rho)-2\right]^{-1}\left[(1-\rho)\left(\frac{\mathscr{f}(\rho)-1}{(1-\rho)}\right)' +2 \left(\mathscr{f}(\rho)-1\right)\partial_{\rho}  \right]\,,
    \end{align}
    and the inner product induced by the energy norm is:
    \begin{align}
        \expval{a_1,a_2}_E=&\int_0^1\frac{d\rho}{(1-\rho)}\Bigl[\boldsymbol{\mathfrak{q}}^2 \overline{a_1}(\mathscr{t},\mathbf{k},\rho)a_2(\mathscr{t},\mathbf{k},\rho)+\mathscr{f}(\rho)\partial_\rho\overline{a_1}(\mathscr{t},\mathbf{k},\rho)\partial_\rho a_2(\mathscr{t},\mathbf{k},\rho) \nonumber\\
        &-\left(\mathscr{f}(\rho)-2\right)\overline{\alpha_1}(\mathscr{t},\mathbf{k},\rho)\alpha_2(\mathscr{t},\mathbf{k},\rho)\Bigr]\,.
    \end{align}

    As with the scalar, we drop the explicit dependence of all functions, denote $\rho$ derivatives with a prime and compute $L^\dagger$ integrating by parts:
    \begin{align}\label{eq:Ldag GF 1}
        \expval{a_1,L\left[a_2\right]}_E=&i\int_0^1\frac{d\rho}{(1-\rho)^3}\Biggl[\boldsymbol{\mathfrak{q}}^2 \overline{a_1}\alpha_2+\overline{a_1}'\mathscr{f} \alpha_2'- \overline{\alpha_1}\left\{\boldsymbol{\mathfrak{q}}^2 a_2-(1-\rho)\left(\frac{\mathscr{f} a_2'}{(1-\rho)}\right)'\right. \nonumber\\
        &\left. +(1-\rho)\left(\frac{(\mathscr{f}-1)\alpha_2}{(1-\rho)}\right)'+\left(\mathscr{f}-1\right)\alpha_2' \right\}\Biggr]\nonumber\\
        =&-i\int_0^1\frac{d\rho}{(1-\rho)^3}\Biggl[\boldsymbol{\mathfrak{q}}^2 a_2\overline{\alpha_1}+a_2'\mathscr{f} \alpha_1'- \alpha_2\left\{\boldsymbol{\mathfrak{q}}^2 \overline{a_1}-(1-\rho)\left(\frac{\mathscr{f} \overline{a_1}'}{(1-\rho)}\right)'\right. \nonumber\\
        &\left. +(1-\rho)\left(\frac{(\mathscr{f}-1)\overline{\alpha_1}}{(1-\rho)}\right)'+\left(\mathscr{f}-1\right)\overline{\alpha_1}' \right\}\Biggr]\nonumber\\
        &+\left. i\frac{\mathscr{f}\overline{a_1}'\alpha_2}{(1-\rho)}\right|_{\rho=0}^{\rho=1}+\left. i\frac{\mathscr{f}\overline{\alpha_1}a_2'}{(1-\rho)}\right|_{\rho=0}^{\rho=1}-\left. 2i\frac{(\mathscr{f}-1)\overline{\alpha_1}\alpha_2}{(1-\rho)}\right|_{\rho=0}^{\rho=1}\nonumber\\
        =&\expval{L\left[a_1\right],a_2}_E+\left. i\frac{\mathscr{f}\overline{a_1}'\alpha_2}{(1-\rho)}\right|_{\rho=0}^{\rho=1}+\left. i\frac{\mathscr{f}\overline{\alpha_1}a_2'}{(1-\rho)}\right|_{\rho=0}^{\rho=1}-\left. 2i\frac{(\mathscr{f}-1)\overline{\alpha_1}\alpha_2}{(1-\rho)}\right|_{\rho=0}^{\rho=1}\nonumber\\
        =&\expval{L^\dagger\left[a_1\right],a_2}_E\,.
    \end{align}
The boundary terms vanish for $\rho=1$ as the chosen function space satisfies the boundary condition \eqref{eq:function space GF}. Nonetheless, at $\rho=0$, we have $\mathscr{f}(0)=0$, and we get a non-zero contribution arising from the last term. Then, the final expression for \eqref{eq:Ldag GF 1} is given by:
    \begin{align}
        \expval{a_1,L\left[a_2\right]}_E=&\expval{L\left[a_1\right],a_2}_E-i\int_0^1\frac{d\rho}{(1-\rho)}(\mathscr{f}-2)\left[2\delta(\rho)\frac{\mathscr{f}(\rho)-1}{\mathscr{f}(\rho)-2}\right]\overline{\alpha_1}\alpha_2\nonumber\\
        =&\expval{L\left[a_1\right],a_2}_E + \expval{\delta L\left[a_1\right],a_2}_E =\expval{L^\dagger\left[a_1\right],a_2}_E\,,
    \end{align}
    from where we recover the expression for the adjoint given in equation \eqref{eq:Ldag GF}:
    \begin{equation}
        L^\dagger=L+\delta L=L+\begin{pmatrix}
        0&0\\[\medskipamount]
        0&-2i\delta(\rho)\frac{\mathscr{f}(\rho)-1}{\mathscr{f}(\rho)-2}
    \end{pmatrix}\,.
    \end{equation}

    \section{Discretization in the Chebyshev grid}\label{app:Extra Details on the Discretization.}
    \label{app:discretization}
    In this section we give some more details regarding the discretization process. This discussion is based on the appendix C of \cite{Jaramillo:2020tuu}.
    \subsection{Collocation Method}
    Working in the Chebyshev grid \eqref{eq:Chebyshev grid} with $N+1$ points is equivalent to approximating a function $f(\rho)$ by a power series
    \begin{equation}\label{eq:Chebyshev expansion}
        f(\rho)\approx \sum_{n=0}^{N}c_n T_n(2\rho-1)\,,
    \end{equation}
    where $T_n$ are the Chebyshev polynomials
    \begin{equation}
        T_n(x)=\cos\left(n\arccos\left(x\right)\right)\,.
    \end{equation}

    However, in the numerical method we do not work with the coefficients $\{c_n\}$ but instead with the values of $f(\rho)$ in the points of the grid \eqref{eq:Chebyshev grid}. In order to connect both descriptions, we note that in our grid we have the following orthogonality relation
    \begin{align}
        \int_{0}^{1} \frac{d\rho}{\sqrt{1-(2\rho-1)^2}}\,\,T_m(2\rho-1)T_n(2\rho-1)&=\pi \int_{0}^{1} d \theta\,\,T_m\left(\cos(\theta\pi)\right)T_n\left(\cos(\theta\pi)\right)\nonumber\\ 
        &\approx \frac{\pi}{N}\sum_{j=0}^{N}\frac{T_m\left(\cos\left(\frac{j\pi}{N}\right)\right)T_n\left(\cos\left(\frac{j\pi}{N}\right)\right)}{1+\delta_{j0}+\delta_{jN}}\nonumber\\
        &=\frac{\pi}{N}\sum_{j=0}^{N}\frac{T_m\left(2\rho_j-1\right)T_n\left(2\rho_j-1\right)}{1+\delta_{j0}+\delta_{jN}}\nonumber\\
        &\approx \frac{\pi}{2}\delta_{mn}\left(1+\delta_{n0}+\delta_{nN}\right)\,,
    \end{align}
    which allows us to the approximate the coefficients $\{c_n\}$ as
    \begin{align}\label{eq:cn coefficients Chebyshev}
        c_n&=\frac{2}{1+\delta_{n0}+\delta_{nN}}\int_{0}^{1}d\theta\,\,f\left(\frac{1}{2}-\frac{1}{2}\cos\left(\theta\pi\right)\right) T_n\left(\cos(\theta\pi)\right)\nonumber \\
        &\approx \frac{2/N}{1+\delta_{n0}+\delta_{nN}}\sum_{j=0}^{N}\frac{f\left(\rho_j\right)T_n\left(2\rho_j-1\right)}{1+\delta_{j0}+\delta_{jN}}\,.
    \end{align}

    \subsection{Construction of the \texorpdfstring{$G_E$}{GE} matrix}
    \label{app:gmatrix}
    Knowing how to express the functions in the grid as a sum of Chebyshev polynomials, we can now turn to discussing how to discretize a generic integral
    \begin{equation}\label{eq:generic integral Chebyshev}
        \int_0^{1}d\rho\,\,\overline{f}(\rho)g(\rho)\,.
    \end{equation}

    In order to do so, we first we note that for the Chebyshev polynomials we have
    \begin{equation}
    \int_0^1d\rho\,\,T_n(2\rho-1)=
        \begin{dcases}
            0&  \quad n \text{ odd.}\\[\medskipamount]
            \frac{1}{1-n^2} &   \quad n \text{ even.}
        \end{dcases}
    \end{equation}
    Then, using expression \eqref{eq:cn coefficients Chebyshev}, we can approximate the integral \eqref{eq:generic integral Chebyshev} by
    \begin{align}
        \int_0^{1}d\rho\,\,\overline{f}(\rho)g(\rho)&\approx \sum_{\substack{n=0 \\n \text{ even}}}^N  \frac{1}{1-n^2} \frac{2/N}{1+\delta_{n0}+\delta_{nN}}\sum_{j=0}^{N}\frac{\overline{f}\left(\rho_j\right)g\left(\rho_j\right)T_n\left(2\rho_j-1\right)}{1+\delta_{j0}+\delta_{jN}} \nonumber \\
        &\approx \sum_{j=0}^N \overline{f}\left(\rho_j\right) \left[ \sum_{\substack{n=0 \\n \text{ even}}}^N  \frac{1}{1-n^2} \frac{2/N}{1+\delta_{n0}+\delta_{nN}}\sum_{j=0}^{N}\frac{T_n\left(2\rho_j-1\right)}{1+\delta_{j0}+\delta_{jN}} \right] g\left(\rho_j\right)\nonumber\\
        &= \sum_{i=0}^N \sum_{j=0}^N \mu_{ij} \overline{f}\left(\rho_i\right)  g\left(\rho_k\right)\,,
    \end{align}
    where we have defined the weight matrix $\mu$ as
    \begin{equation}
     \mu_{ij}=\delta_{ij}\sum_{\substack{n=0 \\n \text{ even}}}^N  \frac{1}{1-n^2} \frac{2/N}{1+\delta_{n0}+\delta_{nN}}\sum_{j=0}^{N}\frac{T_n\left(2\rho_j-1\right)}{1+\delta_{j0}+\delta_{jN}}\,. 
    \end{equation}

    With all this, we can now easily construct the matrix $G_E$. Firstly, we factorize the operator $\mathcal{G}$ introduced in equation \eqref{eq:G operator} as a part acting on the left side and another acting on the right side:
    \begin{equation}
        \mathcal{G}\left[\overleftarrow{\partial}_\rho,\overrightarrow{\partial}_\rho;\bold{k},\rho\right]=\begin{pmatrix}                \tilde{\mathcal{H}}_{11}\left[\overleftarrow{\partial}_\rho;\bold{k},\rho\right] &\tilde{\mathcal{H}}_{12}\left[\overleftarrow{\partial}_\rho;\mathbf{k},\rho\right] \\[\medskipamount] \tilde{\mathcal{H}}_{21}\left[\mathbf{k},\rho\right] & \tilde{\mathcal{H}}_{22}\left[\mathbf{k},\rho\right]
        \end{pmatrix}\begin{pmatrix} \mathcal{H}_{11}\left[\overrightarrow{\partial}_\rho;\bold{k},\rho\right] &\mathcal{H}_{12}\left[\mathbf{k},\rho\right] \\[\medskipamount] \mathcal{H}_{21}\left[\overrightarrow{\partial}_\rho;\mathbf{k},\rho\right] & \mathcal{H}_{22}\left[\mathbf{k},\rho\right]
        \end{pmatrix}\,.
    \end{equation}
    Secondly, we disctretize the differential operators $\tilde{\mathcal{H}}_{ab}$ and $\mathcal{H}_{ab}$, obtaining the $(N+1)\times(N+1)$ matrices $\tilde{H}_{ab}$ and $H_{ab}$. And finally, we construct the $G_E$ as
    \begin{equation}
        G_E= \begin{pmatrix} \tilde{H}_{11}&\tilde{H}_{12}\\ \tilde{H}_{21}&\tilde{H}_{22}
        \end{pmatrix}\begin{pmatrix}\mu &0\\[\medskipamount]0&\mu\end{pmatrix}\begin{pmatrix} H_{11}&H_{12}\\ H_{21}&H_{22}
        \end{pmatrix}\,.
    \end{equation}
    Note that, as we wanted, with this definition we first act with $\mathcal{G}$ on the vectors and then integrate (introducing the $\mu$ matrix).

    \subsection{Interpolation Between Grids}
        When integrating as discussed in the previous subsection, we are losing information about the original functions. This can easily be seen from the fact that, despite needing $2(N+1)$ coefficients to describe the approximates of the functions $f(\rho)$ and $g(\rho)$, we have constructed the integral \eqref{eq:generic integral Chebyshev} in a grid of $(N+1)$ points.

        To minimize this effect on $G_E$, we proceed to construct it in a grid with $M+1$ points and then interpolate to the original grid. To achieve this we need to construct the interpolation matrix $\tilde{I}$ which connects both grids. Denoting by $\varrho$ the points of the new grid, $\tilde{I}$ is defined through the following equation
        \begin{equation}
            \sum_{j=0}^{N}\tilde{I}_{ij}f(\rho_j)=f(\varrho_i)=\sum_{n=0}^N \frac{2/N}{1+\delta_{n0}+\delta_{nN}}\sum_{j=0}^{N}\frac{f\left(\rho_j\right)T_n\left(2\rho_j-1\right)}{1+\delta_{j0}+\delta_{jN}} T_n(2\varrho_i-1)\,,
        \end{equation}
        concluding that 
        \begin{equation}
           \tilde{I}_{ij}=\sum_{n=0}^N \frac{2/N}{1+\delta_{n0}+\delta_{nN}}\frac{T_n(2\varrho_i-1) T_n\left(2\rho_j-1\right)}{1+\delta_{j0}+\delta_{jN}}\,.
        \end{equation}

        Then, defining by $G_E^{(M)}$ as the $G_E$ matrix constructed in the grid with $M+1$ points, our final expression for the $G_E$ matrix in the original grid with $(N+1)$ points is:
        \begin{equation}
            G_E= \begin{pmatrix}\tilde{I} &0\\[\medskipamount]0&\tilde{I}\end{pmatrix}^t \,G_E^{(M)}\begin{pmatrix}\tilde{I} &0\\[\medskipamount]0&\tilde{I}\end{pmatrix}\,.
        \end{equation}

 \section{Pseudospectra in the $L^2$-norm}
 \label{app:l2norm}

In this appendix we present results for the pseudospectra of the same models as in the main text, now in the $L^2$-norm
\begin{equation}
    \norm{u}_{L^2}=\int d\rho\,\, u^*(\mathscr{t},\mathbf{k},\rho)\,u(\mathscr{t},\mathbf{k},\rho)\,.
\end{equation}

We shall first consider the real scalar in SAdS$_{4+1}$ introduced in section~\ref{ssec:scalarmodel}.
In figures \ref{fig:CloseupPseudospectraScalarL2} and \ref{fig:LargePseudospectraScalarL2} we present the condition numbers and the full and selective pseudospectra in the $L^2$-norm. We still observe an increasing instability the further away the QNF is from the real axis. The pseudospectra are qualitatively very similar to those observed in the energy norm. However, the overall instability is larger in the $L^2$-norm (in the full pseudospectrum we observe larger regions for the same value of $\varepsilon$). Remarkably, the spectrum is significantly more stable under local potential perturbations. 
    \begin{figure}[!ht]
        \centering
        \begin{subfigure}[b]{0.47\linewidth}
            \centering
            \includegraphics[width=\linewidth]{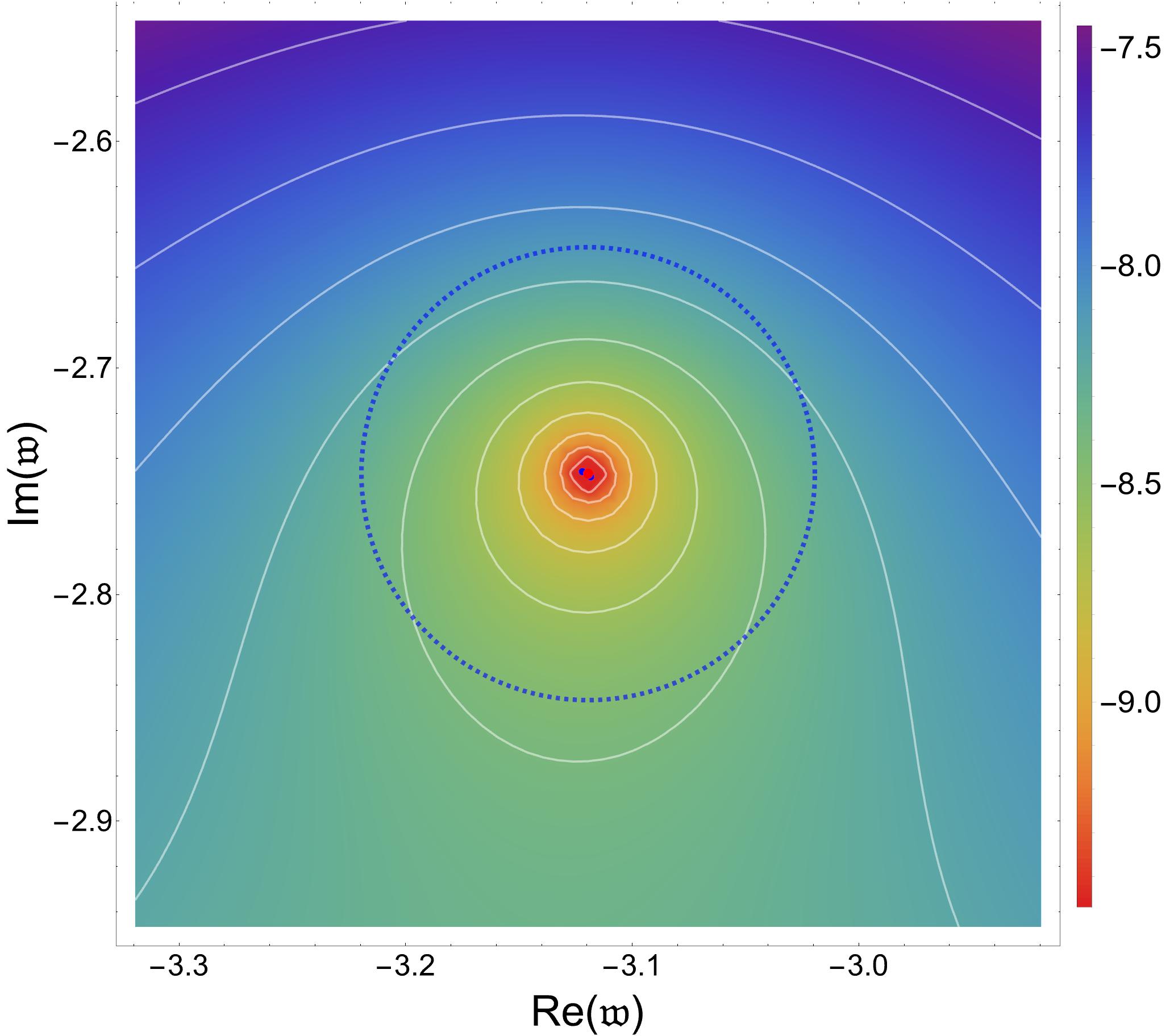}
            \captionsetup{justification=centering}
            \caption{$m^2l^2=0$, $\mathfrak{q}=0$.}
            \label{fig:CloseupPseudoMasslessScalarq0L2}
        \end{subfigure}\hfill
        \begin{subfigure}[b]{0.47\linewidth}
            \centering
            \includegraphics[width=\linewidth]{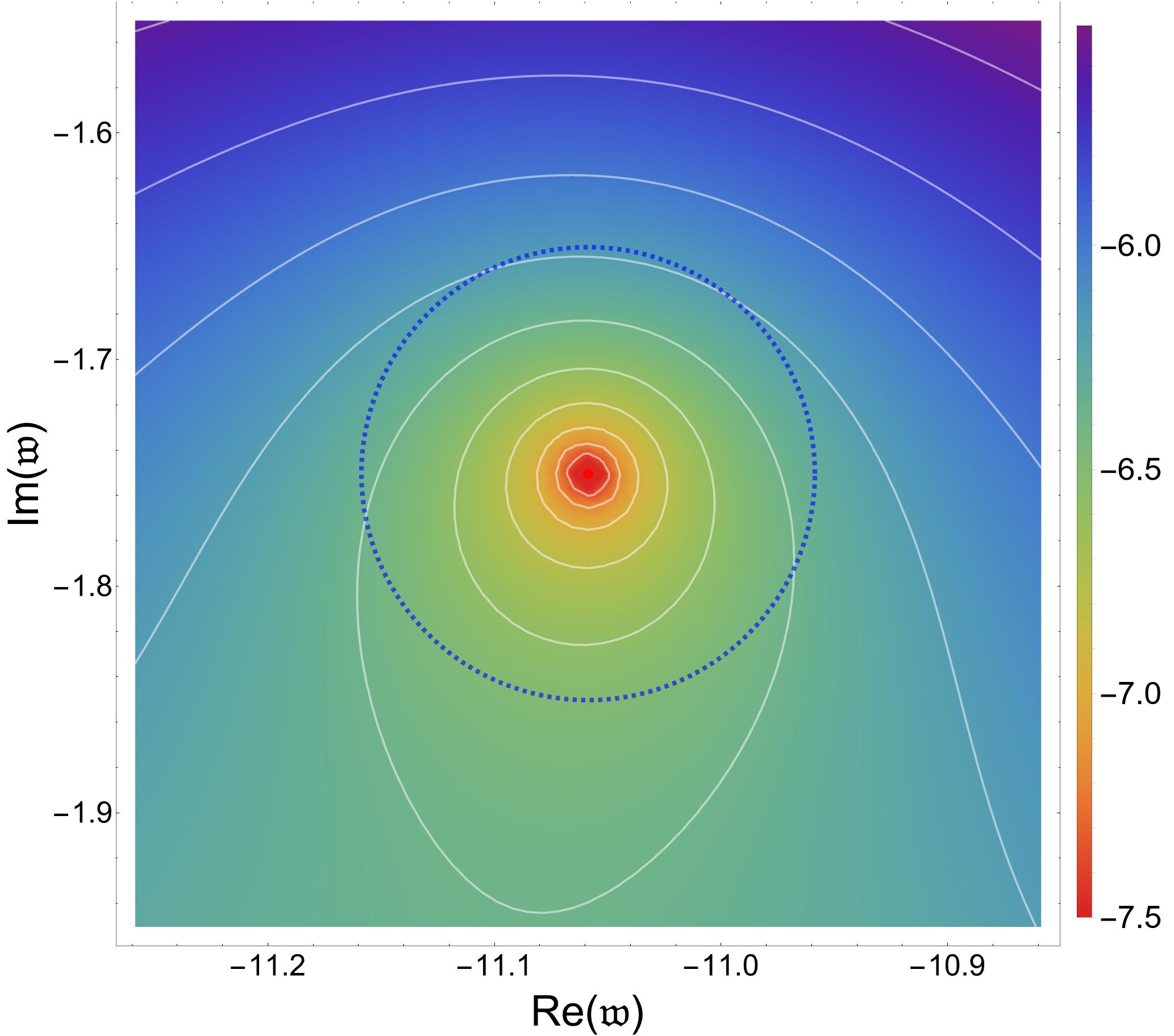}
            \captionsetup{justification=centering}
            \caption{$m^2l^2=0$, $\mathfrak{q}=10$.}
            \label{fig:CloseupPseudoMasslessScalarq10L2}
        \end{subfigure}
        \begin{subfigure}[b]{0.47\linewidth}
            \centering
            \includegraphics[width=\linewidth]{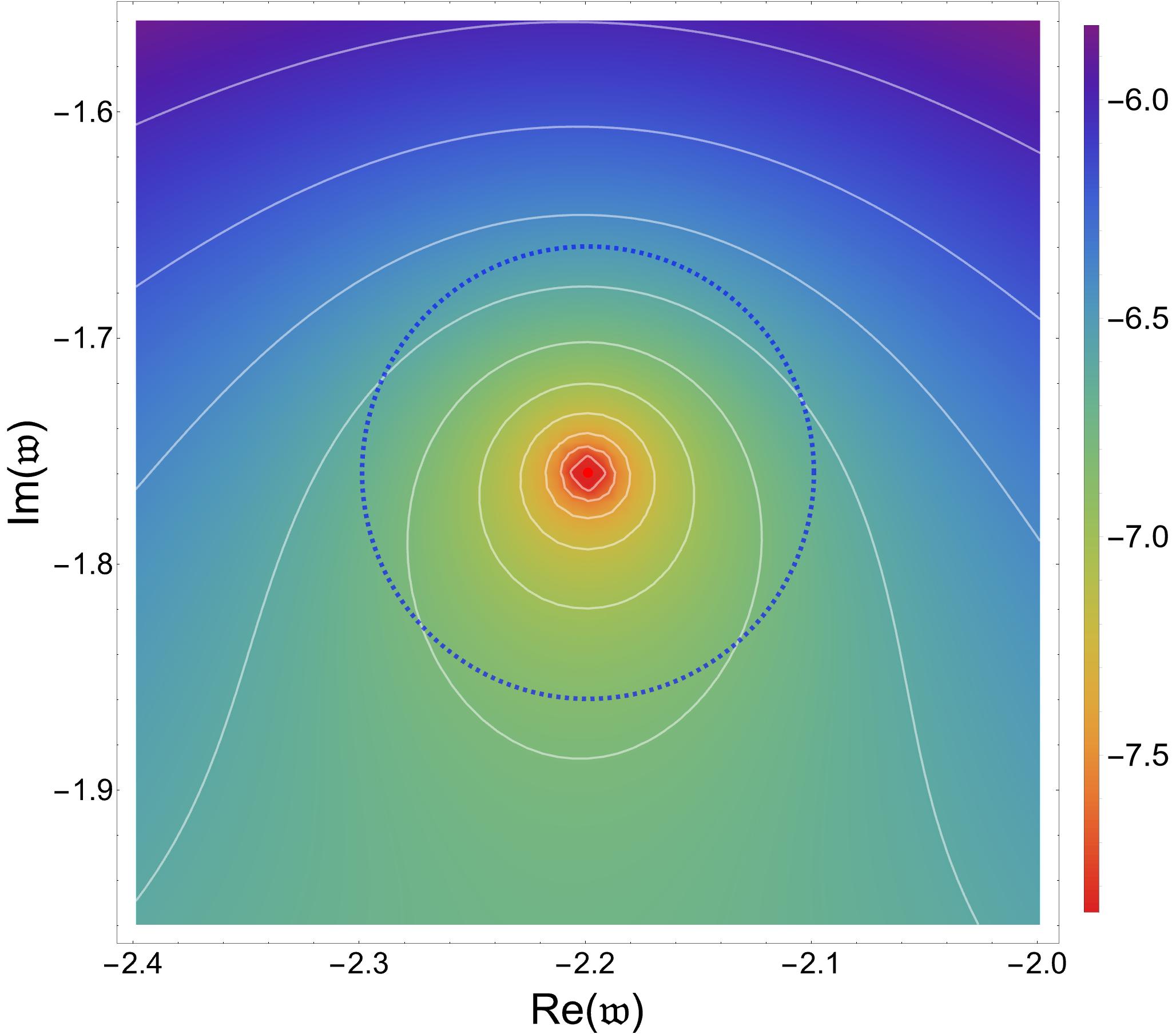}
            \captionsetup{justification=centering}
            \caption{$m^2l^2=-3$, $\mathfrak{q}=0$.}
            \label{fig:CloseupPseudoMassiveScalarq0L2}
        \end{subfigure}\hfill
        \begin{subfigure}[b]{0.47\linewidth}
            \centering
            \includegraphics[width=\linewidth]{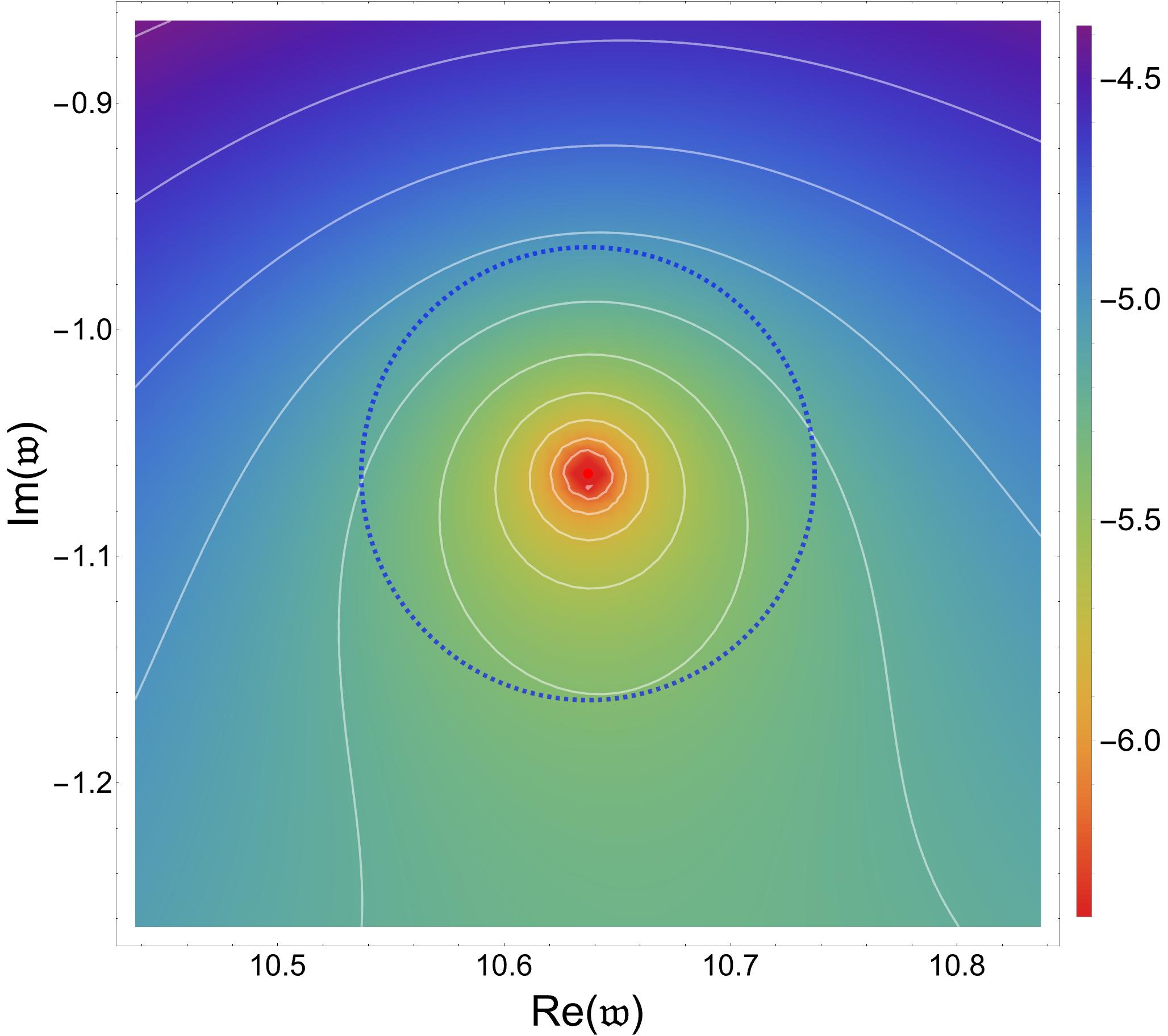}
            \captionsetup{justification=centering}
            \caption{$m^2l^2=-3$, $\mathfrak{q}=10$.}
            \label{fig:CloseupPseudoMassiveScalarq10L2}
        \end{subfigure}
        \caption{Close-up of the scalar pseudospectrum in the $L^2$-norm around the first QNF for different values of $\mathfrak{q}$ and $m^2l^2$. The red dot corresponds to the QNF, the white lines represent the boundaries of various full $\varepsilon$-pseudospectra, and the dashed blue circle symbolizes a circle with a radius of $10^{-1}$ centered on the QNF. The heat map corresponds to the logarithm in base 10 of the inverse of the resolvent. Here the selective pseudospectra computed with random local potential perturbations of size $10^{-1}$ is denoted by blue dots hidden behind the QNF.}
        \label{fig:CloseupPseudospectraScalarL2}
    \end{figure}

    \begin{figure}[!htb]
        \centering
        \begin{subfigure}[b]{0.49\linewidth}
            \centering
            \includegraphics[width=\linewidth]{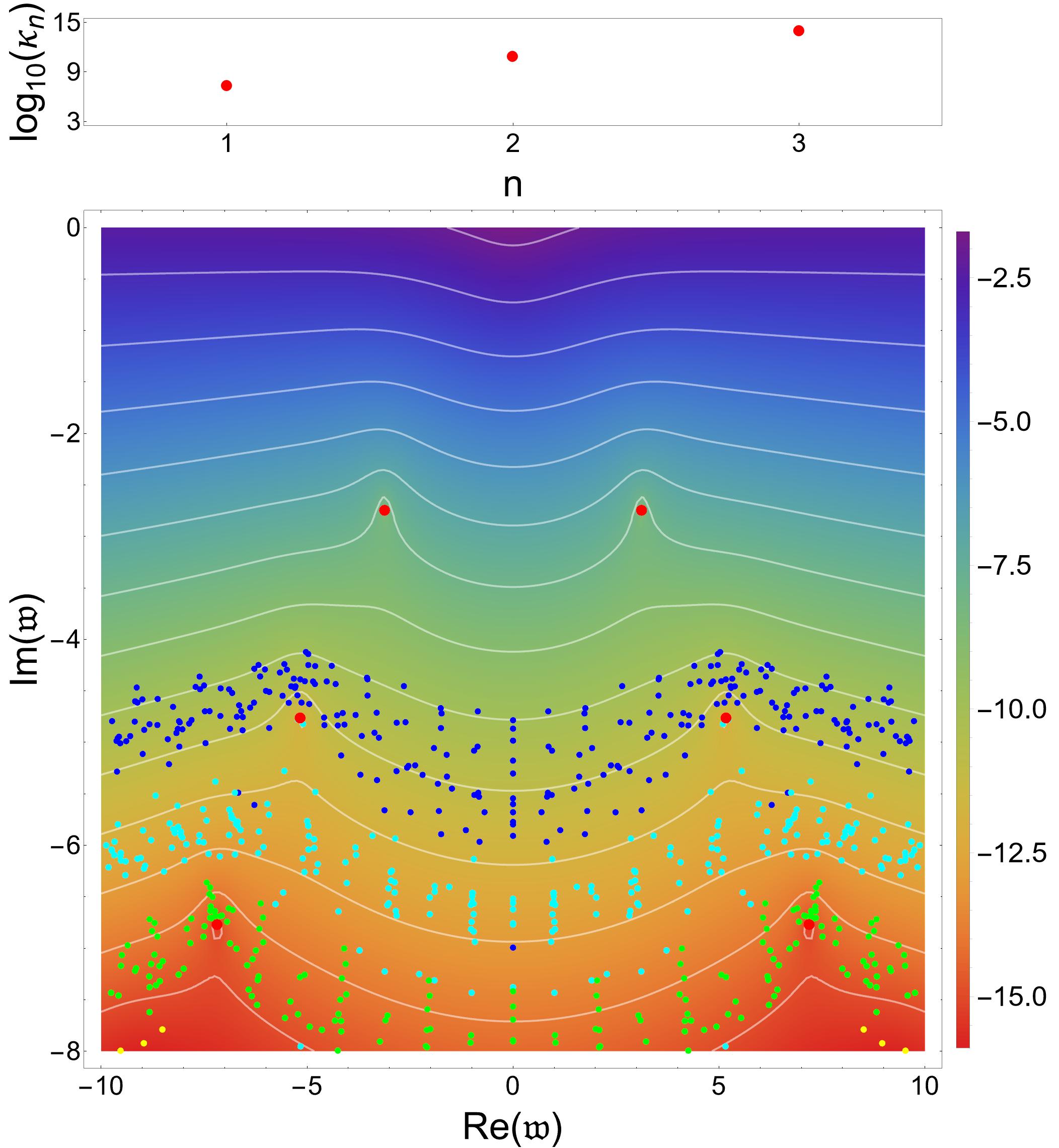}
            \captionsetup{justification=centering}
            \caption{$m^2l^2=0$, $\mathfrak{q}=0$.}
            \label{fig:LargePseudoMasslessScalarq0L2}
        \end{subfigure}\hfill
        \begin{subfigure}[b]{0.49\linewidth}
            \centering
            \includegraphics[width=\linewidth]{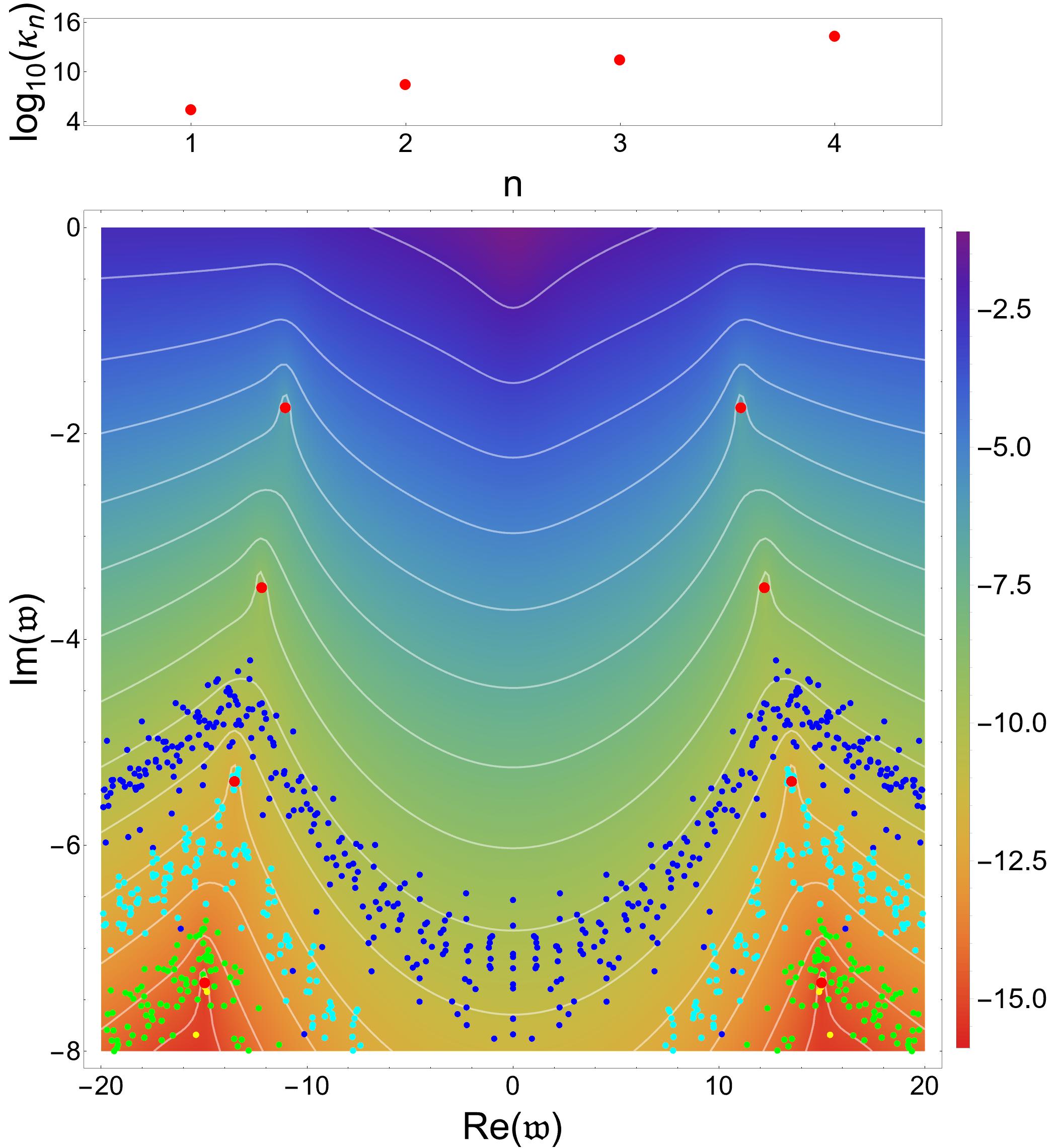}
            \captionsetup{justification=centering}
            \caption{$m^2l^2=0$, $\mathfrak{q}=10$.}
            \label{fig:LargePseudoMasslessScalarq10L2}
        \end{subfigure}
        \begin{subfigure}[b]{0.49\linewidth}
            \centering
            \includegraphics[width=\linewidth]{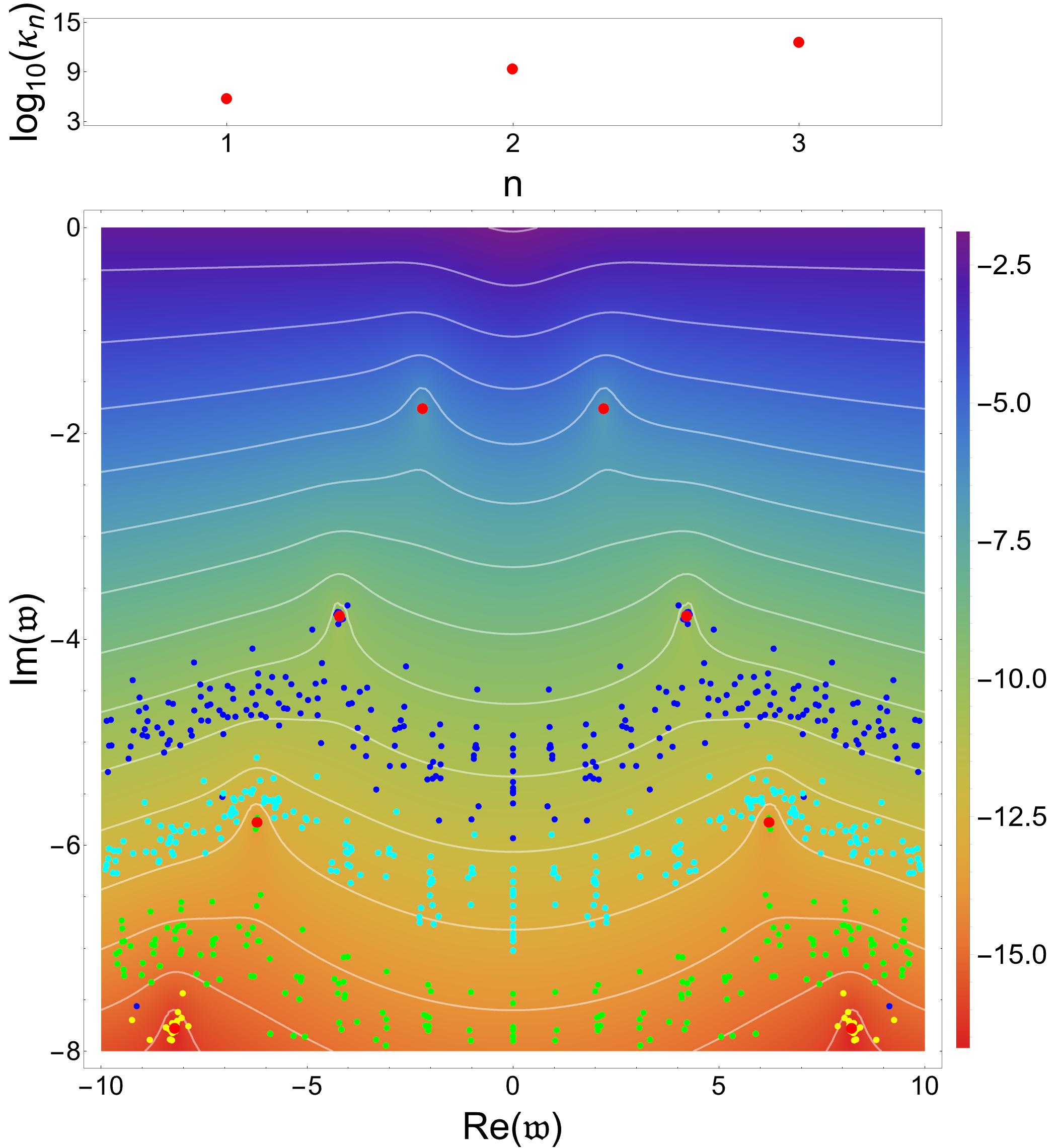}
            \captionsetup{justification=centering}
            \caption{$m^2l^2=-3$, $\mathfrak{q}=0$.}
            \label{fig:LargePseudoMassiveScalarq0L2}
        \end{subfigure}\hfill
        \begin{subfigure}[b]{0.49\linewidth}
            \centering
            \includegraphics[width=\linewidth]{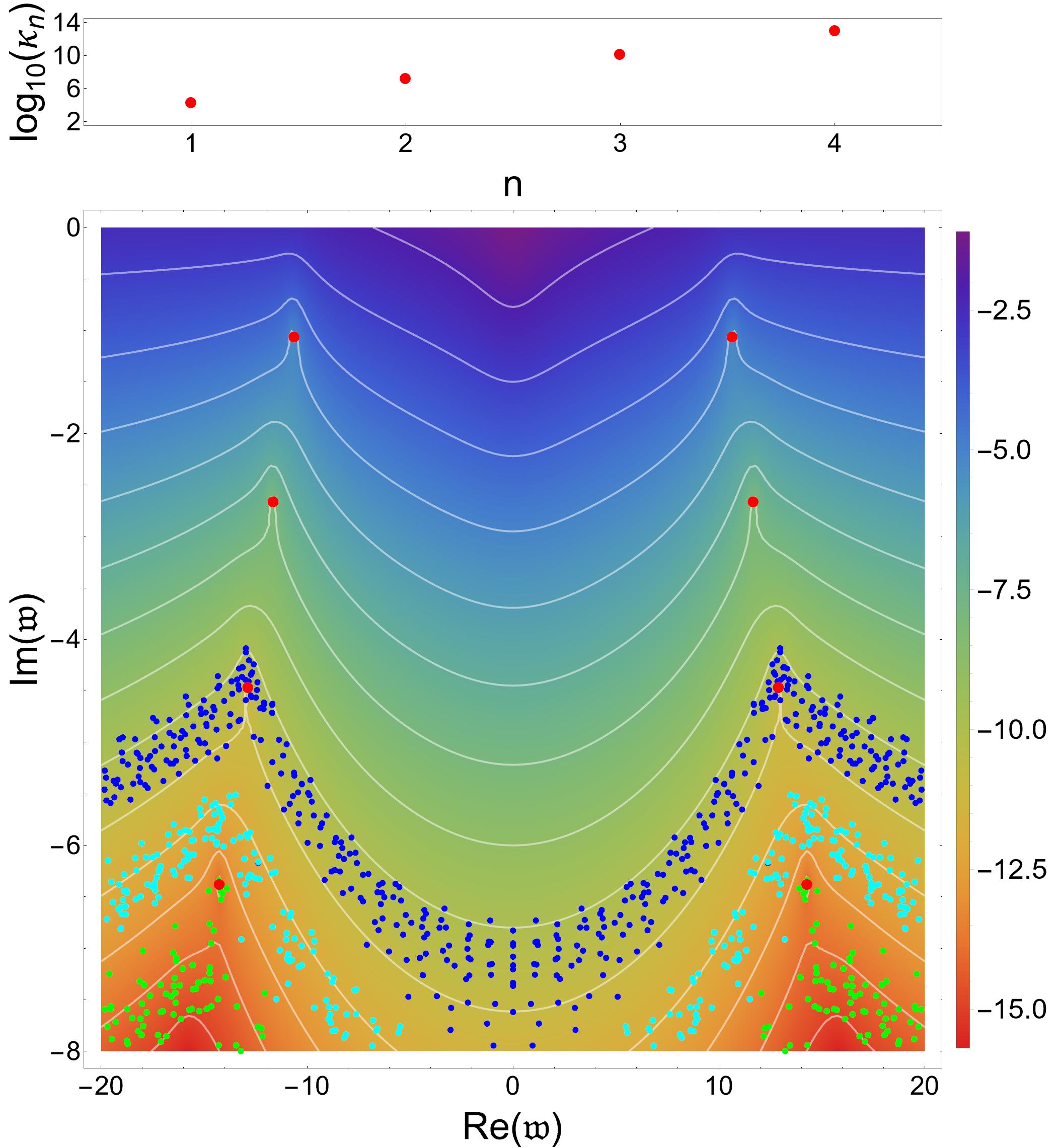}
            \captionsetup{justification=centering}
            \caption{$m^2l^2=-3$, $\mathfrak{q}=10$.}
            \label{fig:LargePseudoMassiveScalarq10L2}
        \end{subfigure}
        \caption{Scalar pseudospectrum in the $L^2$-norm for different values of $\mathfrak{q}$ and $m^2l^2$. In the lower panels, we present selective and full pseudospectra. The red dots represent the QNFs, and the white lines denote the boundaries of different full $\varepsilon$-pseudospectra. The heat map corresponds to the logarithm in base 10 of the inverse of the resolvent, while the blue, cyan, green, and yellow dots indicate different selective $\varepsilon$-pseudospectra computed with random local potential perturbations of size $10^{-1}$, $10^{-3}$, $10^{-5}$, and $10^{-7}$; respectively. In the upper panels, we represent the condition numbers.}
        \label{fig:LargePseudospectraScalarL2}
    \end{figure}

Next we move on to the transverse gauge field of section~\ref{ssec:gaugefield}
and display the results for its  pseudospectrum in the $L^2$-norm.
In figure~\ref{fig:CloseupPseudospectraGFL2} we zoom in on the first QNF, while in figure~\ref{fig:LargePseudospectraGFL2} we show the pseudospectrum down to the position of the fourth QNF. We again observe open contour lines indicating instability under generic perturbations.
As for the scalar field, the qualitative shape of the pseudospectra is very similar to what we found in the energy norm. However, we once again find that in the $L^2$-norm the spectrum is more unstable under generic perturbations and more stable under local potential perturbations.
This becomes clear upon comparing the pseudospectra in figure~\ref{fig:LargePseudospectraGFL2} to those in figure~\ref{fig:LargePseudospectraGF} corresponding to the energy norm.

\begin{figure}[htb!]
        \centering
        \begin{subfigure}[b]{0.49\linewidth}
            \centering
            \includegraphics[width=\linewidth]{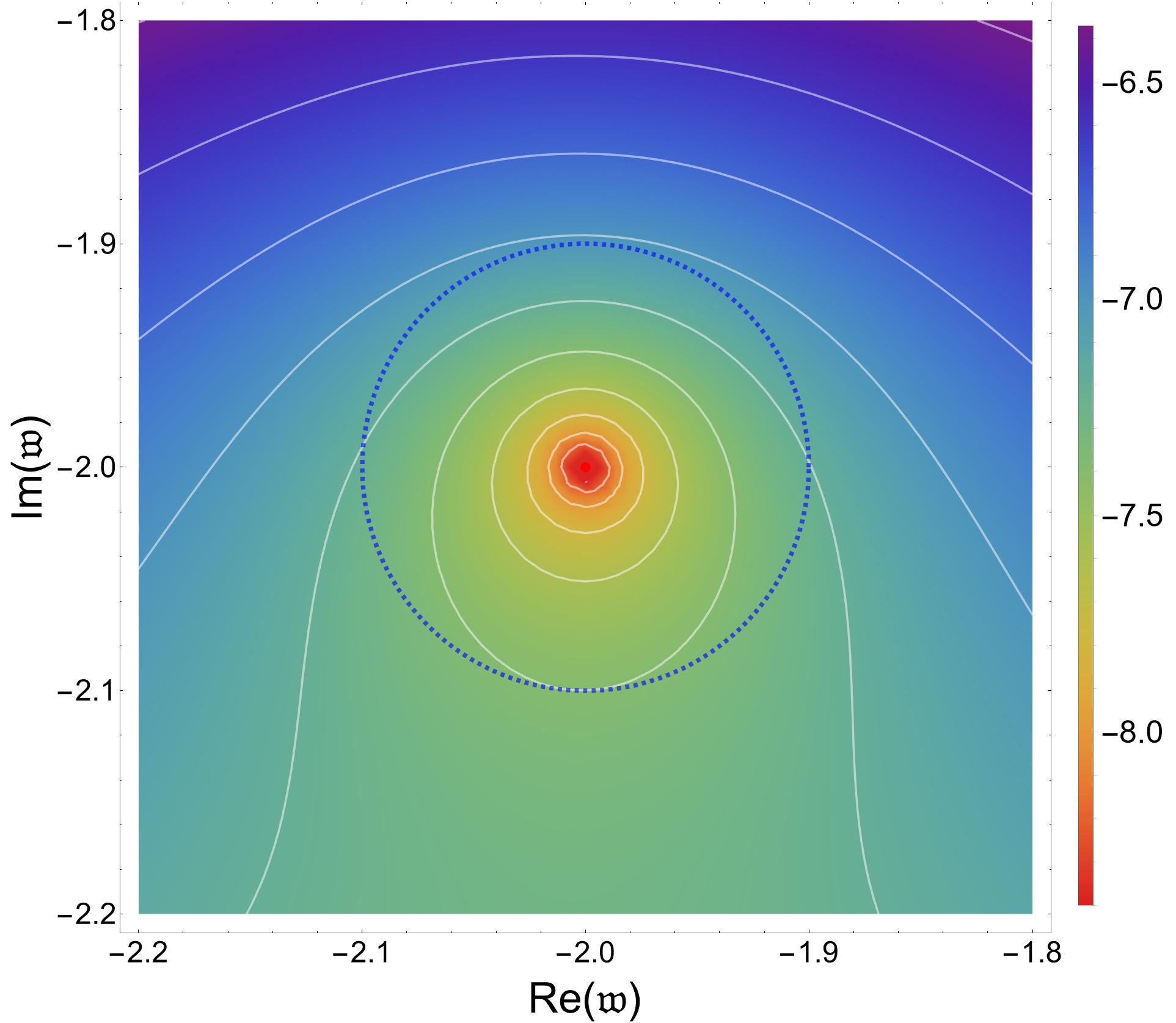}
            \captionsetup{justification=centering}
            \caption{$\mathfrak{q}=0$.}
            \label{fig:CloseupPseudoGFq0L2}
        \end{subfigure}\hfill
        \begin{subfigure}[b]{0.49\linewidth}
            \centering
            \includegraphics[width=\linewidth]{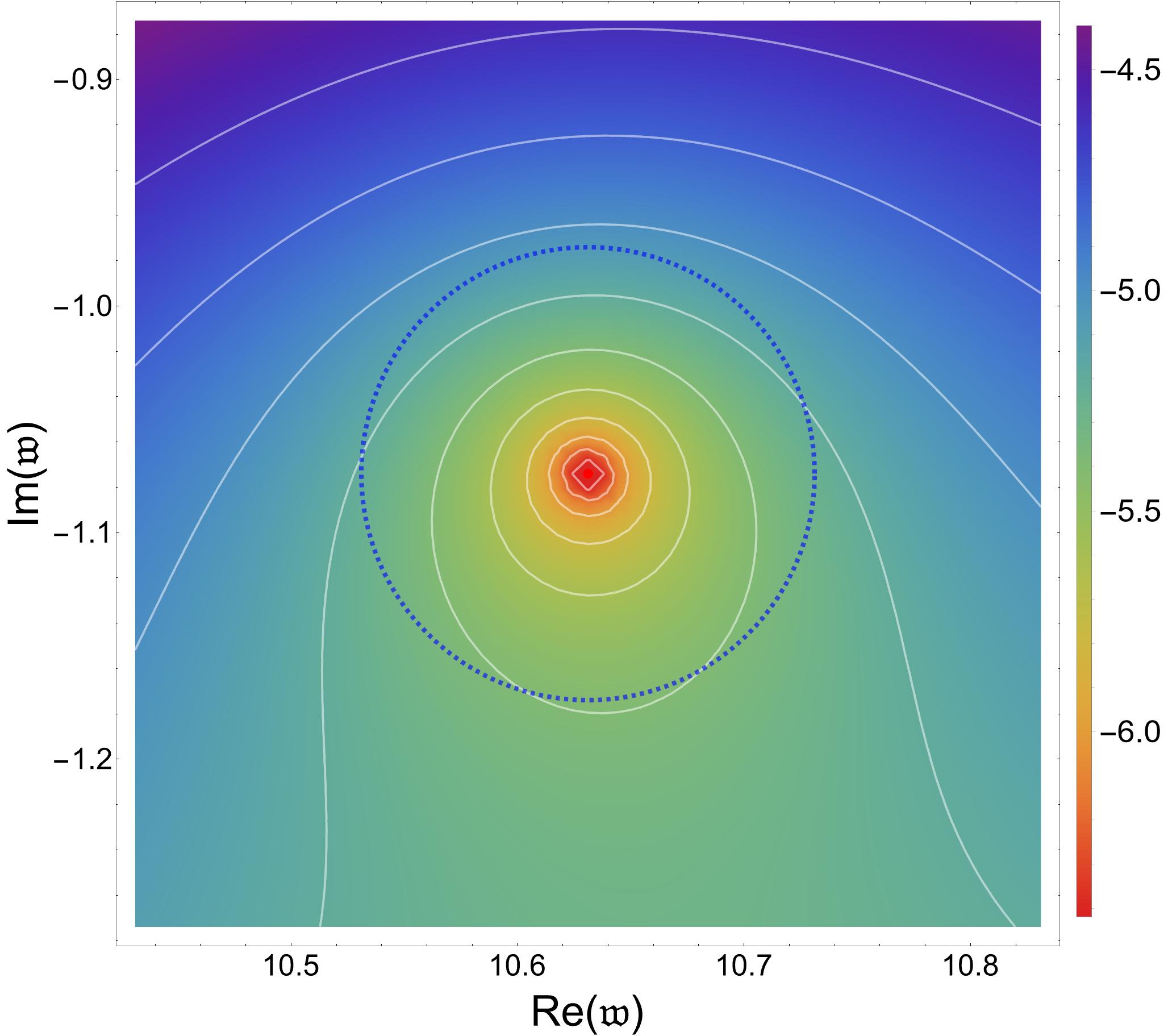}
            \captionsetup{justification=centering}
            \caption{$\mathfrak{q}=10$.}
            \label{fig:CloseupPseudoGFq10L2}
        \end{subfigure}
        \caption{Close-up of the transverse gauge field pseudospectrum in the $L^2$-norm around the first QNF for different values of $\mathfrak{q}$. The red dot corresponds to the QNF, the white lines represent the boundaries of various full $\varepsilon$-pseudospectra, and the dashed blue circle symbolizes a circle with a radius of $10^{-1}$ centered on the QNF. The heat map corresponds to the logarithm in base 10 of the inverse of the resolvent. Here the selective pseudospectra computed with random local potential perturbations of size $10^{-1}$ is hidden by the QNF.}
        \label{fig:CloseupPseudospectraGFL2}
    \end{figure}

 \begin{figure}[!htb]
        \centering
        \begin{subfigure}[b]{0.49\linewidth}
            \centering
            \includegraphics[width=\linewidth]{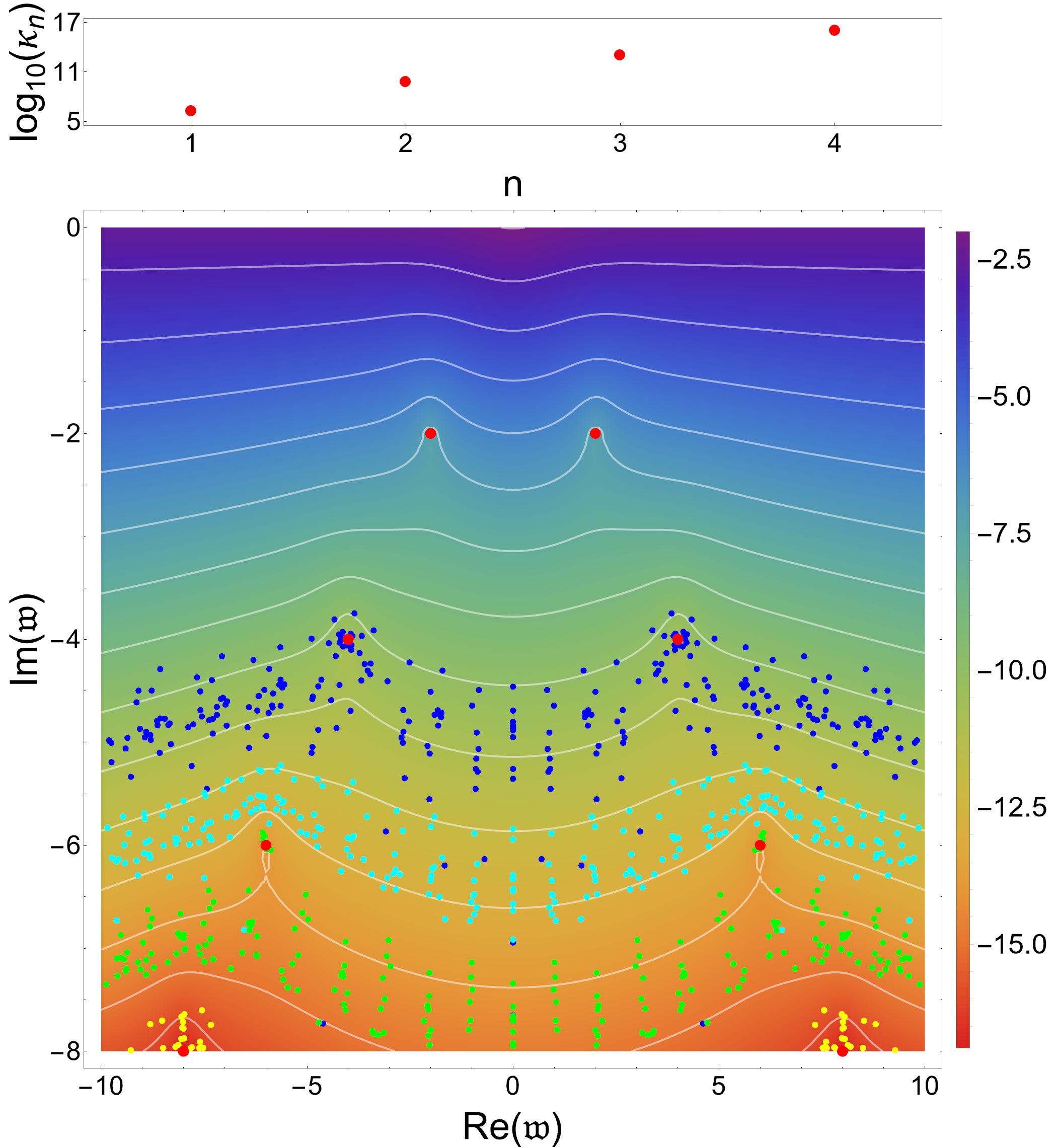}
            \captionsetup{justification=centering}
            \caption{$\mathfrak{q}=0$.}
            \label{fig:LargePseudoGFq0L2}
        \end{subfigure}\hfill
        \begin{subfigure}[b]{0.49\linewidth}
            \centering
            \includegraphics[width=\linewidth]{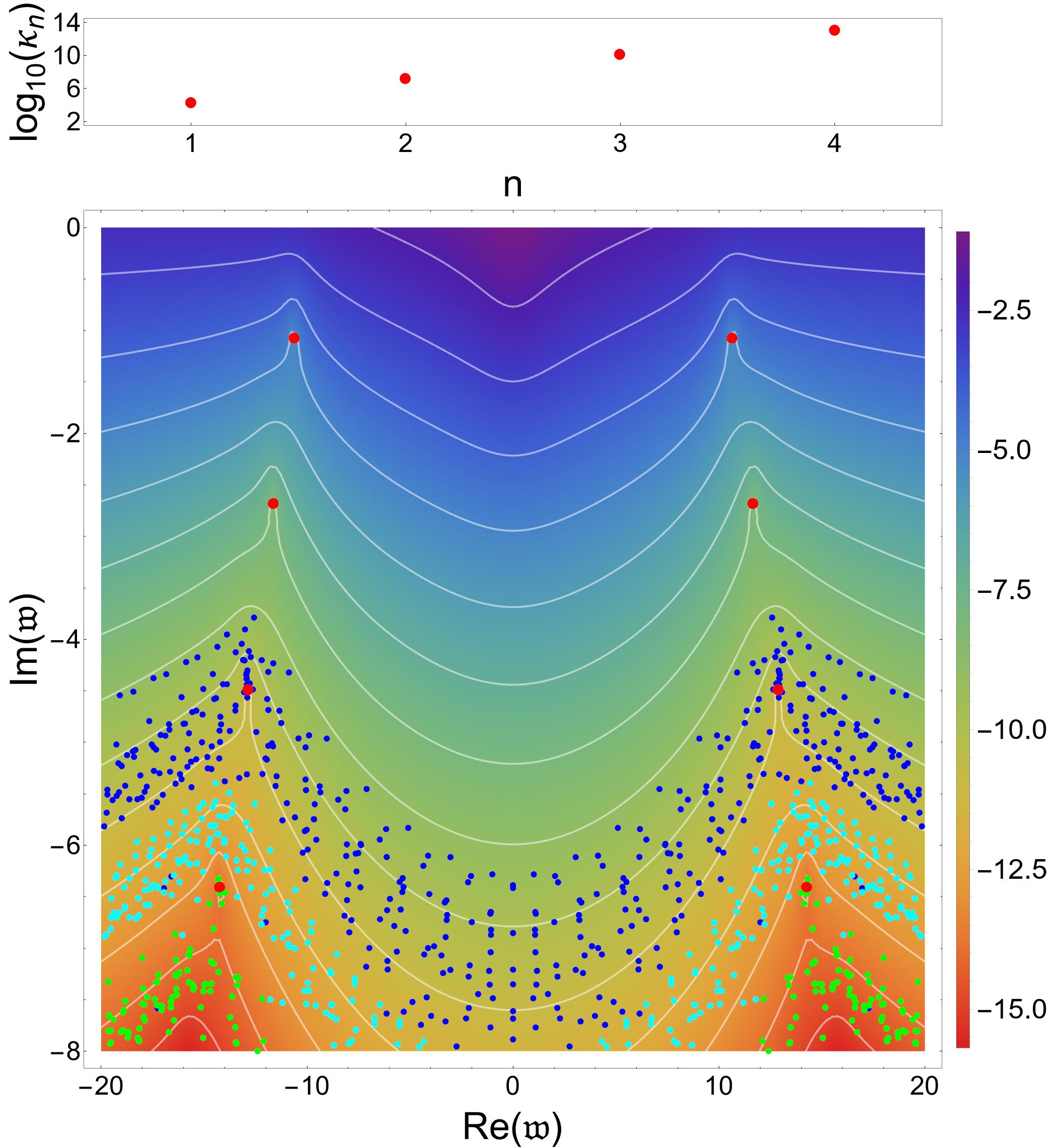}
            \captionsetup{justification=centering}
            \caption{$\mathfrak{q}=10$.}
            \label{fig:LargePseudoGFq10L2}
        \end{subfigure}
        \caption{Transverse gauge field pseudospectrum in the $L^2$-norm for different values of $\mathfrak{q}$. In the lower panels, we present selective and full pseudospectra. The red dots represent the QNFs, and the white lines denote the boundaries of different full $\varepsilon$-pseudospectra. The heat map corresponds to the logarithm in base 10 of the inverse of the resolvent, while the blue, cyan, green, and yellow dots indicate different selective $\varepsilon$-pseudospectra computed with random local potential perturbations of size $10^{-1}$, $10^{-3}$, $10^{-5}$, and $10^{-7}$; respectively. In the upper panels, we represent the condition numbers.}
        \label{fig:LargePseudospectraGFL2}
	\end{figure}

We end this section by pointing out that while qualitatively similar, the pseudospectra in the 
energy and $L^2$-norm are quantitatively different.
Indeed, although the shape of the contour map for the same model is very similar in both norms
the quantitative value of the pseudospectrum as illustrated by the color code varies markedly from the energy to the $L^2$-norm.
This observation further stresses the importance of properly defining a physically motivated norm: the discrepancy in the definition of size of the potential perturbations between the $L^2$ and energy norms results in quantitatively different phenomenology.

\clearpage

\section{Numerical Values of the QNFs}\label{app: Numerical Values of the QNFs}
In this appendix we provide the numerical values of the first 10 quasinormal frequencies for the models discussed in the main text. For purposes of presentation we limit the precision to 15 significant figures. 

\begin{table}[!htb]
    \centering
    \begin{tabular}{|c|c|c|}
        \hline
         $n$&$\Re(\mathfrak{w}_n)$&$\Im(\mathfrak{w}_n)$  \\
         \hline
         \hline
         $1$& $\pm3.11945159116590$ & $-2.74667574434061$  \\
         \hline
         $2$& $\pm5.16952097576223$ & $-4.76357011692046$ \\
         \hline
         $3$& $\pm7.18793082599961$ & $-6.76956490907620$ \\
         \hline
         $4$& $\pm9.19719929048412$ & $-8.77248110302029$ \\
         \hline
         $5$& $\pm11.2026757899702$ & $-10.7741623965097$ \\
         \hline
         $6$& $\pm13.2062465176047$ & $-12.7752389828404$ \\
         \hline
         $7$& $\pm15.2087360641333$ & $-14.7759793025747$ \\
         \hline
         $8$& $\pm17.2105584199107$ & $-16.7765153586509$ \\
         \hline
         $9$& $\pm19.2119427139029$ & $-18.7769189939199$ \\
         \hline
         $10$& $\pm21.2130253291193$ & $-20.7772323845860$ \\
         \hline
    \end{tabular}
    \caption{Real scalar QNFs for $m^2l^2=0$ and $\mathfrak{q}=0$.}
\end{table}

\hfill\break

\hfill\break

\begin{table}[!htb]
    \centering
    \begin{tabular}{|c|c|c|}
        \hline
         $n$&$\Re(\mathfrak{w}_n)$&$\Im(\mathfrak{w}_n)$  \\
         \hline
         \hline
         $1$& $\pm11.0586349863470$ & $-1.75038566826434$ \\
         \hline
         $2$& $\pm12.2041502935249$ & $-3.49827856173063$ \\
         \hline
         $3$& $\pm13.5234927685557$ & $-5.37996997551963$ \\
         \hline
         $4$& $\pm14.9758148001587$ & $-7.33629305380262$ \\
         \hline
         $5$& $\pm16.5292996000810$ & $-9.33351552595414$ \\
         \hline
         $6$& $\pm18.1599632430830$ & $-11.3527046127907$ \\
         \hline
         $7$& $\pm19.8500241826712$ & $-13.3832519413241$ \\
         \hline
         $8$& $\pm21.5863602228304$ & $-15.4192020702742$ \\
         \hline
         $9$& $\pm23.3592481371511$ & $-17.4572088679779$ \\
         \hline
         $10$& $\pm25.1614156884878$ & $-19.4954019803536$ \\
         \hline
    \end{tabular}
    \caption{Real scalar QNFs for $m^2l^2=0$ and $\mathfrak{q}=10$.}
\end{table}

\clearpage

\hfill\break

\begin{table}[!htb]
    \centering
    \begin{tabular}{|c|c|c|}
        \hline
         $n$&$\Re(\mathfrak{w}_n)$&$\Im(\mathfrak{w}_n)$  \\
         \hline
         \hline
         $1$& $\pm2.19881456585250$ & $-1.75953462713300$ \\
         \hline
         $2$& $\pm4.21189720328773$ & $-3.77488823578666$ \\
         \hline
         $3$& $\pm6.21554314901884$ & $-5.77725701316514$ \\
         \hline
         $4$& $\pm8.21716723825394$ & $-7.77808021954300$ \\
         \hline
         $5$& $\pm10.2180612360290$ & $-9.77847388400658$ \\
         \hline
         $6$& $\pm12.2186177048135$ & $-11.7786974714808$ \\
         \hline
         $7$& $\pm14.2189931785234$ & $-13.7788389290514$ \\
         \hline
         $8$& $\pm16.2192614139935$ & $-15.7789352848867$ \\
         \hline
         $9$& $\pm18.2194613734973$ & $-17.7790045343160$ \\
         \hline
         $10$& $\pm20.2196154334378$ & $-19.7790563669305$ \\
         \hline
    \end{tabular}
    \caption{Real scalar QNFs for $m^2l^2=-3$ and $\mathfrak{q}=0$.}
\end{table}

\hfill\break

\hfill\break

\begin{table}[!htb]
    \centering
    \begin{tabular}{|c|c|c|}
        \hline
         $n$&$\Re(\mathfrak{w}_n)$&$\Im(\mathfrak{w}_n)$  \\
         \hline
         \hline
         $1$& $\pm10.6370087422563$ & $-1.06356241400480$ \\
         \hline
         $2$& $\pm11.6542293745917$ & $-2.66476173979828$ \\
         \hline
         $3$& $\pm12.8816366495948$ & $-4.46855950305745$ \\
         \hline
         $4$& $\pm14.2652641255372$ & $-6.38140408881256$ \\
         \hline
         $5$& $\pm15.7667370212845$ & $-8.35383981758053$ \\
         \hline
         $6$& $\pm17.3576290472434$ & $-10.3587501137841$ \\
         \hline
         $7$& $\pm19.0169202769993$ & $-12.3810472715686$ \\
         \hline
         $8$& $\pm20.7291315800900$ & $-14.4122658126651$ \\
         \hline
         $9$& $\pm22.4828389281553$ & $-16.4476266560906$ \\
         \hline
         $10$& $\pm24.2695453187766$ & $-18.4844261072115$ \\
         \hline
    \end{tabular}
    \caption{Real scalar QNFs for $m^2l^2=-3$ and $\mathfrak{q}=10$.}
\end{table}

\clearpage

\hfill\break

\begin{table}[!htb]
    \centering
    \begin{tabular}{|c|c|c|}
        \hline
         $n$&$\Re(\mathfrak{w}_n)$&$\Im(\mathfrak{w}_n)$  \\
         \hline
         \hline
         $1$& $\pm2.00000000000000$ & $-2.00000000000000$ \\
         \hline
         $2$& $\pm4.00000000000000$ & $-4.00000000000000$ \\
         \hline
         $3$& $\pm6.00000000000000$ & $-6.00000000000000$ \\
         \hline
         $4$& $\pm8.00000000000000$ & $-8.00000000000000$ \\
         \hline
         $5$& $\pm10.0000000000000$ & $-10.0000000000000$ \\
         \hline
         $6$& $\pm12.0000000000000$ & $-12.0000000000000$ \\
         \hline
         $7$& $\pm14.0000000000000$ & $-14.0000000000000$ \\
         \hline
         $8$& $\pm16.0000000000000$ & $-16.0000000000000$ \\
         \hline
         $9$& $\pm18.0000000000000$ & $-18.0000000000000$ \\
         \hline
         $10$& $\pm20.0000000000000$ & $-20.0000000000000$ \\
         \hline
    \end{tabular}
    \caption{Transverse gauge field QNFs for $\mathfrak{q}=0$.}
\end{table}

\hfill\break

\hfill\break

\begin{table}[!htb]
    \centering
    \begin{tabular}{|c|c|c|}
        \hline
         $n$&$\Re(\mathfrak{w}_n)$&$\Im(\mathfrak{w}_n)$  \\
         \hline
         \hline
         $1$& $\pm10.6311305880202$ & $-1.07401086341959$ \\
         \hline
         $2$& $\pm11.6456703226270$ & $-2.68058057268396$ \\
         \hline
         $3$& $\pm12.8710360299982$ & $-4.48891090932708$ \\
         \hline
         $4$& $\pm14.2530488610374$ & $-6.40563035311599$ \\
         \hline
         $5$& $\pm15.7531754716734$ & $-8.38144582889587$ \\
         \hline
         $6$& $\pm17.3428934333756$ & $-10.3893551546947$ \\
         \hline
         $7$& $\pm19.0011254785217$ & $-12.4143536422258$ \\
         \hline
         $8$& $\pm20.7123581254599$ & $-14.4480366538192$ \\
         \hline
         $9$& $\pm22.4651463361787$ & $-16.4856700172425$ \\
         \hline
         $10$& $\pm24.2509798299875$ & $-18.5245835573417$ \\
         \hline
    \end{tabular}
    \caption{Transverse gauge field QNFs for $\mathfrak{q}=10$.}
\end{table}

\clearpage

\section{Pseudospectrum Algorithm}

   \label{app:algorithm}
    In this appendix we give a more explicit description of the algorithm employed when computing the full pseudospectrum and condition numbers using theorem  \ref{th:psuedospectrum in different norms}. For the sake of clarity, we do not present here the code but instead the following illustrative flowchart:
    \hfill\break
    \begin{figure}[H]
        \centering
        \includegraphics[width=0.89\linewidth]{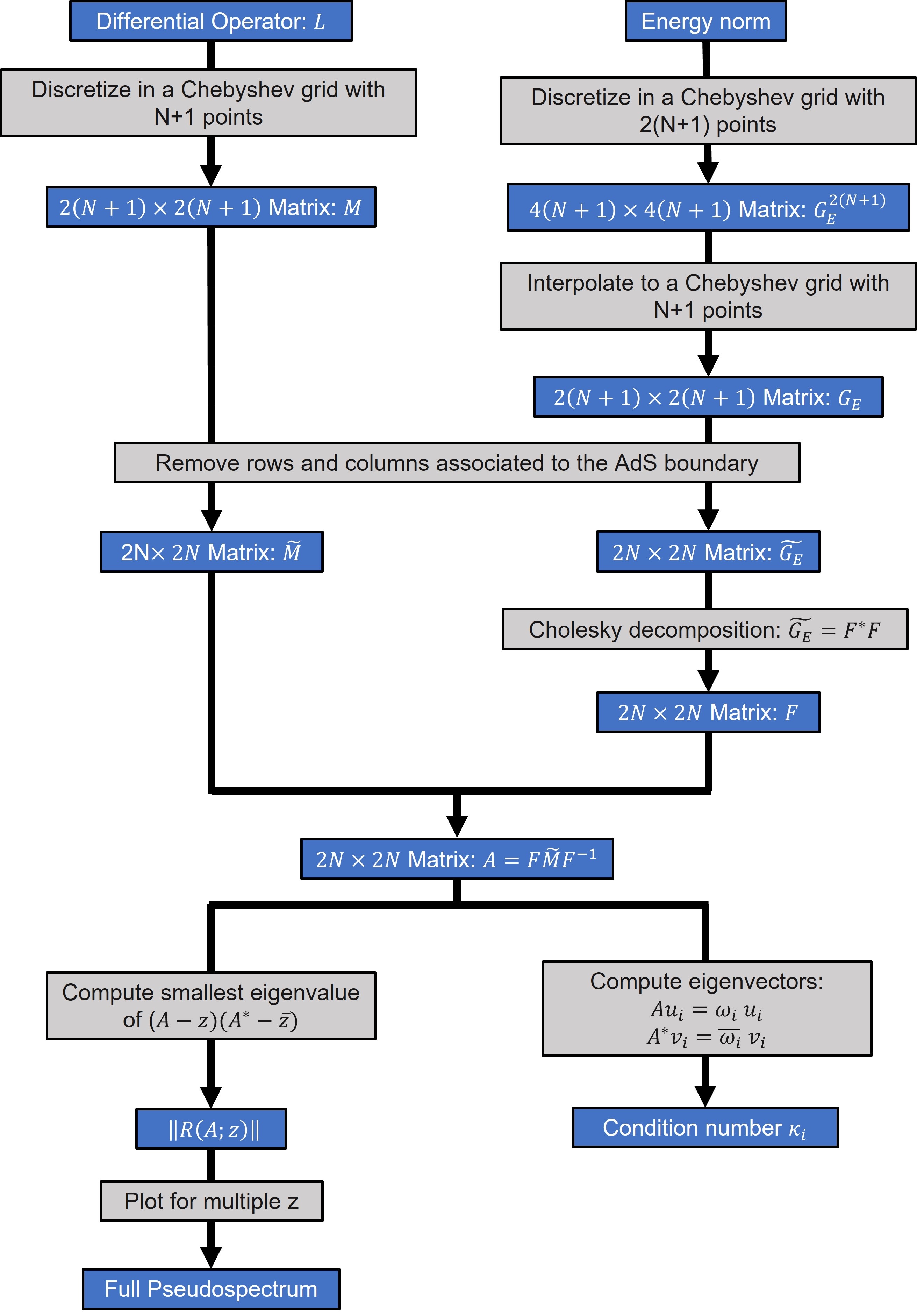}
    \end{figure}
\clearpage

\bibliographystyle{JHEP}
\bibliography{biblio}

\end{document}